\documentclass[final,leqno,showlabe]{siamltex}
\usepackage{amsmath}
\usepackage{amssymb}
\usepackage{cases}
\usepackage{enumerate}

\usepackage{epsfig}

\input psfig.sty

\newcommand{\Og}{\Omega}

\newcommand{\nn}{\nonumber}

\newcommand{\gm}{\gamma}

\newcommand{\be}{\begin{equation}}
\newcommand{\ee}{\end{equation}}
\newcommand{\ba}{\begin{array}}
\newcommand{\ea}{\end{array}}
\newcommand{\bea}{\begin{eqnarray}}
\newcommand{\eea}{\end{eqnarray}}
\newcommand{\beas}{\begin{eqnarray*}}
\newcommand{\eeas}{\end{eqnarray*}}

\newtheorem{remark}{Remark}[section]

 \newcommand{\bx}{{\bf x} }
 \newcommand{\by}{{\bf y} }
 \DeclareMathOperator*{\argmin}{arg\,min}

\renewcommand{\ldots}{\dotsc}

\title{Fundamental gaps of the Gross-Pitaevskii equation with repulsive interaction\thanks{This work was supported by the Ministry of Education of Singapore grant R-146-000-223-112.}}
\author{Weizhu Bao\thanks{Department of Mathematics, National
University of Singapore, Singapore 119076 ({\it
matbaowz@nus.edu.sg}, URL: http://www.math.nus.edu.sg/\~{}bao/).
}
\and Xinran Ruan\thanks{Department of Mathematics, National
University of Singapore, Singapore 119076 ({\it a0103426@u.nus.edu}).} }

\date{}
\begin{document}

\maketitle

\begin{abstract}
We study asymptotically and numerically the fundamental gaps (i.e. the difference between
the first excited state and the ground state) in energy and chemical
potential of the  Gross-Pitaevskii equation (GPE) -- nonlinear Schr\"{o}dinger
equation  with cubic nonlinearity -- with repulsive interaction under different
trapping potentials including box potential and harmonic potential.
Based on our asymptotic and numerical results, we formulate a gap conjecture
on the fundamental gaps in energy and chemical potential of the GPE on bounded
domains with the homogeneous Dirichlet boundary condition, and in the whole space
with a convex trapping potential growing at least quadratically in the far field.
We then extend these results to the GPE on bounded
domains with either the homogeneous Neumann boundary condition
or periodic boundary condition.
\end{abstract}

\begin{keywords} Gross-Pitaevskii equation, fundamental gap,
 ground state, first excited state, energy asymptotics, repulsive interaction
\end{keywords}

\begin{AMS} 35B40, 35P30, 35Q55, 65N25.
\end{AMS}

\pagestyle{myheadings} \markboth{Weizhu Bao and Xinran Ruan}
{Fundamental gaps of the GPE}

\section{Introduction}
\label{s1} \setcounter{equation}{0}
The time-independent Schr\"{o}dinger equation (in dimensionless form by taking $\hbar=m=1$ with
$m$ the mass of the particle) \cite{Sch,Atk,Hook,Quan_chem}
\begin{equation}\label{eq:N_body}
H\Phi: =\biggl[\sum_{j=1}^N\left(-\frac{1}{2}\Delta_{j}+V({\bf r}_j)\right)
+\sum_{1\le j<k\le N}V_{\rm{int}}({\bf r}_j-{\bf r}_k)\biggr]\Phi=E\,\Phi,
\end{equation}
has been widely used in quantum physics and chemistry to
mathematically predict the property of a quantum system with $N$ particles (usually atoms,
molecules, and subatomic particles whether free, bound or localized). 
Here ${\bf r}_j\in{\mathbb R}^3$  is the
spatial coordinate of the $j$-th particle, $\Delta_{j}$ is the Laplacian operator with
respect to the spatial coordinate ${\bf r}_j$ for
$j=1,2,\ldots,N$, $\Phi:=\Phi({\bf r}_1,\ldots,{\bf r}_N)$
is the complex-valued wave function of the quantum system,
$V({\bf r})$ (for ${\bf r}\in{\mathbb R}^3$) is a given real-valued potential,
$V_{\rm{int}}({\bf r})$ is a given real-valued interaction kernel for two-body interaction satisfying
$V_{\rm{int}}({\bf r})=V_{\rm{int}}(-{\bf r})$
and $H$ is the Hamiltonian operator. When the wave function
is normalized as $\int_{{\mathbb R}^{3N}} |\Phi|^2d{\bf r}_1\ldots d{\bf r}_N=1$,
$E$ is the total energy of the quantum system with respect to the wave function
$\Phi$. The time-independent Schr\"{o}dinger equation
(\ref{eq:N_body}), also an eigenvalue problem in mathematics, predicts that wave function can form
stationary states including ground and excited states
\cite{Sch,Atk,Hook,Quan_chem}.
Finding the ground state and its energy, as well as
the energy gap (or band gap)
between the ground and first excited states via Eq. (\ref{eq:N_body})
has become a fundamental and  highly challenging problem in computational
quantum physics and chemistry, as well as material simulation and design.

By setting $N=1$ in Eq. (\ref{eq:N_body}) and performing
a dimension reduction from three dimensions (3D)
to two dimensions (2D) and one dimension (1D) under proper assumptions on the
potential $V({\bf r})$ such that
separation of (the spatial) variables for the wave function is valid
\cite{Bao,Dalfovo1,Pitaevskii,BJP,Bao2013},
one can get the $d$-dimensional ($d=3,2,1$) time-independent Schr\"{o}dinger equation with complex-valued
wave function $\phi:=\phi(\bx)$, which has been widely used in
the physics literature \cite{Sch,Atk,Hook,Quan_chem}
\begin{equation}\label{eq:1_body}
H\phi: =\left(-\frac{1}{2}\Delta_\bx+V({\bf x})\right)\phi({\bf x})=E\,\phi({\bf x}), \qquad {\bf x}=(x_1,\ldots,x_d)^T\in\Omega\subseteq {\mathbb R}^d.
\end{equation}
Using the rescaling formulas ${\bf y}=\sqrt{2}\,{\bf x}\in {\mathbb R}^d$ 
and $\varphi({\bf y})
=2^{-d/4}\phi({\bf x})$,
one can derive the following
time-independent Schr\"{o}dinger equation 
\cite{gap,intro_gap,gap2,gap0}
\be\label{eigen}
L\varphi:=\left[-\Delta_\by+W(\mathbf{y})\right]\varphi(\by)=E\,\varphi(\by), \qquad \by\in U\subseteq {\mathbb R}^d,
\ee
where  $W(\by)=V({\bf x})=V(\by/\sqrt{2})$, $U=\{\mathbf{y}\,|\,\mathbf{y}/\sqrt{2}\in \Og\}$
and the operator $L:=-\Delta_\by+W(\mathbf{y})$ is called
the Schr\"{o}dinger operator \cite{intro_gap}. If $U$ is bounded, then we require the homogeneous
Dirichlet boundary condition (BC) $\varphi(\by)|_{\partial U}=0$ to be imposed.
In this case, we can also simply define $W(\by)=+\infty$
for $\by$ outside $U$ and Eq. (\ref{eigen}) can be defined in the whole
space without BC.
If $W(\by)$ is bounded below in $U$, i.e. $\inf_{\by\in U} W(\by)>-\infty$,
without loss of generality, we can always assume that $W(\by)\ge0$ for $\by\in U$
when we are interested in the ground and excited states and the energy gap.
Under proper assumptions on the potential $W(\by)$, the eigenvalues $E_g,\, E_1,\, E_2,\,\cdots$ of
the Sturm-Liouville eigenvalue problem  (\ref{eigen}) under the normalization condition \cite{intro_gap}
\be
\label{norm}
\|\varphi\|_2^2:=\int_{U} |\varphi(\by)|^2d\by=1 \quad \Longleftrightarrow \quad
\|\phi\|_2^2:=\int_{\Og} |\phi(\bx)|^2d\bx=1,
\ee
are real and can be ordered such that $0<E_g<E_1\le E_2\le \cdots$ with corresponding
eigenfunctions (or stationary states) $\varphi_g(\by),\, \varphi_1(\by),\, \varphi_2(\by), \, \cdots$.
Then $\varphi_g(\by)$ and $\varphi_1(\by)$ are called the ground state and the first excited state,
respectively. $\delta_0:=E_1-E_g>0$ is usually called the {\bf fundamental gap}
in the literature \cite{intro_gap,gap2,gap0}.
Assuming that $U$ is  a bounded convex  domain and the potential $W(\by)\in C(U)$,  based on
results for special cases, the {\bf gap conjecture} was
formulated in the literature \cite{intro_gap,gap2,gap0} as:
\be\label{gap679}
\delta_0=E_1-E_g \ge \frac{3\pi^2}{D_{_U}^2}, \qquad \hbox{with}\quad D_{U}:=\sup_{\mathbf{y},\mathbf{z}\in U}|\mathbf{y}-\mathbf{z}|.
\ee
The gap conjecture is sharp when $d=1$, $U=(0,L)$ with $0<L\in {\mathbb R}$
and $W(y)\equiv 0$ for $y\in U$ \cite{intro_gap}. Recently, by the use of the gradient flow and geometric
analysis and assuming that $W(\mathbf{y})\in C(U)$ is convex,
Andrews and Clutterbuck proved the gap conjecture \cite{gap}.
In addition, they showed that if $U={\mathbb R}^d$ and the potential $W(\by)$ satisfies
$D^2W(\by)\ge\gm_w^2I_d$ for $\by\in {\mathbb R}^d$ with $\gm_w>0$, where $I_d$ is the identity matrix in
$d$-dimensions, the  fundamental gap described by Eq. (\ref{eigen}) under the condition (\ref{norm}) satisfies
$\delta_0:=E_1-E_g\ge\sqrt{2}\gm_w$ \cite{gap}.

In this paper, we will consider the dimensionless
time-independent Gross-Pitaevskii equation (GPE)
in $d$-dimensions ($d=3,2,1$) \cite{Bao,Dalfovo1,Pitaevskii,BJP,Bao2013}
\begin{equation}\label{eq:eig}
\left[-\frac{1}{2}\Delta+V(\mathbf{x})+\beta|\phi(\bx)|^{2}\right]\phi(\bx)=
\mu\phi(\mathbf{x}),\qquad \mathbf{x}\in\Omega\subseteq \mathbb{R}^d,
\end{equation}
where $\phi:=\phi(\bx)$ is the complex-valued wave function (or eigenfunction) normalized via (\ref{norm}),
$V:=V(\bx)$ is a given real-valued potential,
$\beta\ge0$ is a dimensionless constant describing
the repulsive (defocussing) interaction strength,
and the eigenvalue (or chemical potential in the physics literature)
$\mu:=\mu(\phi)$ is defined as \cite{Bao,Dalfovo1,Pitaevskii,Bao2013}
\begin{equation}\label{def:mu}
\mu(\phi)=E(\phi)+\frac{\beta}{2}\int_{\Omega}|\phi(\bx)|^{4}d\mathbf{x},
\end{equation}
with the energy $E:=E(\phi)$ defined as \cite{Bao2013}
\begin{equation}\label{def:E}
E(\phi)=\int_{\Omega}\left[\frac{1}{2}|\nabla\phi(\bx)|^2+V(\mathbf{x})|\phi(\bx)|^2+
\frac{\beta}{2}|\phi(\bx)|^{4}\right]d\mathbf{x}.
\end{equation}
Again, if $\Omega$ is bounded, the homogeneous
Dirichlet BC, i.e. $\phi(\bx)|_{\partial \Omega}=0$, needs to be imposed.
Thus, the time-independent GPE (\ref{eq:eig}) is a nonlinear eigenvalue problem under
the constraint (\ref{norm}).
It is a mean field model arising from  Bose-Einstein condensates
(BECs) \cite{Anderson,Dalfovo1,GPE_BEC1,Bao} that can be obtained from
the Schr\"{o}dinger equation (\ref{eq:N_body}) via the Hartree ansatz and mean field approximation \cite{Bao2013,MFT,TGmath,Pitaevskii}.
When $\beta=0$, it collapses to the time-independent Schr\"{o}dinger equation
(\ref{eq:1_body}).
It is worth mentioning that if the domain $U$ in (\ref{eigen}) is bounded,
the domain $\Og$ in (\ref{eq:1_body}) (or (\ref{eq:eig})) can be
defined as $\Og=\{\mathbf{x}\,|\,\sqrt{2}\mathbf{x}\in U\}$ via the re-scaling ${\bf y}=\sqrt{2}\,{\bf x}$.
Thus the diameter for the domain $\Og$ becomes $D:=D_{_\Og}=D_{_U}/\sqrt{2}$,
and the lower bound in the fundamental gap (\ref{gap679}) for the Schr\"{o}dinger equation (\ref{eq:1_body})
becomes $\frac{3\pi^2}{D_{_U}^2}=\frac{3\pi^2}{2D^2}$.

The ground state of the GPE (\ref{eq:eig}) is usually defined as the minimizer
of the non-convex minimization problem (or constrained minimization problem) \cite{Bao,Dalfovo1,GPE_BEC1,Bao2013}
\begin{equation}\label{def:ground}
\phi_g=\argmin_{\phi\in{S}}E(\phi),
\end{equation}
where $S=\{\phi\, |\, \|\phi\|_2^2:=\int_{\Omega} |\phi(\bx)|^2d\bx =1,\ E(\phi)<\infty,
\ \phi|_{\partial\Omega}=0 \ \hbox{if} \ \Omega
\ \hbox{is bounded}\}$.
The ground state can be chosen as nonnegative
$|\phi_g|$, i.e. $\phi_g=|\phi_g| e^{i\theta}$ for some constant $\theta\in {\mathbb R}$
and $i=\sqrt{-1}\,$.
Moreover, the nonnegative ground state $|\phi_g|$ is unique \cite{Lieb,Bao2013}.
Thus, from now on, we refer to the ground state as the nonnegative one.
It is easy to see that the ground state $\phi_g$ satisfies the
time-independent GPE (\ref{eq:eig}) and the constraint (\ref{norm}). Hence
it is an eigenfunction (or stationary state) of (\ref{eq:eig}) with the least energy.
Any other eigenfunctions of the GPE (\ref{eq:eig}) under the constraint (\ref{norm})
whose energies are larger than that of the ground state are usually called the excited states
in the physics literature \cite{Dalfovo1,Pitaevskii,Bao2013}. Specifically,
the excited state with the least energy among all excited states is usually
called the first excited state, which is denoted as $\phi_1$.

For the GPE (\ref{eq:eig}),
the ground state has been obtained asymptotically in weakly
and strongly repulsive interaction regimes, i.e. $0\le \beta\ll1$ and  $\beta\gg1$, respectively,
for several different trapping potentials \cite{LZBao}. 
In fact, by ordering all the eigenfunctions
of the GPE with a repulsive interaction and a confinement potential,
i.e. $\beta\ge0$ and $\lim_{|\bx|\to+\infty}V(\mathbf{x})=+\infty$,
under the constraint (\ref{norm}) according to their energies with
$\phi_g^{\beta},\, \phi_1^{\beta},\,\phi_2^{\beta},\ldots$ satisfying $E_g(\beta):=E(\phi_g^{\beta})<E_1(\beta):=E(\phi_1^{\beta})\le{E}(\phi_2^{\beta})\le\ldots$,
it can be shown that $\mu_g(\beta):=\mu(\phi_g^{\beta})<\mu_1(\beta):=\mu(\phi_1^{\beta})$
\cite{Can1}, and thus
$\phi_1^\beta$ is usually called the first excited state.
We define the {\bf fundamental gaps} in energy and chemical potential of the
time-independent GPE
under the constraint (\ref{norm}) as 
\be \label{gapp21}
\begin{split}
&\delta_E(\beta):={E}(\phi_1^{\beta})-E(\phi_g^{\beta})>0,
\quad\delta_\mu(\beta):=\mu(\phi_1^{\beta})-\mu(\phi_g^{\beta})>0, \quad \beta\ge0,\qquad \\
&\delta_E^\infty:=\inf_{\beta\ge0} \delta_E(\beta), \qquad
\delta_\mu^\infty:=\inf_{\beta\ge0} \delta_\mu(\beta).
\end{split}
\ee
In general, the first excited state $\phi_1^\beta$ is not unique.
Since we are mainly interested in its energy and chemical potential as well
as the fundamental gaps, it does not matter which first excited state is taken in our analysis and simulation below.
The main purpose of this paper is to study asymptotically and numerically
the fundamental gaps of the GPE with different trapping potentials and
to formulate a gap conjecture for the GPE.




The rest of this paper is organized as follows. In Section \ref{box}, we  study asymptotically and numerically the fundamental gaps of GPE on bounded domains
with homogeneous Dirichlet BC. In Section \ref{har}, we obtain
results for GPE in the whole space with a
confinement potential. Extension to GPE on bounded domains with either
periodic or homogeneous Neumann BC are presented
in Section \ref{asym:neumann_periodic}.  Finally, some conclusions are drawn in Section \ref{conclusion}. In  order to distinguish two different cases, i.e. nondegenerate and degenerate cases, we define the eigenspace of
(\ref{eq:1_body}) corresponding to the eigenvalue $E_1$ (the second smallest eigenvalue) as
\be\label{def:eigspace}
W_1=\{ \phi(\bx): \ \Omega \to {\mathbb C} \,|\, H\phi=E_1\phi, \ \phi|_{\partial\Omega}=0 \ \hbox{if} \ \Omega
\ \hbox{is bounded}\}.
\ee
Then the dimension of $W_1$, i.e.  $\dim(W_1)=1$,
is referred to the nondegenerate case, and resp.,
$\dim(W_1)\ge2$ is referred to the degenerate case.
Denote $\Omega_0=\prod_{j=1}^d(0,L_j)$ satisfying $L_1\ge L_2 \ge \dots \ge L_d>0$ and $A_0=1/\sqrt{\prod_{j=1}^dL_j}$\;.

\section{On bounded domains with homogeneous Dirichlet BC}\label{box}
\setcounter{equation}{0}
\setcounter{figure}{0}
In this section, we obtain fundamental gaps of
the GPE \eqref{eq:eig} on a bounded domain $\Omega$ with homogeneous Dirichlet BC
asymptotically under a box potential and numerically under a general potential. Based on the results, we formulate a novel gap conjecture.

\subsection{Nondegenerate case, i.e. $\dim(W_1)=1$}\label{sec:box_asym}
We first consider a special case by taking  $\Omega=\Omega_0$ satisfying $d=1$ or $L_1> L_2$ when $d\ge2$ and $V(\bx)\equiv0$ for
$\bx\in \Omega$ in \eqref{eq:eig}.
For simplicity, we define
\be\label{boxct11} A_1=\frac{2}{L_1}\left(\frac{25}{9L_1}+\frac{2}{9}\sum_{j=1}^d\frac{1}{L_j}\right),\quad A_2=\frac{\pi^2}{2}\sum_{j=1}^d\frac{1}{L_j^2}.
\ee

In this scenario, when $\beta=0$, all the eigenfunctions can be
obtained via the sine series \cite{BaoL,LZBao}. Thus the ground state $\phi_g^0(\bx)$ and the first excited state
$\phi_1^0(\bx)$ can be given explicitly as \cite{BaoL,LZBao} for $\bx\in\bar{\Omega}$
\be\label{box00}
\phi_g^0(\mathbf{x})=2^{\frac{d}{2}}A_0\prod_{j=1}^d\sin\left(\frac{\pi x_j}{L_j}\right),
\quad  \phi_1^{0}(\mathbf{x})=2^{\frac{d}{2}}A_0\sin\left(\frac{2\pi x_1}{L_1}\right)
\prod_{j=2}^d\sin\left(\frac{\pi x_j}{L_j}\right).
\ee

\begin{lemma}\label{box:ground_weak}
In the weakly repulsive interaction regime, i.e. $0<\beta\ll1$, we have
\begin{align} \label{boxsmall}
E_g(\beta)&=A_2+\frac{3^dA_0^2}{2^{d+1}}\beta+o(\beta),\quad
\mu_g(\beta)=A_2+\frac{3^dA_0^2}{2^{d}}\beta+o(\beta),\quad 0\le \beta\ll1,\\
\label{boxsmall1}
E_1(\beta)&=\frac{3\pi^2}{2L_1^2}+A_2+\frac{3^dA_0^2}{2^{d+1}}\beta+o(\beta), \quad
\mu_1(\beta)=\frac{3\pi^2}{2L_1^2}+A_2+\frac{3^dA_0^2}{2^{d}}\beta+o(\beta).
\end{align}
\end{lemma}

\begin{proof}
When $0<\beta\ll1$, we can approximate the ground
state $\phi_g^\beta(\bx)$ and the first excited state $\phi_1^\beta(\bx)$ by $\phi_g^0(\bx)$ and
$\phi_1^0(\bx)$, respectively. Thus we have
\be\label{box01}
\phi_g^\beta(\mathbf{x})\approx \phi_g^0(\mathbf{x}),\qquad
\phi_1^{\beta}(\mathbf{x})\approx \phi_1^{0}(\mathbf{x}),
\qquad \bx\in\bar{\Omega}.
\ee
Plugging \eqref{box01} into \eqref{def:mu} and \eqref{def:E}, after a detailed computation which is omitted here for brevity, we can obtain \eqref{boxsmall}-\eqref{boxsmall1}.
\end{proof}

Lemma \ref{box:ground_weak} implies that $\delta_E(\beta)=E_1(\beta)-E_g(\beta)\approx \frac{3\pi^2}{2L_1^2}$
and $\delta_\mu(\beta)=\mu_1(\beta)-\mu_g(\beta)\approx \frac{3\pi^2}{2L_1^2}$ for $0\le \beta\ll1$,  which are independent of $\beta$.
In order to get the dependence on $\beta$,
 we need to find more accurate approximations of  $\phi_g^\beta$
and  $\phi_1^\beta$ and can obtain the following asymptotics of the fundamental gaps.

\begin{lemma}\label{asym:boxbt}
When $0\le \beta\ll1$, we have
\be \label{gapbox14}
\delta_E(\beta)=
\frac{3\pi^2}{2L_1^2}+G_d^{(1)}\beta^{2}+o(\beta^2),\qquad
\delta_{\mu}(\beta)=\frac{3\pi^2}{2L_1^2}+G_d^{(2)}\beta^{2}+o(\beta^2),
\ee
where
\[G_d^{(1)}=\left\{\ba{l}
\frac{3}{64\pi^2},\\
\frac{A_0^4}{64\pi^2}\left(\frac{27}{4}L_1^2+\frac{3}{A_6(A_6L_1^2+3)}\right),\\
\frac{1}{256\pi^2}(C_{1,1,1}-C_{2,1,1}),\\
\ea\right.  G_d^{(2)}=\left\{\ba{ll}
\frac{9}{64\pi^2}, &d=1,\\
\frac{3A_0^4}{64\pi^2}\left(\frac{27}{4}L_1^2+\frac{3}{A_6(A_6L_1^2+3)}\right), &d=2,\\
\frac{3}{256\pi^2}(C_{1,1,1}-C_{2,1,1}), &d=3,\\
\ea\right.
\]
with
\be
A_6=\sum_{j=1}^d\frac{1}{L_j^2},
\quad
C_{k_1,k_2,k_3}=A_0^4\left(81\sum_{j=1}^3\frac{L_j^2}{k_j^2}+9\sum_{i<j}
\frac{1}{\frac{k_i^2}{L_i^2}+\frac{k_j^2}{L_j^2}}+\frac{1}{\sum_{j=1}^3\frac{k_j^2}{L_j^2}}\right).
\ee
\end{lemma}

\begin{proof}
When $0< \beta\ll1$, 
we assume
\be\label{phig1x}
\phi_g^\beta(\bx)\approx \phi_g^0(\bx)+\beta\varphi_g(\bx)+o(\beta),
\quad \bx\in\Omega.
\ee
Plugging (\ref{phig1x}) into (\ref{eq:eig}), noticing (\ref{box00}),
(\ref{boxsmall}) and (\ref{boxsmall1}),
and dropping all terms at $O(\beta^2)$ and above, we obtain
\be \label{vph85}
\Delta \varphi_g(\bx)+2A_2\varphi_g(\bx)=
2(\phi_g^0(\bx))^3-\frac{3^dA_0^2}{2^{d-1}}\phi_g^0(\bx), \quad \bx\in\Omega,\qquad
\varphi_g(\bx)|_{\partial\Omega}=0.
\ee
Substituting (\ref{box00}) into (\ref{vph85}), we can solve it analytically.
For the simplicity of notations, here we only present the case
when $d=1$. Extensions to $d=2$ and $d=3$ are straightforward and the details are
omitted here for brevity \cite{Ruan}.
When $d=1$, we have
\be\label{vpg67}
\varphi_g(x)=\frac{\sqrt{2L_1}}{8\pi^2}\sin\left(\frac{3\pi{x}}{L_1}\right), \qquad 0\le x\le L_1.
\ee
Plugging \eqref{vpg67} into \eqref{phig1x} and using $\|\phi_g^\beta\|_2=1$, we get
\be\label{phig678}
\phi_g^\beta(x)\approx \sqrt{\frac{64\pi^4-\beta^2L_1^2}{32\pi^4L_1}}
\sin\left(\frac{\pi{x}}{L_1}\right)+\frac{\beta\sqrt{2L_1}}{8\pi^2}
\sin\left(\frac{3\pi{x}}{L_1}\right),\quad 0\le x\le L_1.
\ee
Inserting (\ref{phig678}) into (\ref{def:E}) and (\ref{def:mu}) with $V(x)\equiv0$,
we have
\be\label{ge634}
E_g(\beta)=\frac{\pi^2}{2L_1^2}+\frac{3\beta}{4L_1}-\frac{\beta^2}{16\pi^2}+o(\beta^2),\,
\mu_g(\beta)=\frac{\pi^2}{2L_1^2}+\frac{3\beta}{2L_1}-\frac{3\beta^2}{16\pi^2}+o(\beta^2).
\ee
Similarly, we can obtain results for the first excited state
\be\label{ge637}
E_1(\beta)=\frac{2\pi^2}{L_1^2}+\frac{3\beta}{4L_1}-\frac{\beta^2}{64\pi^2}+o(\beta^2),\,
\mu_1(\beta)=\frac{2\pi^2}{L_1^2}+\frac{3\beta}{2L_1}-\frac{3\beta^2}{64\pi^2}+o(\beta^2).
\ee
Subtracting (\ref{ge634}) from (\ref{ge637}), we obtain (\ref{gapbox14}) when $d=1$.
\end{proof}

\begin{lemma}\label{box:ground_strong}
In the strongly repulsive interaction regime, i.e. $\beta\gg1$, we have
\begin{align}
\label{boxtfg1}
E_g(\beta)&=\frac{A_0^2}{2}\beta+\frac{4A_0A_3}{3}\beta^{\frac{1}{2}}+2A_3^2-\frac{8A_4}{9}+o(1),\\
\label{boxtfm1}
\mu_g(\beta)&=A_0^2\beta+2A_0A_3\beta^{\frac{1}{2}}+2A_3^2-A_4+o(1),\qquad \beta\gg1,\\
\label{boxtfg2}
E_1(\beta)&=\frac{A_0^2}{2}\beta+\frac{4A_0(A_3L_1+1)}{3L_1}\beta^{\frac{1}{2}}
+\frac{2(A_3L_1+1)^2}{L_1^2}-\frac{8A_5}{9}+o(1),\\
\label{boxtfm2}
\mu_1(\beta)&=A_0^2\beta+\frac{2A_0(A_3L_1+1)}{L_1}\beta^{\frac{1}{2}}+\frac{2(A_3L_1+1)^2}{L_1^2}
-A_5+o(1),
\end{align}
where
\be
A_3=\sum_{j=1}^d\frac{1}{L_j},\quad
A_4=4\sum_{1\le j<k\le d}\frac{1_{\{d\ge2\}}}{L_jL_k},\quad A_5=A_4+4\sum_{1<j\le d}\frac{1_{\{d\ge2\}}}{L_1L_j},
\ee
with $1_{\{d\ge2\}}$ the standard set function, which takes  $1$ when $d\ge2$ and $0$ otherwise.
\end{lemma}

\begin{proof}
When $\beta\gg1$,
the ground and first excited states can be approximated by
the Thomas-Fermi (TF) approximations and/or uniformly accurate matched
approximations. For $d=1$ and $\Omega=(0,L)$, these approximations have
been given explicitly and verified numerically in the literature \cite{Wz1,Wg1,BaoL,LZBao} as
\begin{equation}\label{def:TF}
\phi_g(x)\approx \sqrt{\frac{\mu_g}{\beta}}\,\phi_{L,\mu_g}^{(1)}(x), \qquad
\phi_1(x)\approx \sqrt{\frac{\mu_1}{\beta}}\,\phi_{L,\mu_1}^{(2)}(x), \qquad 0\le x\le L,
\end{equation}
where
\be\label{def:TF_basic}
\begin{split}
\phi_{L,\mu}^{(1)}(x)&=\tanh\left(\sqrt{\mu}x\right)+\tanh\left(\sqrt{\mu}(L-x)\right)-
\tanh\left(\sqrt{\mu}L\right),\
0\le x\le L, \\
\phi_{L,\mu}^{(2)}(x)&=\tanh\left(\sqrt{\mu}x\right)-\tanh\left(\sqrt{\mu}(L-x)\right)
+\tanh\left(\sqrt{\mu}\left(L/2-x\right)\right),
\end{split}
\ee
with $\mu_g$ and $\mu_1$ determined from the normalization condition (\ref{norm}) and
$\tanh\left(\sqrt{\mu}L\right)\approx 1$.
These results in 1D can be extended to $d$-dimensions ($d=1,2,3$) for the approximations of
the ground and  first excited states as
\bea \label{boxg31}
\phi_g^{\beta}(\mathbf{x})&\approx&\phi_g^{\rm MA}(\mathbf{x})=
\sqrt{\frac{\mu_g(\beta)}{\beta}}\,\prod_{j=1}^d\phi_{L_j,\mu_g}^{(1)}(x_j), \qquad \bx\in\bar{\Omega},\\
\label{boxe31}
\phi_1^{\beta}(\mathbf{x})&\approx&\phi_1^{\rm MA}(\mathbf{x})=
\sqrt{\frac{\mu_1(\beta)}{\beta}}\,\phi_{L_1,\mu_1}^{(2)}(x_1)\prod_{j=2}^d\phi_{L_j,\mu_1}^{(1)}(x_j),
\eea
where $\mu_g(\beta)$ and $\mu_1(\beta)$ are determined from the normalization condition (\ref{norm}).
Inserting \eqref{boxg31} and \eqref{boxe31} into
\eqref{def:mu} and \eqref{def:E}, after a detailed computation which is omitted here for brevity, we can obtain \eqref{boxtfg1}-\eqref{boxtfm2}.
\end{proof}


%

From Lemmas \ref{box:ground_weak}-\ref{box:ground_strong}, we
have asymptotic results for the
 fundamental gaps.

\begin{proposition}[For GPE under a box potential in nondegenerate case]\label{asym:box}
When $\Omega=\Omega_0$ satisfying $d=1$ or $L_1> L_2$ when $d\ge2$
and $V(\mathbf{x})\equiv 0$ for $\bx\in\Omega$ in (\ref{eq:eig}), i.e. GPE with a box potential,
we have
\be \label{gapbox1}
\delta_E(\beta)=\left\{\begin{array}{ll}
\frac{3\pi^2}{2L_1^2}+o(\beta),\\
\frac{4A_0}{3L_1}\beta^{\frac{1}{2}}+A_1+o(1),\\
\end{array}\right.
\delta_{\mu}(\beta)=\left\{\begin{array}{ll}
\frac{3\pi^2}{2L_1^2}+o(\beta),&0\le\beta\ll1,\\
\frac{2A_0}{L_1}\beta^{\frac{1}{2}}+\frac{6}{L_1^2}+o(1),&\beta\gg1.\\
\end{array}\right.
\ee
\end{proposition}

\begin{proof} When $0\le \beta \ll1$, subtracting (\ref{boxsmall})
from  (\ref{boxsmall1}), noting (\ref{gapp21}), we obtain (\ref{gapbox1})
in this parameter regime. Similarly, when $\beta\gg1$,
subtracting (\ref{boxtfg1}) and (\ref{boxtfm1})
from  (\ref{boxtfg2}) and (\ref{boxtfm2}), respectively, we get (\ref{gapbox1})
in this parameter regime.
\end{proof}

\begin{figure}[htb]
\centerline{\psfig{figure=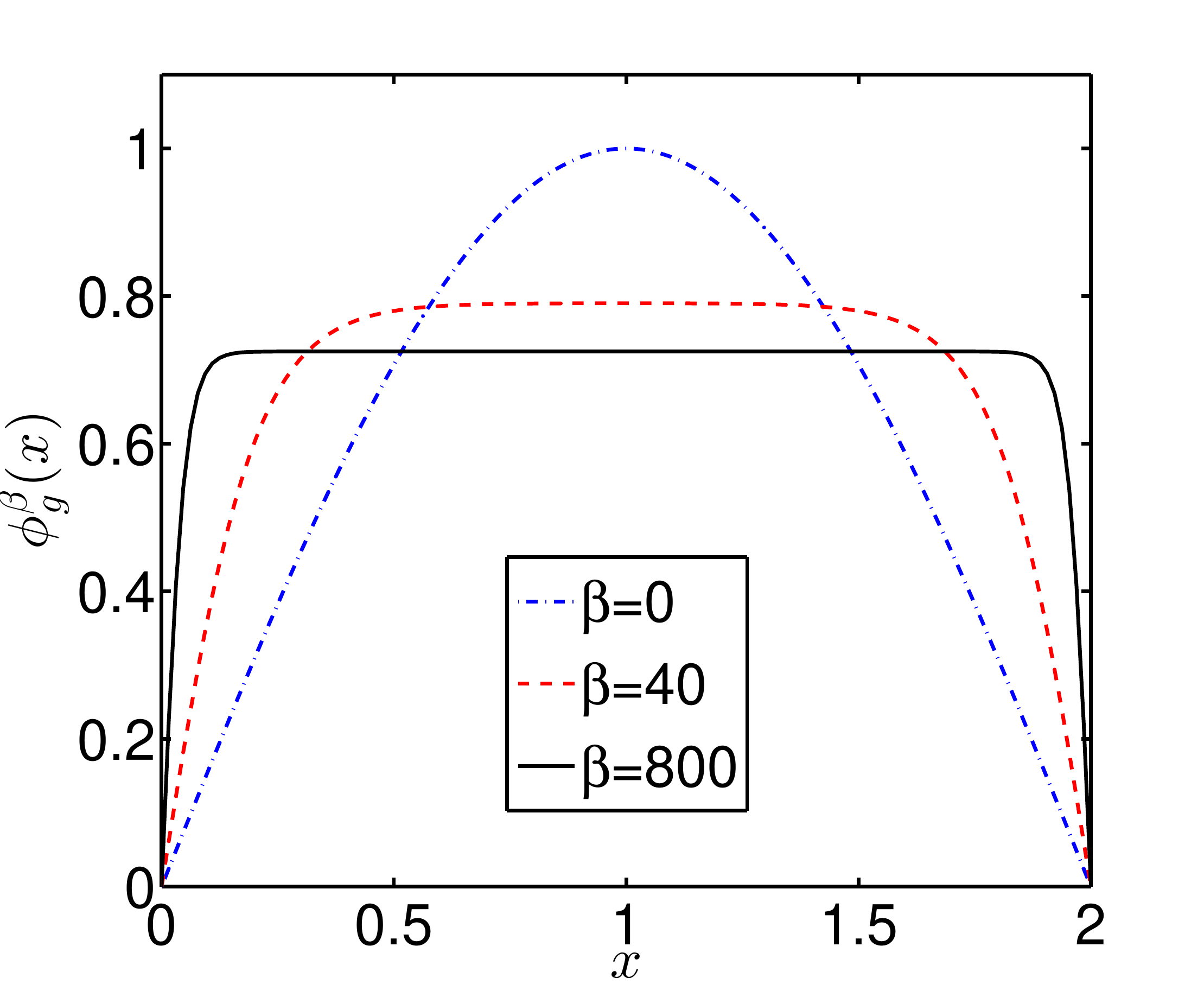,height=4cm,width=6.5cm,angle=0}
\psfig{figure=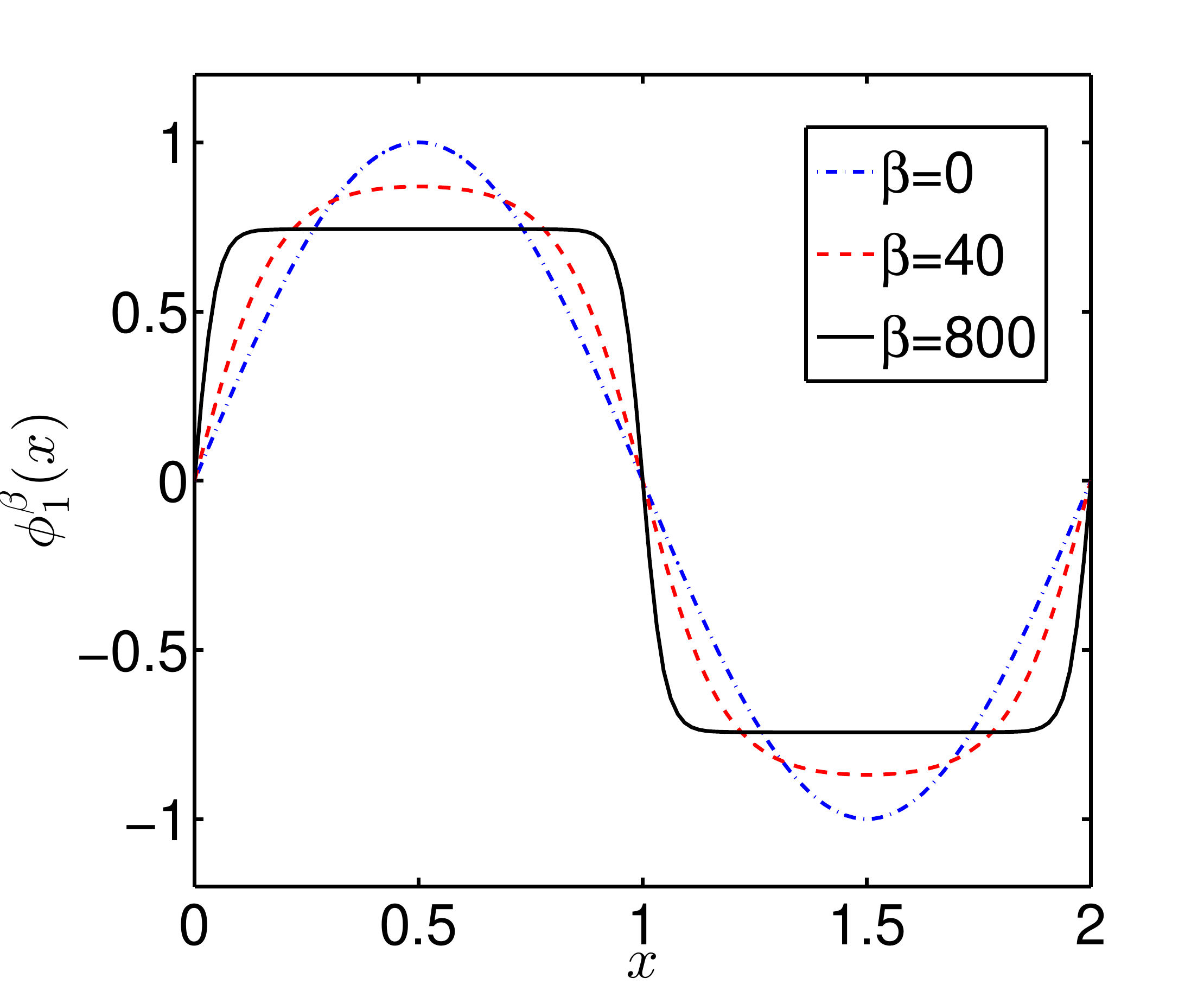,height=4cm,width=6.5cm,angle=0}}
\caption{Ground states $\phi_g^\beta(x)$ (left) and first excited states $\phi_1^\beta(x)$ (right) of GPE in 1D
with $\Omega=(0,2)$ and a box potential for different $\beta\ge0$. }
\label{fig:box1d_sol}
\end{figure}

\begin{figure}[htb]
\centerline{\psfig{figure=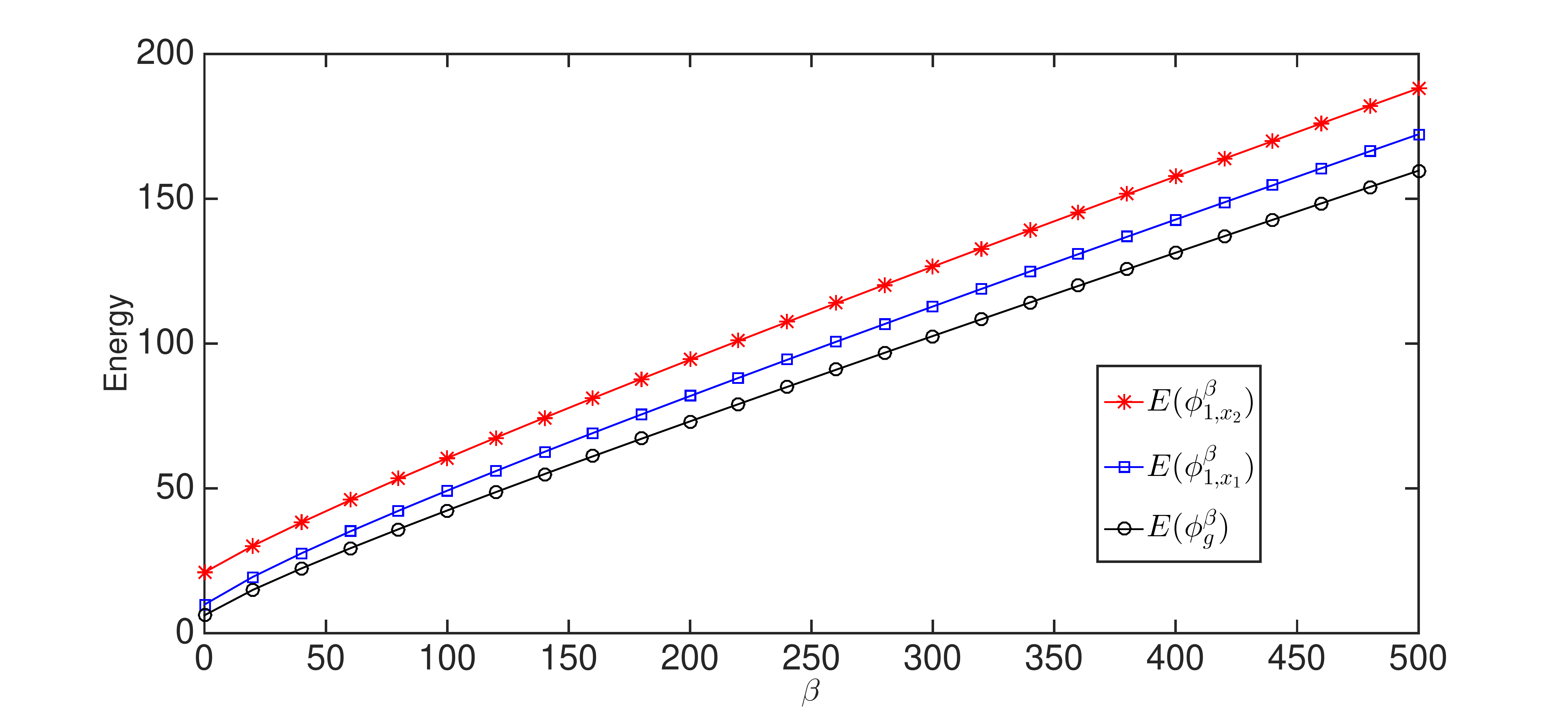,height=4cm,width=13cm,angle=0}}
\caption{Energy $E_g(\beta):=E(\phi_g^{\beta})<E_1(\beta):=E(\phi_1^\beta=\phi_{1,x_1}^{\beta}) <E_2(\beta):=E(\phi_{1,x_2}^{\beta})$ of GPE in 2D with $\Omega=(0,2)\times(0,1)$ and a box potential for different $\beta\ge0$.}
\label{fig:box_dE_general_2d}
\end{figure}

\begin{figure}[htb]
\centerline{\psfig{figure=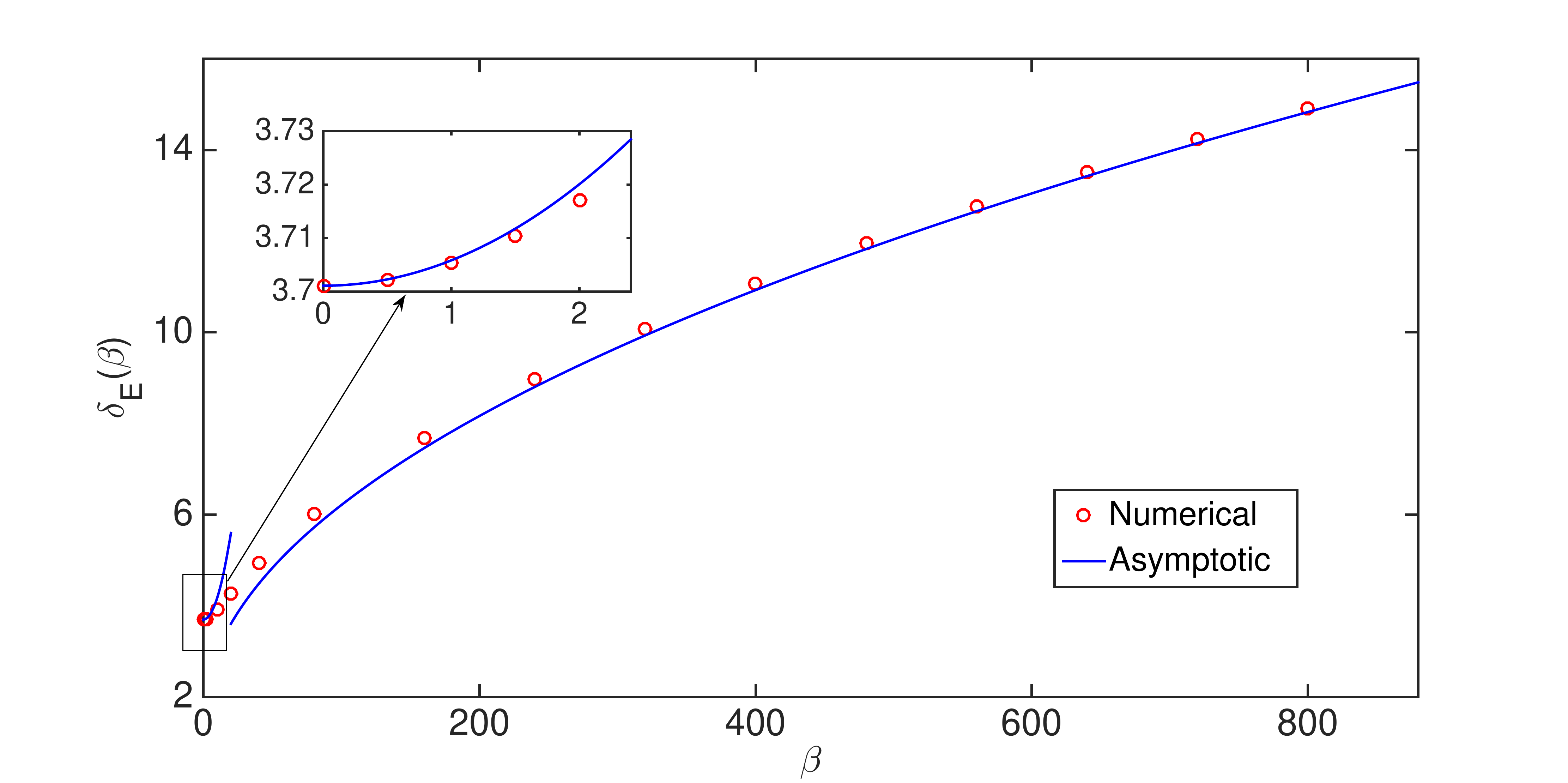,height=3.5cm,width=13cm,angle=0}}
\centerline{\psfig{figure=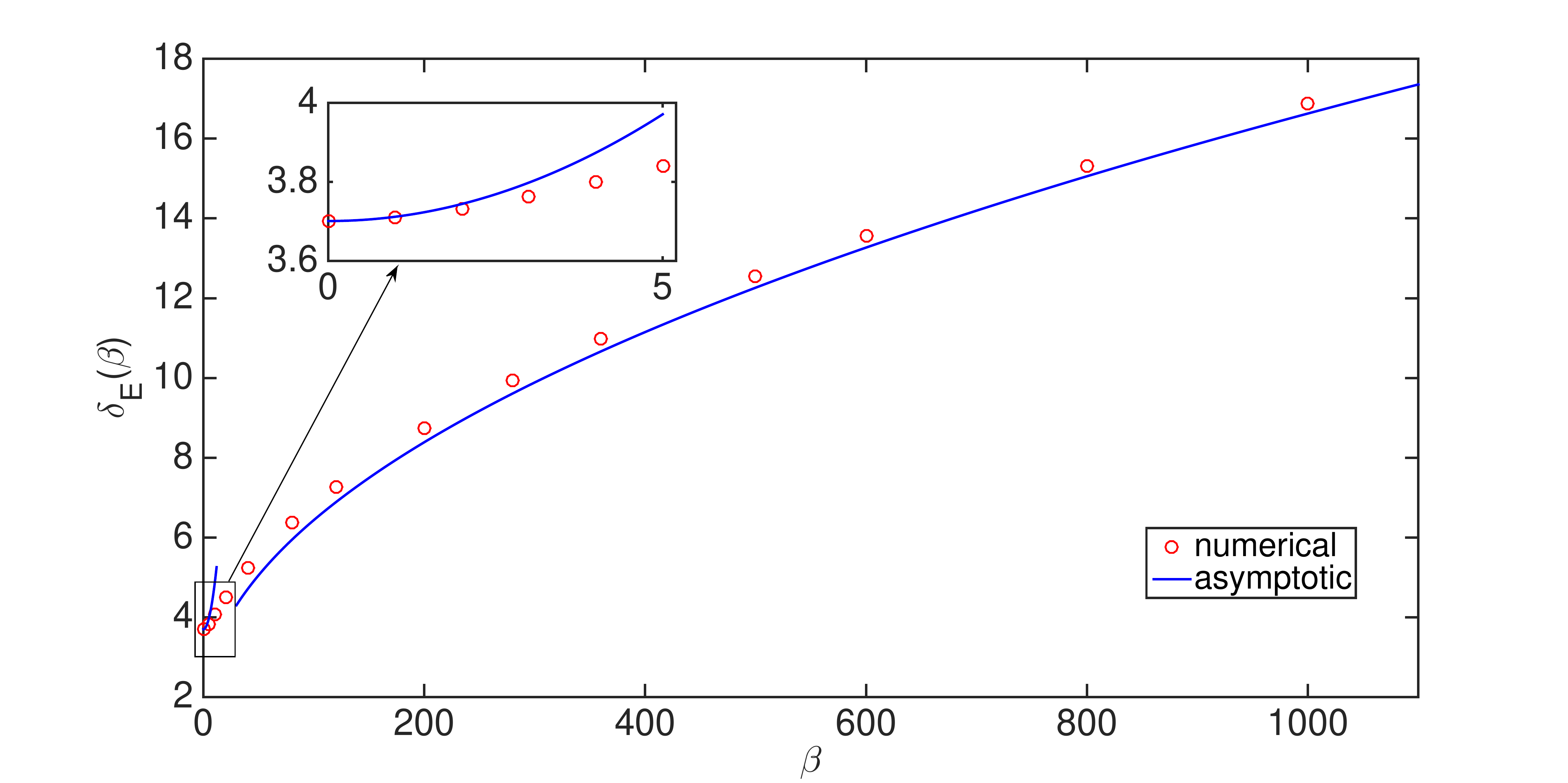,height=3.5cm,width=13cm,angle=0}}
\centerline{\psfig{figure=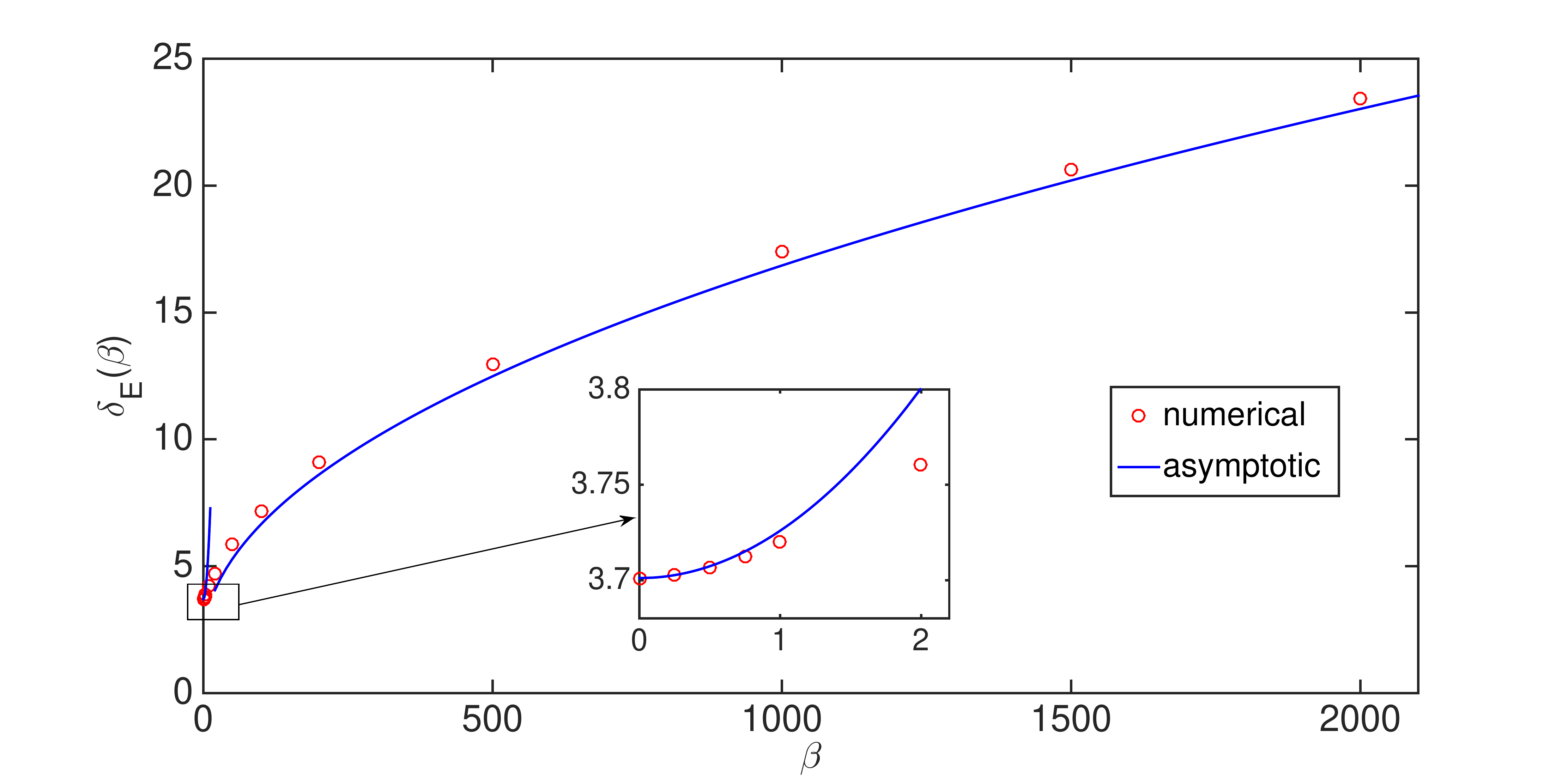,height=3.5cm,width=13cm,angle=0}}
\caption{Fundamental gaps in energy  of GPE with a box potential in
1D with $\Omega=(0,2)$ (top), in 2D with $\Omega=(0,2)\times(0,1)$ (middle), and
in 3D with $\Omega=(0,2)\times(0,1)\times(0,1)$ (bottom).}
\label{fig:box_dE_asym}
\end{figure}

To verify numerically our asymptotic results in Proposition \ref{asym:box},
we solve the time-independent GPE (\ref{eq:eig}) numerically by using the
normalized gradient flow via backward Euler finite difference discretization \cite{Bao_comp1,comp_gf,Bao2013,Wz1}
to find the ground and first excited states and their corresponding energy and chemical potentials.
Fig.~\ref{fig:box1d_sol} shows the ground and first excited states for different $\beta\ge0$ in 1D,
Fig.~\ref{fig:box_dE_general_2d} shows the energy of the ground  and excited states which are excited in $x_1$- or $x_2$-direction, while the
first excited state is taken as the one excited in $x_1$-direction,
and Fig.~\ref{fig:box_dE_asym} depicts fundamental gaps in energy obtained numerically and asymptotically
in 1D, 2D and 3D.
From Fig. \ref{fig:box_dE_asym}, we can see that the asymptotic results in Proposition \ref{asym:box}
are very accurate in both weakly
and strongly repulsive interaction regimes.
In addition, our numerical results suggest that both $\delta_E(\beta)$ and $\delta_\mu(\beta)$
are increasing functions for $\beta\ge0$ (cf. Fig. \ref{fig:box_dE_asym}).

\begin{figure}[htb]
\centerline{\psfig{figure=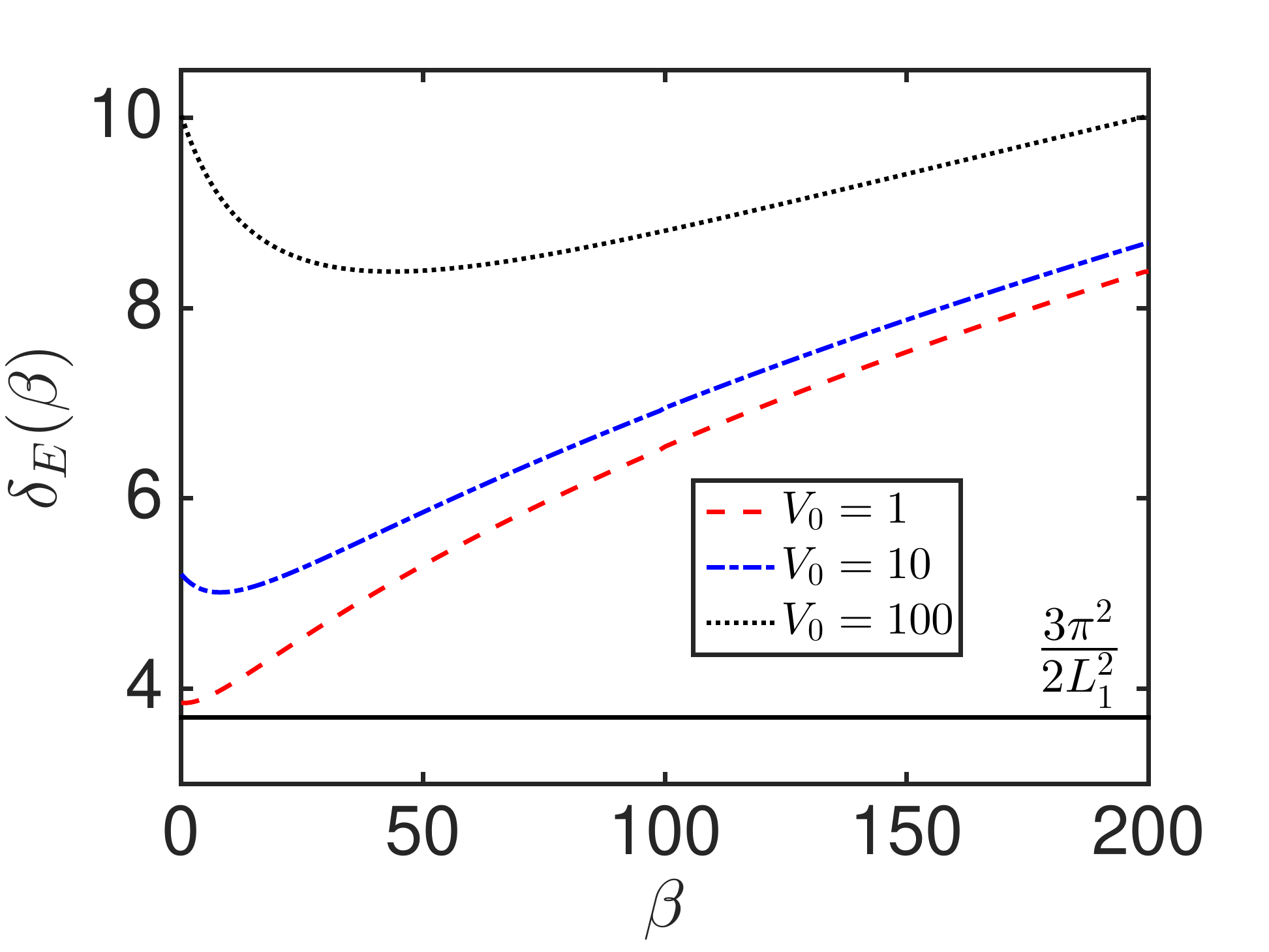,height=4cm,width=6.5cm,angle=0}
\psfig{figure=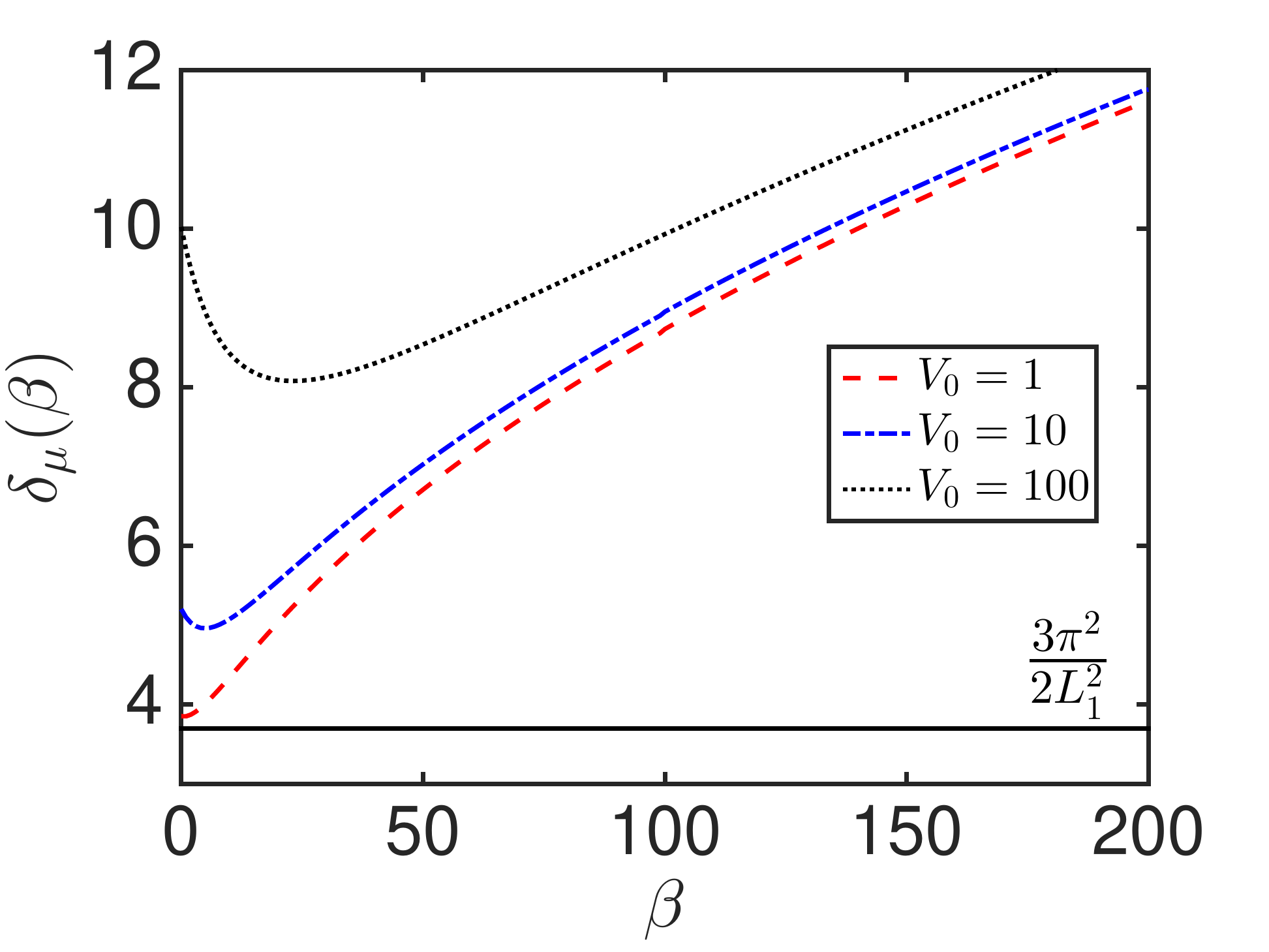,height=4cm,width=6.5cm,angle=0}}
\caption{Fundamental gaps in energy $\delta_E(\beta)$ (left) and chemical potential $\delta_\mu(\beta)$ (right)
of GPE in 1D with $\Omega=(0,2)$ and $V(x)=V_0(x-1)^2$ for different $V_0>0$ and $\beta\ge0$. }
\label{fig:test_box_dE}
\end{figure}

For a general bounded domain $\Omega$ and/or $V(\bx)\ne0$, we cannot get asymptotic
results on the fundamental gaps, but we can study the problem numerically.
If $\Og$ and $V(\bx)$ are symmetric with respect to the axis,
we can compute numerically the ground and first excited states and their
corresponding energy and chemical potential
as well as the fundamental gaps via the
normalized gradient flow method 
\cite{Bao_comp1,comp_gf,Bao2013,Wz1}.
We remark here that for a general bounded domain $\Omega$, discretization
in space can be performed by using the finite element method instead of finite difference or spectral method
for the normalized gradient flow to compute the ground and first excited states \cite{BaoT}.
For arbitrarily chosen external potentials, the first excited state might
not have any symmetric property. In this case, we can obtain numerically the first excited  state by using the numerical method proposed in
 \cite{BaoT}.
Fig. \ref{fig:test_box_dE} depicts fundamental gaps
in energy and chemical potential of the GPE in 1D with $\Omega=(0,2)$ and the potential $V(x)=V_0(x-1)^2$ for different $V_0$ and $\beta$.
Fig. \ref{fig:test_counterexample_dE} plots the fundamental gaps in energy and chemical potential of the GPE in 1D with
 $\Og=(0,2)$  and some nonconvex trapping potentials for different $\beta\ge0$.

\begin{figure}[htb]
\centerline{\psfig{figure=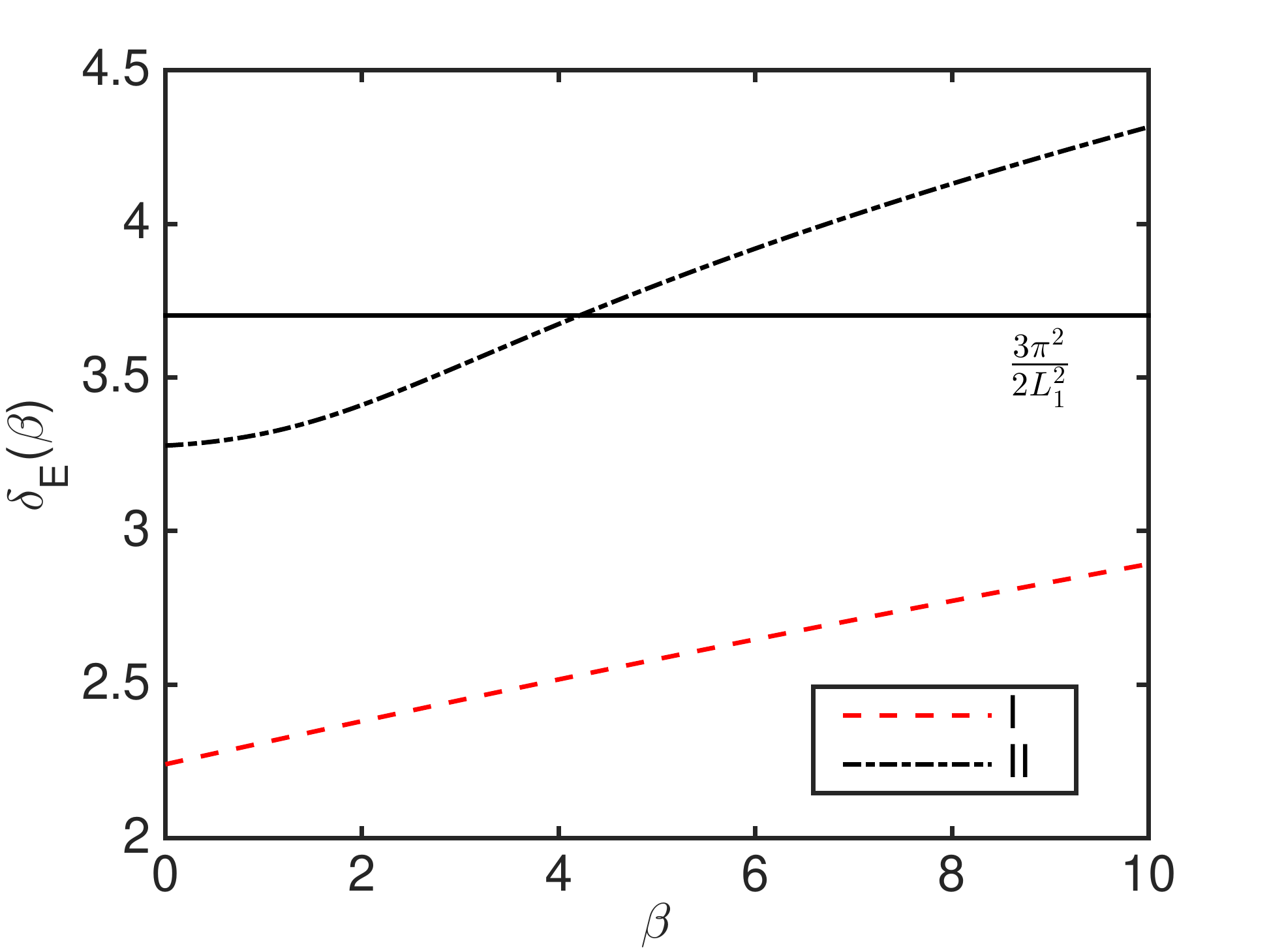,height=4cm,width=6.5cm,angle=0}
\psfig{figure=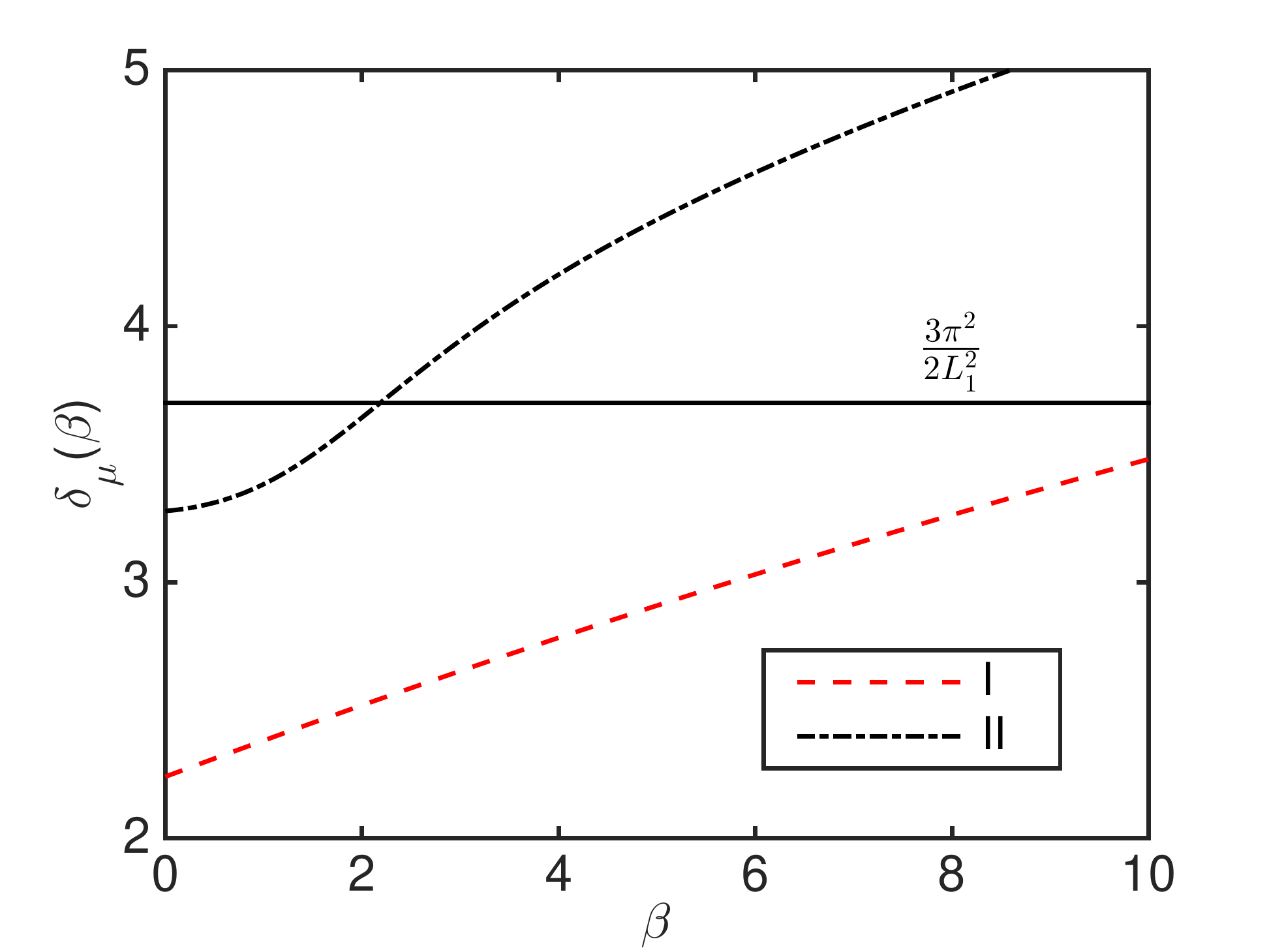,height=4cm,width=6.5cm,angle=0}}
\caption{Fundamental gaps in energy (left) and chemical potential (right)
of GPE in 1D with $\Og=(0,2)$ and non-convex trapping potentials:
(I) $V(x)=-10x^2$, and (II) $V(x)=10\sin(10(x-1))$ for different $\beta\ge0$. }
\label{fig:test_counterexample_dE}
\end{figure}

Based on the asymptotic results in Proposition \ref{asym:box} and the above numerical results as well as additional extensive numerical results not shown here for brevity \cite{Ruan}, we speculate the following gap conjecture.

{\bf Gap conjecture}
(For GPE on a bounded domain with homogeneous Dirichlet BC in nondegenerate case)
Suppose $\Omega$ is a convex bounded domain,
  the external potential $V(\mathbf{x})$ is convex and $\dim(W_1)=1$,  we have
\begin{equation}\label{bgap976}
\delta^{\infty}_E:=\inf_{\beta\ge0} \delta_E(\beta)\ge\frac{3\pi^2}{2D^2},
\qquad\delta^{\infty}_{\mu}:=\inf_{\beta\ge0} \delta_\mu(\beta)\ge\frac{3\pi^2}{2D^2},
\end{equation}
where $D:=\sup_{\mathbf{x},\mathbf{z}\in\Omega}|\mathbf{x}-\mathbf{z}|$ is the diameter of $\Omega$.
In fact, our numerical results
suggest a stronger gap as
\bea
\label{bapb987}
\qquad \quad \delta_{E}(\beta)\ge\left\{\ba{ll}
\frac{3\pi^2}{2D^2}, &0\le \beta\le \frac{81\pi^4|\Omega|}{64D^2},\\
\frac{4\beta^{1/2}}{3D|\Omega|^{1/2}}, &\beta\ge\frac{81\pi^4|\Omega|}{64D^2},\\
\ea\right.
\delta_{\mu}(\beta)\ge\left\{\ba{ll}
\frac{3\pi^2}{2D^2}, &0\le \beta\le\frac{9\pi^4|\Omega|}{16D^2},\\
\frac{2\beta^{1/2}}{D|\Omega|^{1/2}}, &\beta\ge\frac{9\pi^4|\Omega|}{16D^2}, \\
\ea\right.
\eea
where $|\Omega|$ is the volume of $\Omega$.
On the other hand, Fig.~\ref{fig:test_counterexample_dE} suggests that the gap
conjecture (\ref{bgap976}) is not valid for non-convex trapping potentials.


\subsection{Degenerate case, i.e. $\dim(W_1)\ge2$}\label{sec:box_asym_degenerate}
Again, we first consider a special case by taking $\Omega=\Omega_0$ satisfying $L_1=L_2:=L$ and $d\ge2$
and $V(\bx)\equiv0$ for $\bx\in \Omega$ in \eqref{eq:eig}.
In this case, the approximations of the ground states and their energy and chemical potential are the same as those in the previous subsection by letting $L_2\to L_1=L$. On the contrary, the approximations of the first excited states
are completely different with those in the non-degenerate case.

\begin{lemma}\label{box:ground_weak_degenerate2}
For weakly repulsive interaction, i.e. $0<\beta\ll1$, we have for $d\ge2$
\begin{align} \label{boxsmall_degenerate}
E_1(\beta)&=\frac{3\pi^2}{2L^2}+A_2+\frac{13d}{32}A_0^2\beta+o(\beta), \,
\mu_1(\beta)=\frac{3\pi^2}{2L^2}+A_2+\frac{13d}{16}A_0^2\beta+o(\beta).
\end{align}
\end{lemma}
\begin{proof} For simplicity, we only present the 2D case and
extension to 3D is straightforward. Denote
\be
\phi_g^0(x)=\sqrt{\frac{2}{L}}\sin\left(\frac{\pi x}{L}\right), \qquad \phi_1^0(x)=\sqrt{\frac{2}{L}}\sin\left(\frac{2\pi x}{L}\right), \qquad
0\le x\le L.
\ee
When $d=2$ and $\beta=0$, it is easy to see that
$\varphi_1(\bx):=\phi_1^{0}(x_1)\phi_g^{0}(x_2)$
and $\varphi_2(\bx):=\phi_g^{0}(x_1)\phi_1^{0}(x_2)$
are two linearly independent orthonormal first excited states.
In fact, $W_1={\rm span}\{\varphi_1,\varphi_2\}$.
In order to find an appropriate approximation of the first
excited state when $0<\beta \ll1$, we take an ansatz
\be\label{box01_degenerate}
\varphi_{a,b}(\mathbf{x})=a \varphi_1(\bx)+b\varphi_2(\bx),
\qquad \bx=(x_1,x_2)\in\bar{\Omega},
\ee
where $a,b\in\mathbb{C}$ satisfying $|a|^2+|b|^2=1$ implies
$\|\varphi_{a,b}\|_2=1$. Then $a$ and $b$ will be determined by
minimizing $E(\varphi_{a,b})$. Plugging \eqref{box01_degenerate}
into \eqref{def:E}, a simple direct computation implies that
\bea E(\varphi_{a,b})&=&\frac{3\pi^2}{2L^2}+A_2+\frac{8\beta}{L^4}\iint_{[0,L]^2}\left|a\sin(\frac{2\pi x_1}{L})\sin(\frac{\pi x_2}{L})+b\sin(\frac{\pi x_1}{L})\sin(\frac{2\pi x_2}{L})\right|^4\, d\bx\nn\\
&=&\frac{3\pi^2}{2L^2}+A_2+\frac{9\beta}{8L^2}(|a|^4+|b|^4)+
\frac{\beta}{2L^2}(4|a|^2|b|^2+a^2\bar{b}^2+\bar{a}^2b^2)\nn\\
&=&\frac{3\pi^2}{2L^2}+A_2+\frac{9\beta}{8L^2}+\frac{\beta}{4L^2}
(2a^2\bar{b}^2+2\bar{a}^2b^2-|a|^2|b|^2).\nn
\eea
To minimize $E(\varphi_{a,b})$, noting $|a|^2+|b|^2=1$, we take $a=e^{i\xi}\cos(\theta)$ and $b=e^{i\eta}\sin(\theta)$ with $\xi,\eta,\theta\in[-\pi,\pi)$. Then we have
\be\nn
E(\varphi_{a,b})=\frac{3\pi^2}{2L^2}+A_2+\frac{9\beta}{8L^2}-\frac{\beta}{16L^2}
\sin^2(2\theta)\left[1-4\cos\left(2(\xi-\eta)\right)\right],
\ee
which is minimized when $\theta=\pm\pi/4$ and $\xi-\eta=\pm\pi/2$,
i.e. $a=\pm ib$.
By taking $a=1/\sqrt{2}$ and $b=i/\sqrt{2}$, we obtain an approximation of the first excited state $\phi_1^\beta$ when $0<\beta\ll1$ as
\be\label{phi1box11}
\phi_1^{\beta}(\mathbf{x})\approx \frac{1}{\sqrt{2}}\left[\phi_1^{0}(x_1)\phi_g^{0}(x_2)
+i\phi_g^{0}(x_1)\phi_1^{0}(x_2)\right], \qquad \bx\in\bar\Omega.
\ee
Substituting \eqref{phi1box11} into (\ref{def:E}) and (\ref{def:mu}),
we get (\ref{boxsmall_degenerate}) when $d=2$.
\end{proof}


\begin{lemma}\label{box:ground_strong_degenerate} When $d=2$ and $\beta\gg1$,
   we have
\begin{align}
\label{boxstrong_degenerate_E}
&E_1(\beta)=\frac{\beta}{2L^2}+\frac{8\sqrt{\beta}}{3L^2}+
\frac{\pi}{2L^2}\ln(\beta)+o(\ln(\beta)),  \\
\label{boxstrong_degenerate_mu}
&\mu_1(\beta)=\frac{\beta}{L^2}+\frac{4\sqrt{\beta}}{L^2}+
\frac{\pi}{2L^2}\ln(\beta)+o(\ln(\beta)).
\end{align}
\end{lemma}
\begin{proof}
From Lemma \ref{box:ground_weak_degenerate2}, when $0<\beta\ll1$,
the first excited state needs to be taken as a vortex-type solution.
By assuming that there is no band crossing when $\beta>0$,
then the first excited state can be well approximated by the vortex-type solution when $\beta\gg1$ too. Thus when $\beta\gg1$, we approximate
the first excited state via a matched asymptotic approximation.

(i) In the outer region, i.e. $|\bx-(L/2,L/2)|>o(1)$, it is approximated
by the ground state profile as
\be
\phi_1^\beta(\bx)\approx \phi^{\rm{out}}(\mathbf{x})=\sqrt{\frac{\mu_1}
{\beta}}\phi^{(1)}_{L,\mu_1}(x_1)\phi^{(1)}_{L,\mu_1}(x_2),
\ee
where $\phi^{(1)}_{L,\mu}(x)$ is given in (\ref{def:TF_basic}) with $\mu=\mu_1(\beta)$  the chemical potential of the first excited state.

(ii) In the inner region near the center $(L/2,L/2)$, i.e. $|\bx-(L/2,L/2)|\ll1$, it is approximated by a
vortex solution with winding number $m=1$ as
\begin{equation}\label{sol:vortex}
\phi_1^\beta(\bx)\approx \phi^{\rm{in}}(\mathbf{x})=\sqrt{\frac{\mu_1}
{\beta}}f(r)e^{i\theta}, \qquad |\bx-(L/2,L/2)|\ll1,
\end{equation}
where $r$ and $\theta$ are the modulus and argument of $(x_1-L/2)+i(x_2-L/2)$, respectively.
Substituting (\ref{sol:vortex}) into (\ref{eq:eig}), we get the equation for $f(r)$
\begin{equation}\label{eq:vortex}
-\frac{1}{2}f''(r)-\frac{1}{2r}f'(r)+\frac{1}{2r^2}f(r)+\mu_1 f^3(r)=\mu_1f(r), \qquad r>0,
\end{equation}
with BCs $f(0)=0$ and $\lim_{r\to+\infty}f(r)=1$. When $\beta\gg1$, by dropping the term $-\frac{1}{2}f''(r)$
in (\ref{eq:vortex}) and then solving it analytically, we get
\be\label{sol:vortex_approx}
f(r)\approx f_a(r):=\sqrt{\frac{2\mu_1 r^2}{1+2\mu_1 r^2}},
\qquad r\ge0.
\ee
Combining the outer and inner approximations via the matched asymptotic
technique, we obtain an asymptotic approximation of the density of the
first excited state as
\be\label{box02_degenerate}
\rho_1^\beta(\bx):=|\phi_1^{\beta}(\mathbf{x})|^2\approx \sqrt{\frac{\mu_1}
{\beta}}\left[f_a^2(r)+\left(\phi^{(1)}_{L,\mu_1}(x_1)
\phi^{(1)}_{L,\mu_1}(x_2)\right)^2-1\right], \qquad \bx\in\bar\Omega.
\ee
Substituting \eqref{box02_degenerate} into
the normalization condition $\|\phi_1^\beta\|_2=1$, a detailed computation gives
the approximation of the chemical potential in (\ref{boxstrong_degenerate_mu}).
Plugging \eqref{box02_degenerate} into \eqref{def:mu} and noticing (\ref{boxstrong_degenerate_mu}), a detailed computation implies the approximation
of the energy in (\ref{boxstrong_degenerate_E}).  The details of the computation are omitted here for brevity \cite{Ruan}.
\end{proof}

From lemmas \ref{box:ground_weak}, \ref{box:ground_strong}, \ref{box:ground_weak_degenerate2} and \ref{box:ground_strong_degenerate}, we have the following result about the fundamental gaps for the degenerate case.

\begin{proposition}[For GPE under a box potential in degenerate case]\label{asym:box_degen}
When $\Omega=\Omega_0$ satisfying $L_1=L_2:=L$ and $d\ge2$
and $V(\mathbf{x})\equiv 0$ for $\bx\in\Omega$ in (\ref{eq:eig}), i.e. GPE with a box potential,
we have

(i) when $0\le\beta\ll1$ and $d\ge2$,
\be \label{gapbox2}
\delta_E(\beta)=
\frac{3\pi^2}{2L^2}-\frac{5dA_0^2}{32}\beta+o(\beta),\quad
\delta_{\mu}(\beta)=
\frac{3\pi^2}{2L^2}-\frac{5dA_0^2}{16}\beta+o(\beta);
\ee

(ii) when $\beta\gg1$ and $d=2$,
\be\label{gapbox3}
\delta_E(\beta)=\frac{\pi}{2L^2}\ln(\beta)+\mathcal{O}(1),\quad \delta_{\mu}(\beta)=\frac{\pi}{2L^2}\ln(\beta)+\mathcal{O}(1).
\ee

\end{proposition}


Again, to verify numerically our asymptotic results in Proposition \ref{asym:box_degen}, Fig.~\ref{fig:box2d_sol_degenerate} plots
the ground state $\phi_g^\beta$, the first excited state
$\phi_1^\beta=\phi_{1,v}^\beta$, and other excited states
$\phi_{1,x}^\beta$ and $\phi_{1,c}^\beta$, of the GPE in 2D with $\Omega=(0,2)^2$
and a box potential for different $\beta\ge0$,
which were obtained numerically \cite{Bao_comp1,comp_gf,Bao2013,Wz1}.
Fig.~\ref{fig:box_2d_degenerate} depicts the energy $E_g(\beta)=E(\phi_g^\beta)<E_1(\beta)=E(\phi_1^\beta=\phi_{1,v}^\beta)
<E(\phi_{1,x}^\beta)
<E(\phi_{1,c}^\beta)$ for different $\beta\ge0$ and the corresponding
fundamental gaps in energy, and Fig.~\ref{fig:box_3d_E_degenerate} shows
the fundamental gaps in energy of GPE in 3D with $\Omega=(0,1)^3$ and a box
potential. In addition, Fig.~\ref{fig:box2d_disk_degenerate}
depicts the fundamental gaps in energy of GPE in 2D with $\Omega=B_1({\bf 0})=\{\bx\ |\ |\bx|<1\}$ and a box
potential.

\begin{figure}[!]
\centerline{\psfig{figure=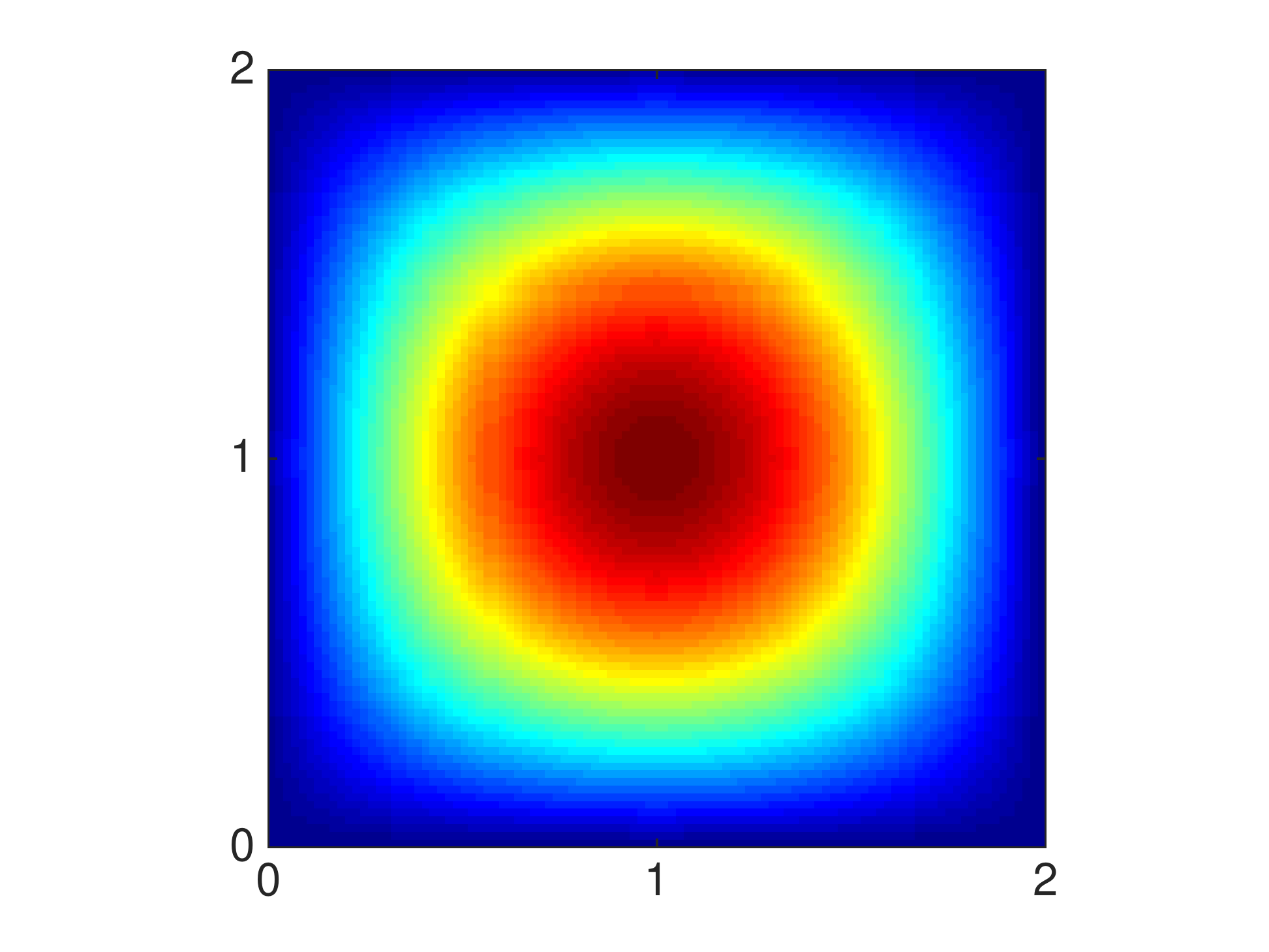,height=3.5cm,width=5cm,angle=0}
\psfig{figure=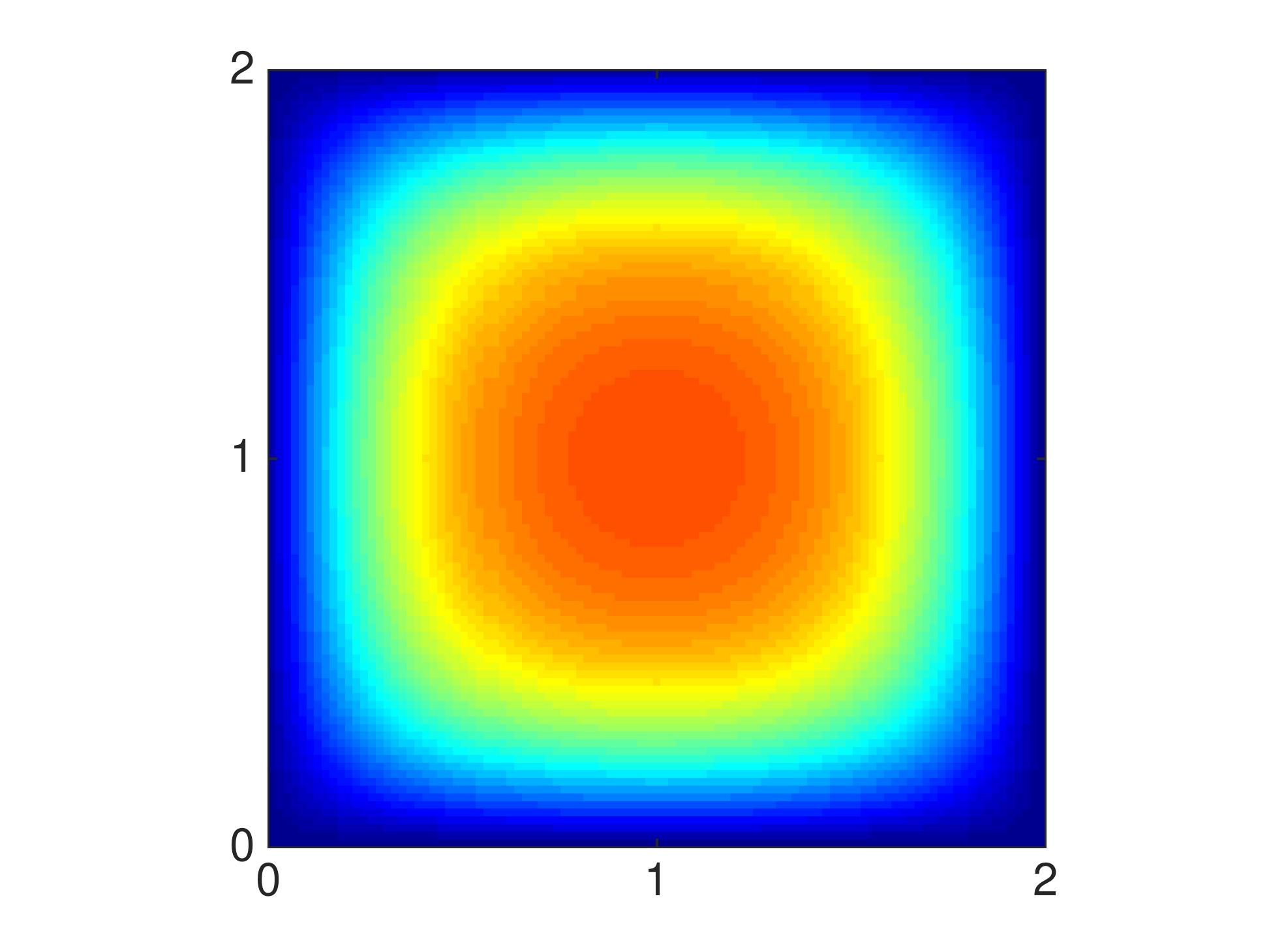,height=3.5cm,width=5cm,angle=0}
\psfig{figure=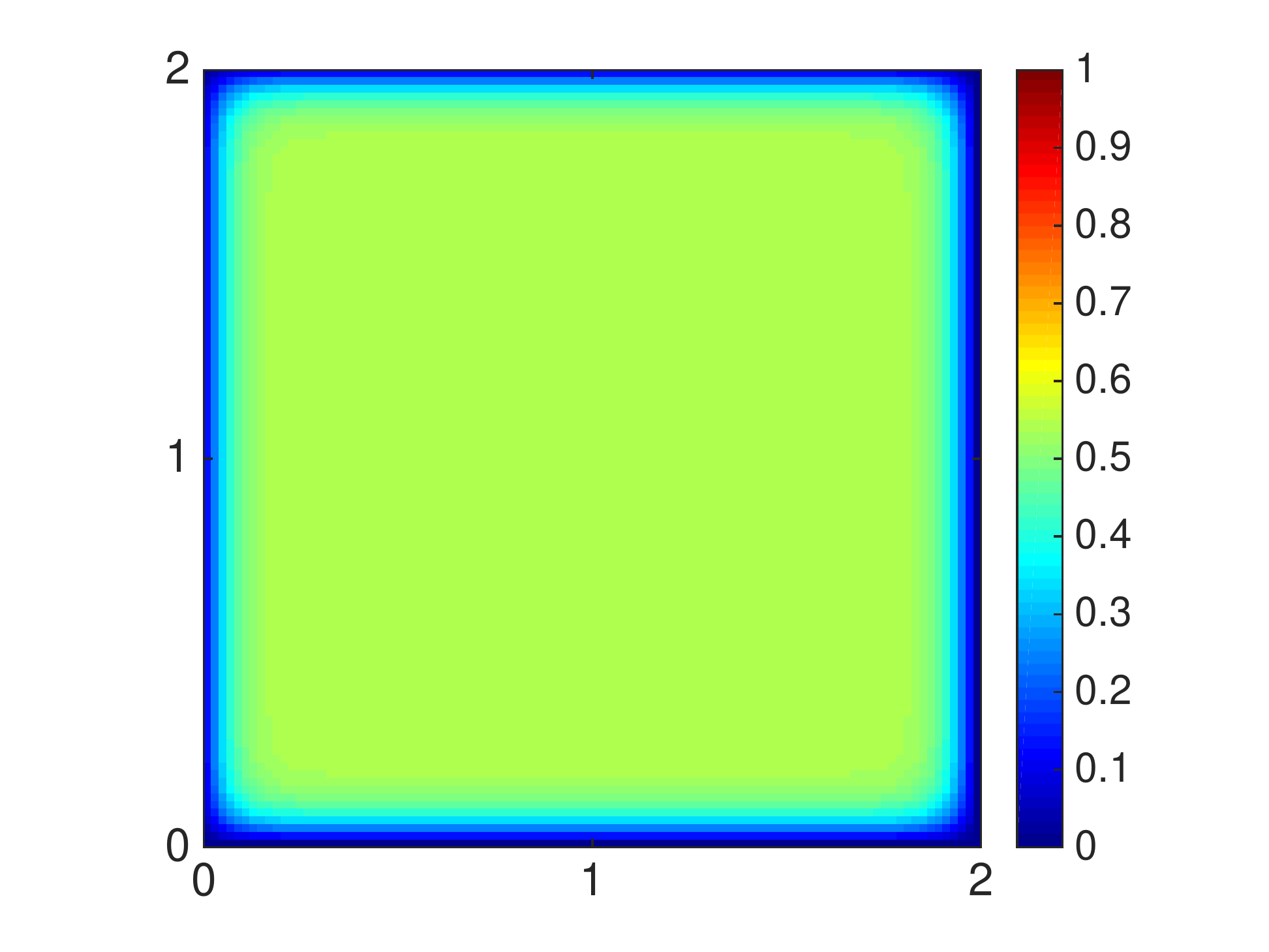,height=3.5cm,width=5cm,angle=0}}
\centerline{\psfig{figure=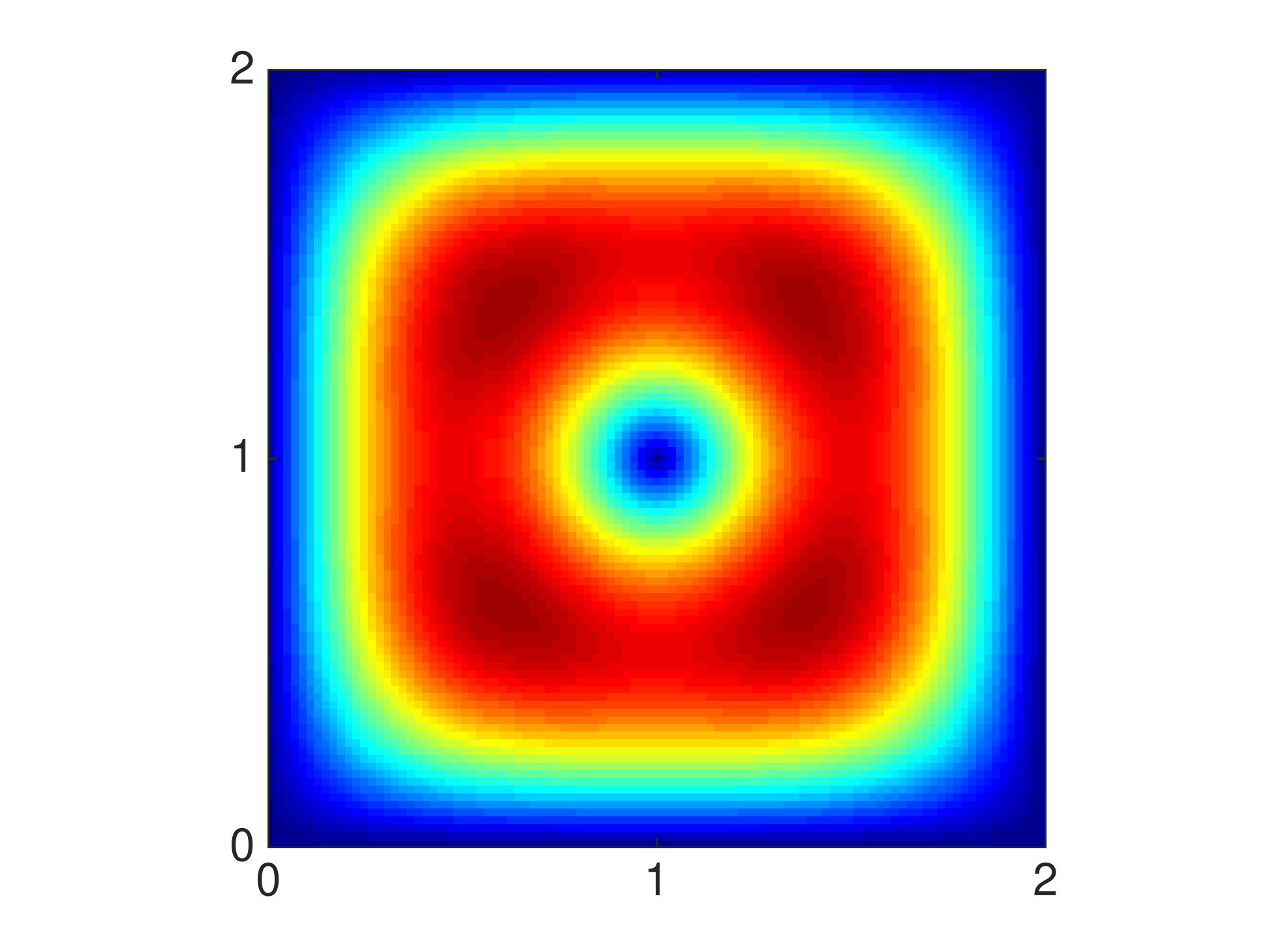,height=3.5cm,width=5cm,angle=0}
\psfig{figure=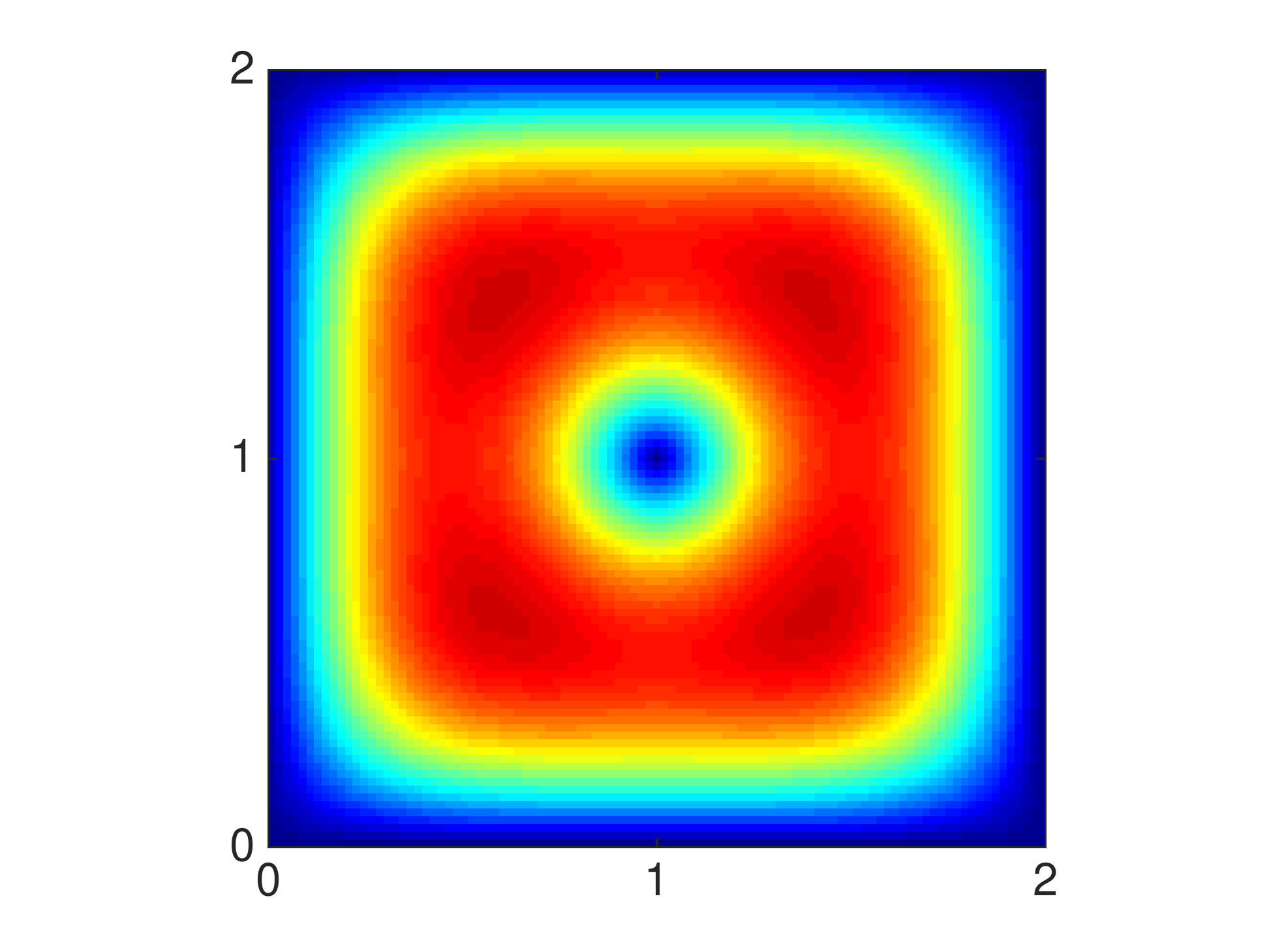,height=3.5cm,width=5cm,angle=0}
\psfig{figure=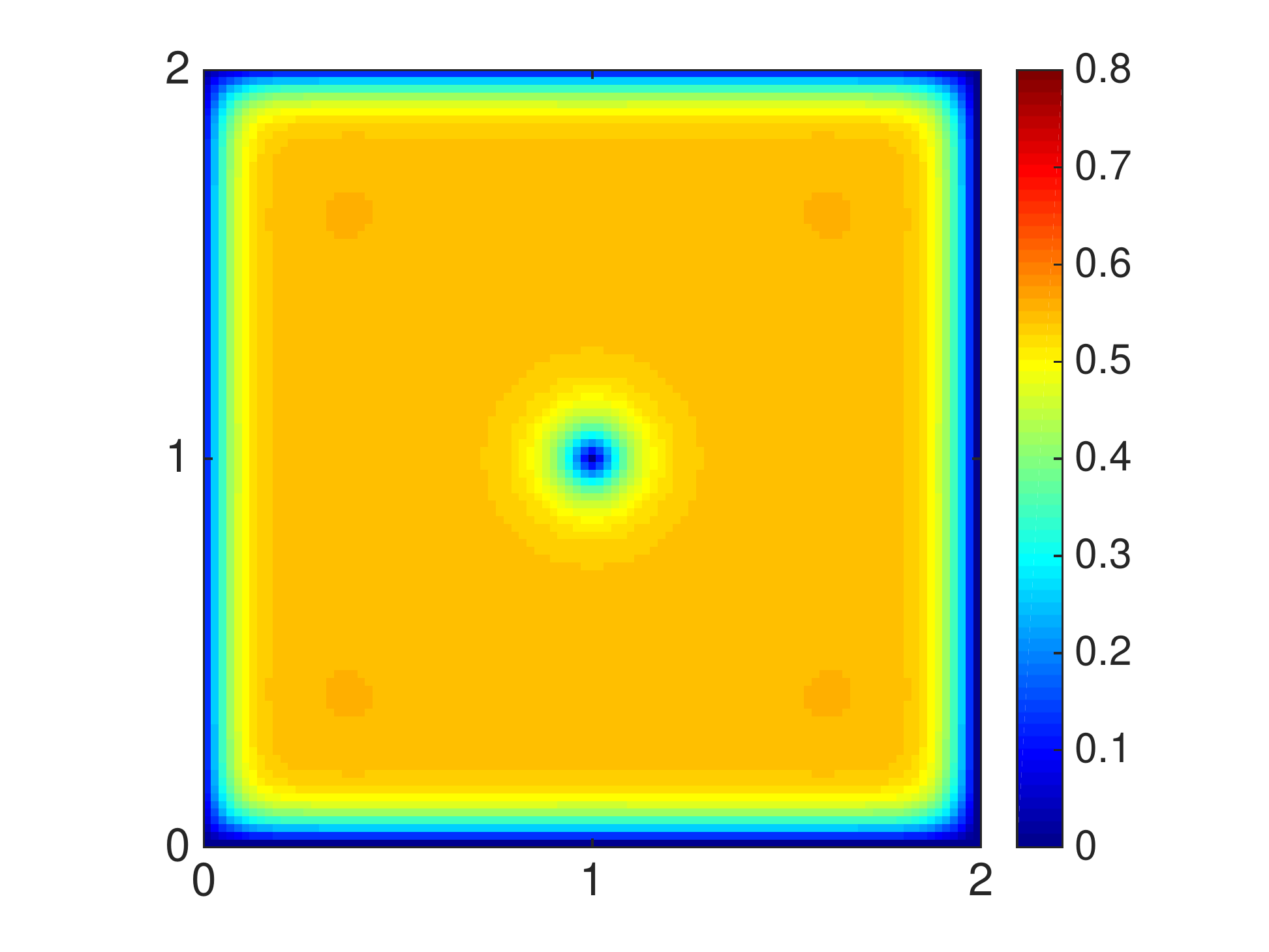,height=3.5cm,width=5cm,angle=0}}
\centerline{\psfig{figure=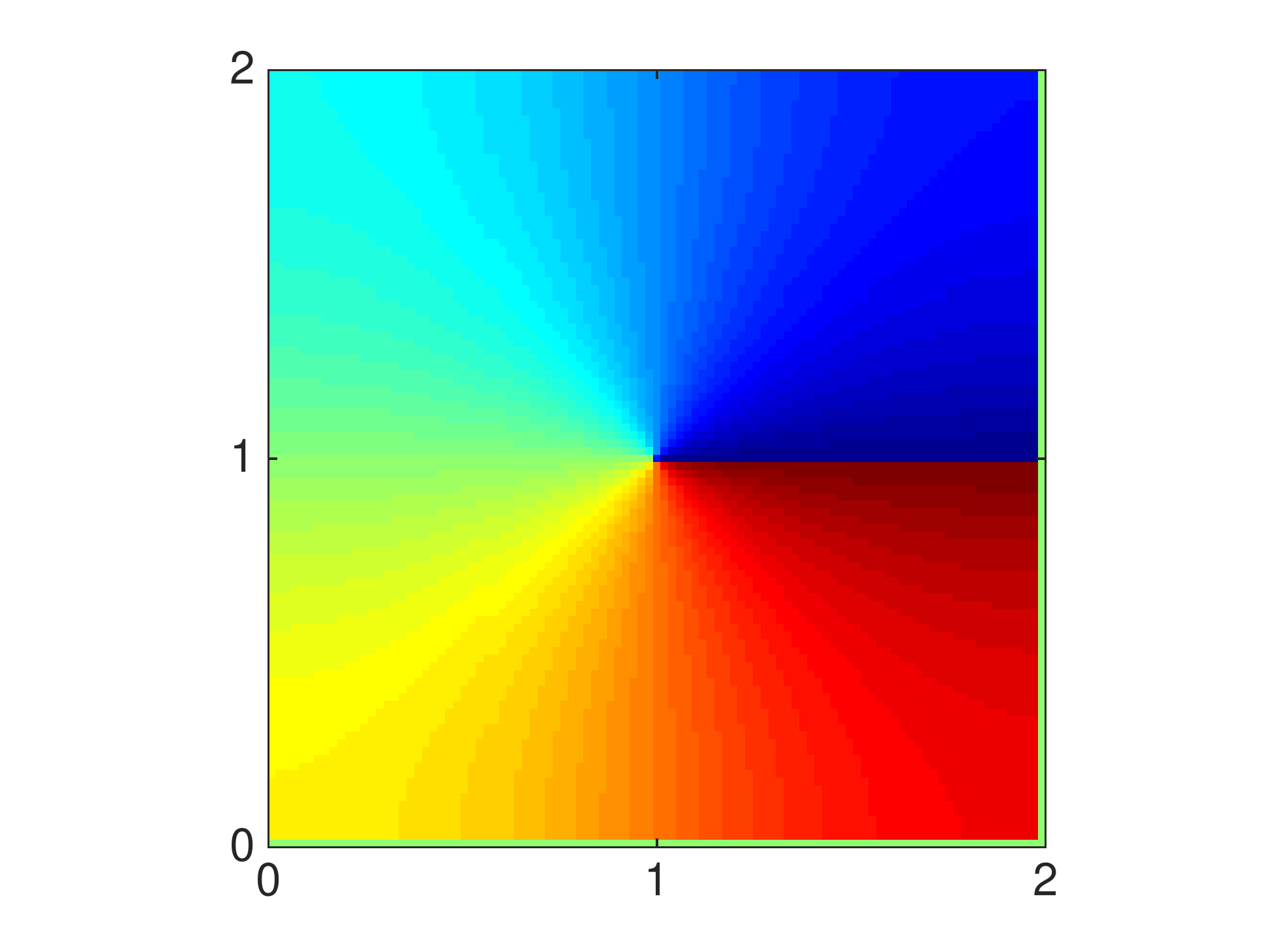,height=3.5cm,width=5cm,angle=0}
\psfig{figure=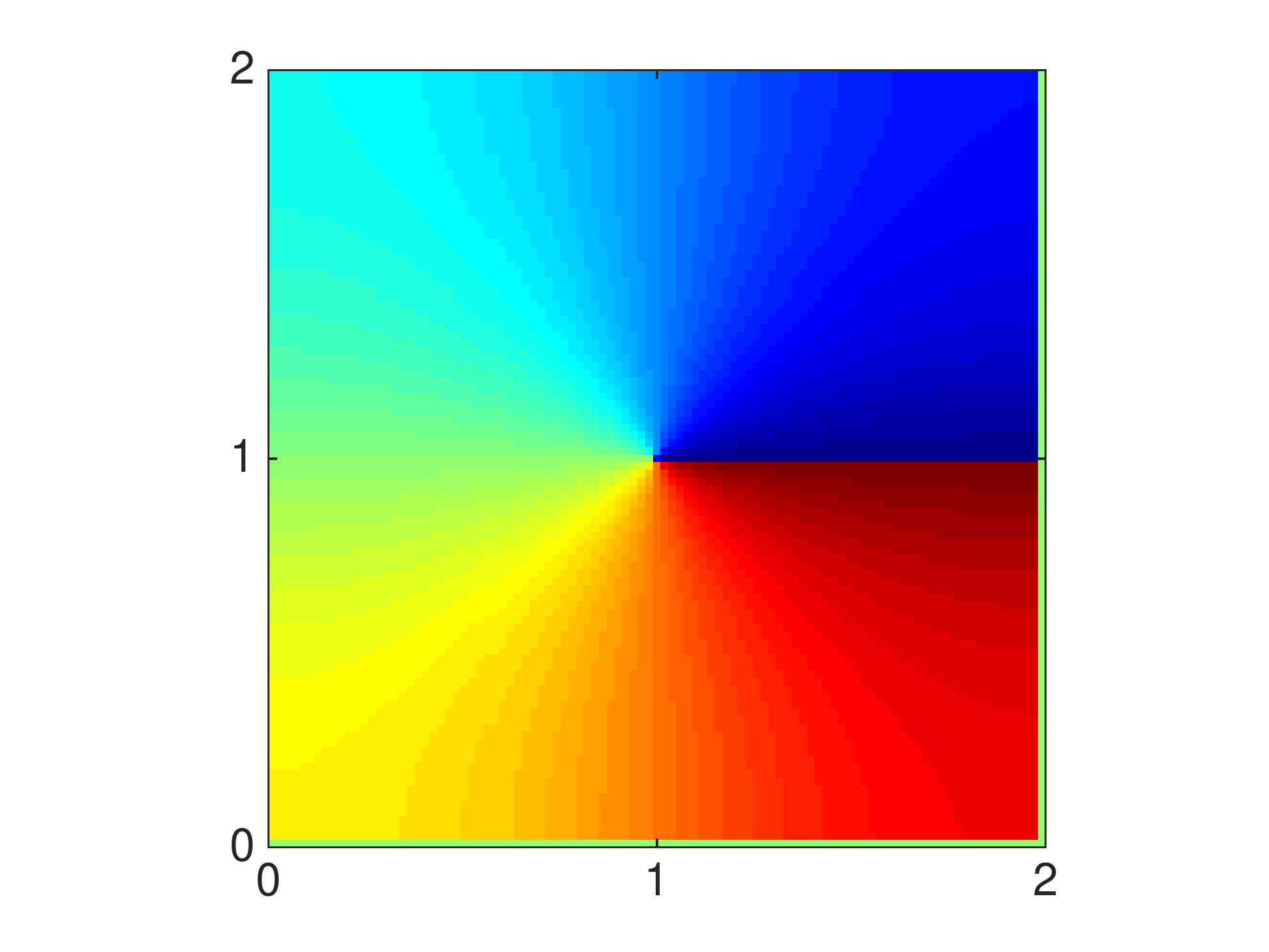,height=3.5cm,width=5cm,angle=0}
\psfig{figure=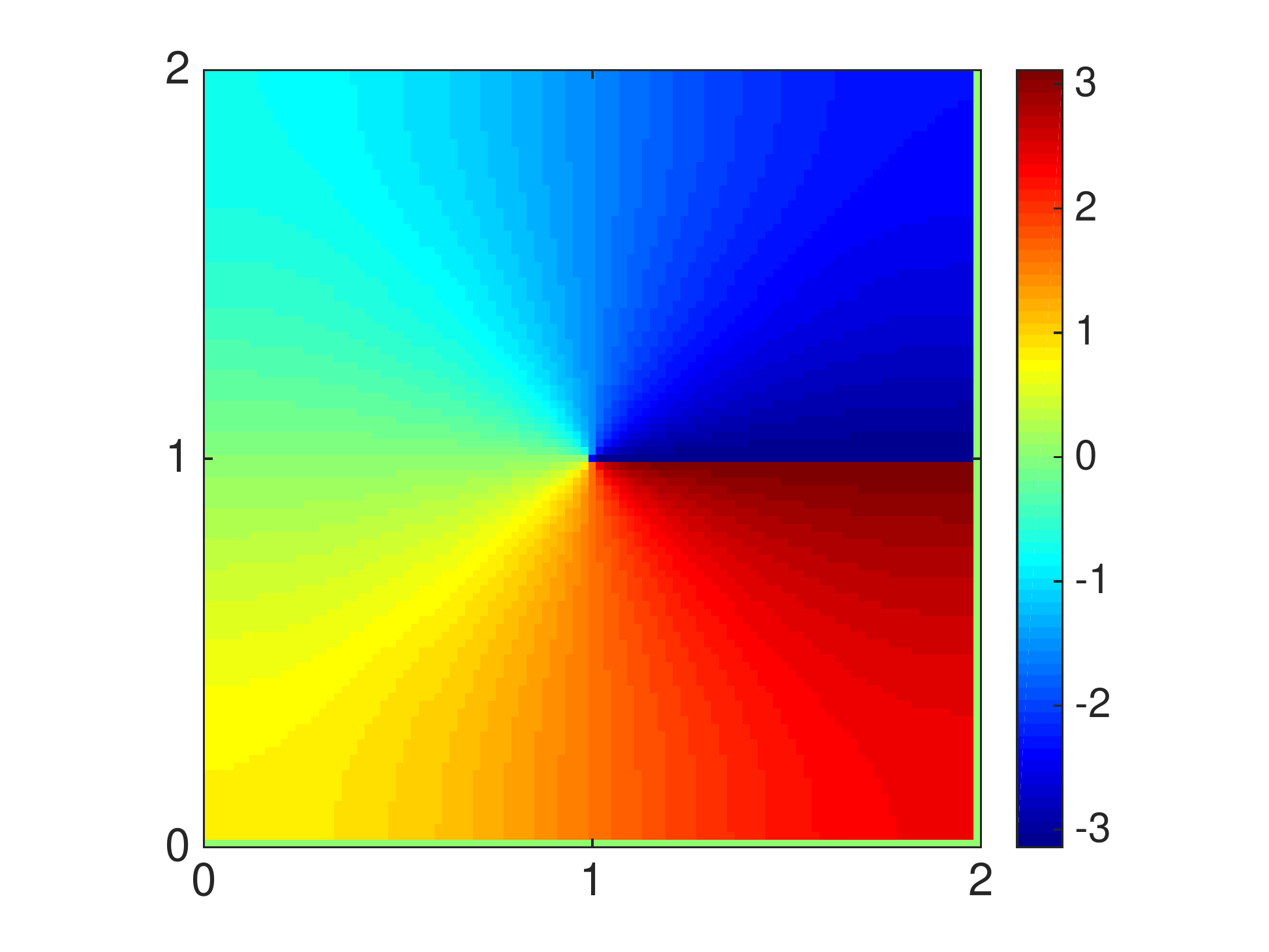,height=3.5cm,width=5cm,angle=0}}
\centerline{\psfig{figure=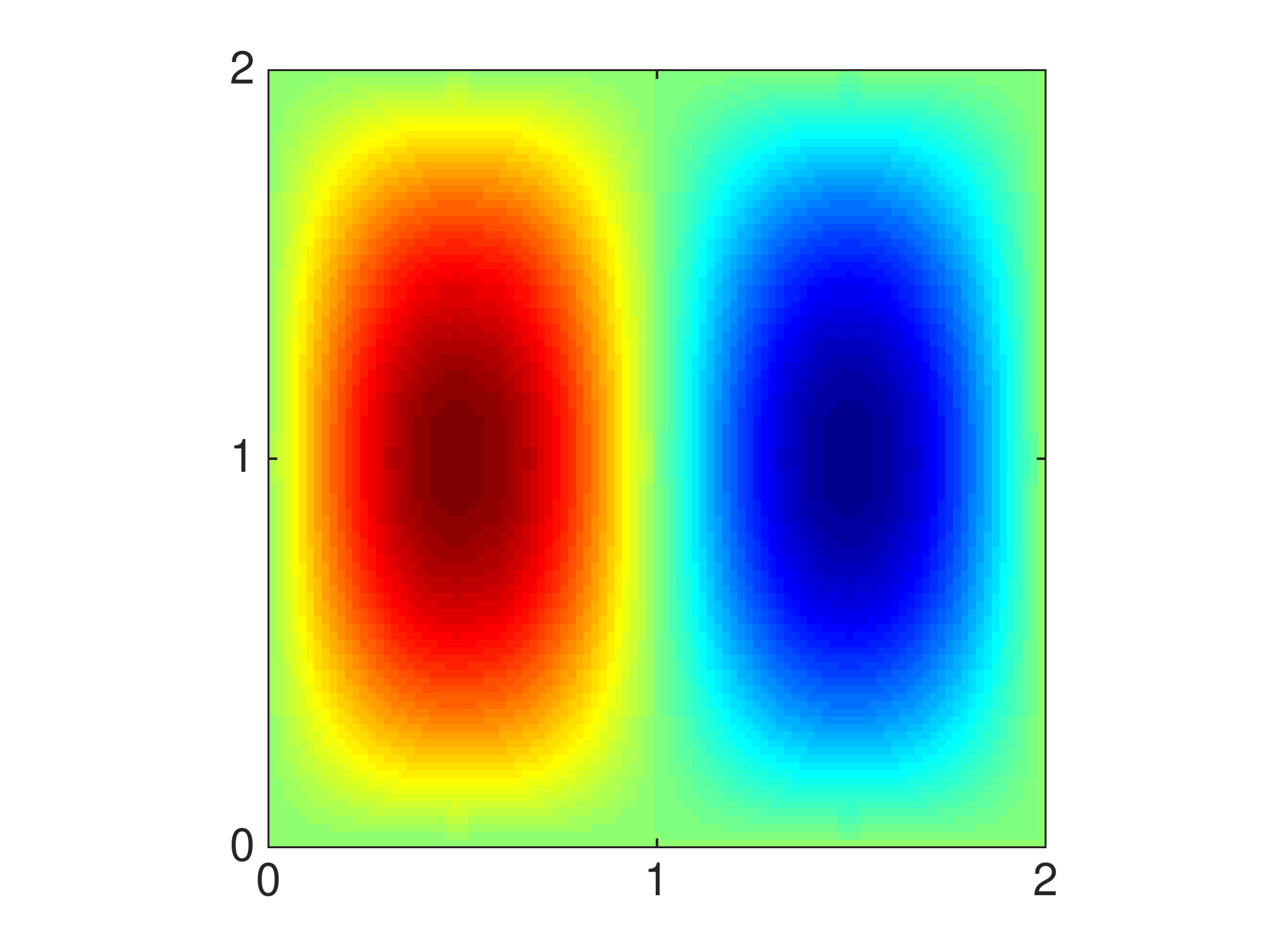,height=3.5cm,width=5cm,angle=0}
\psfig{figure=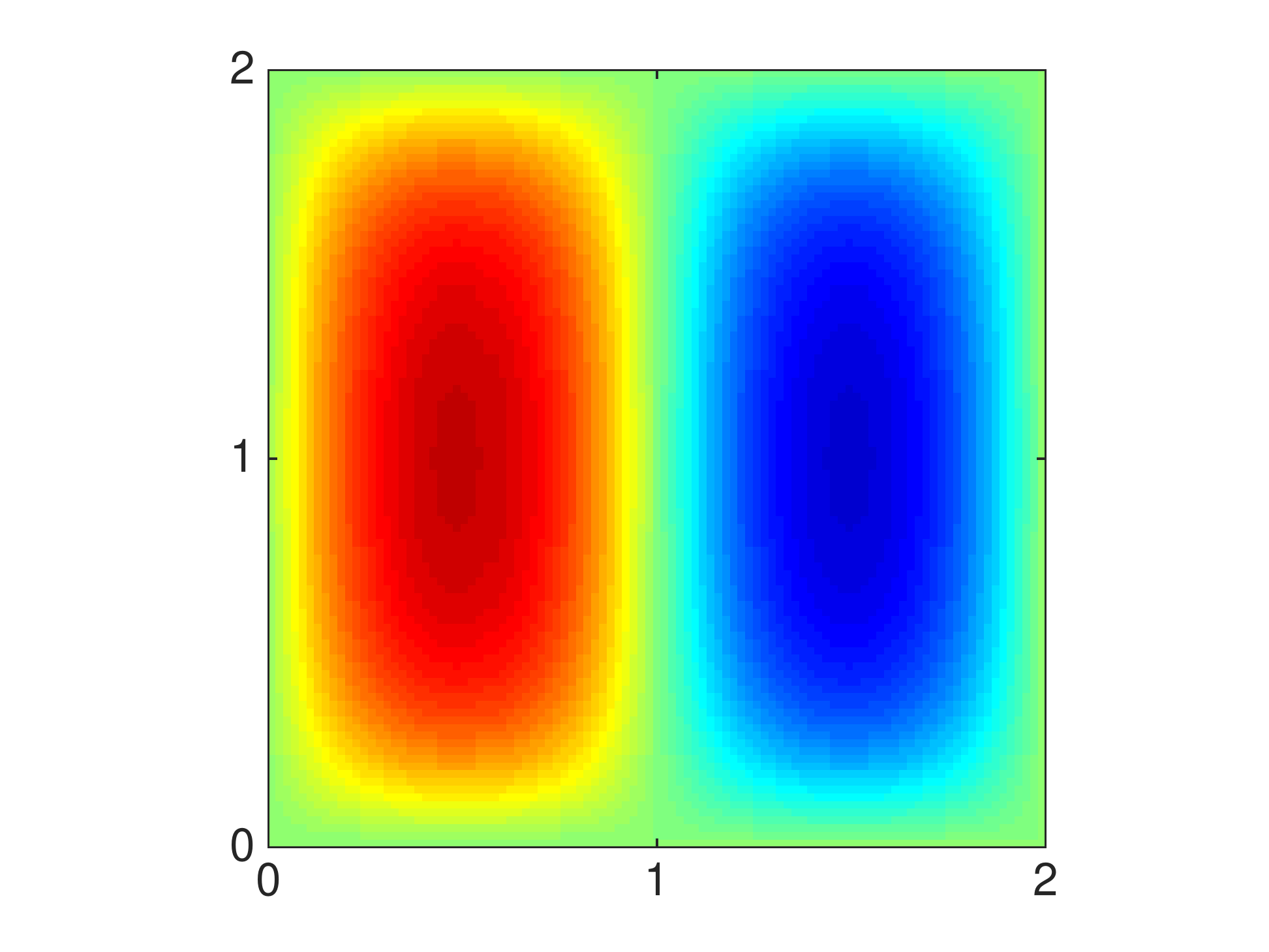,height=3.5cm,width=5cm,angle=0}
\psfig{figure=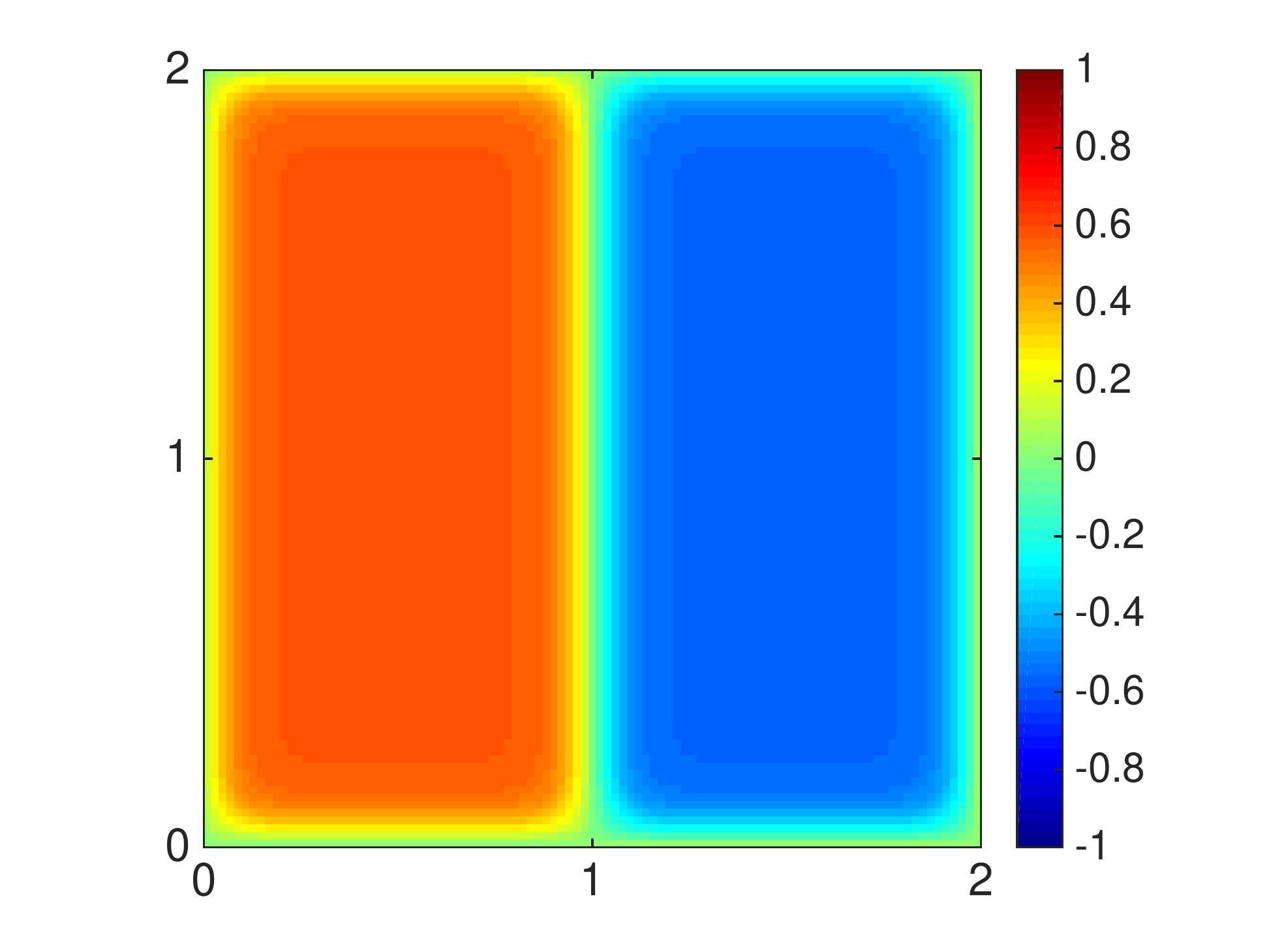,height=3.5cm,width=5cm,angle=0}}
\centerline{\psfig{figure=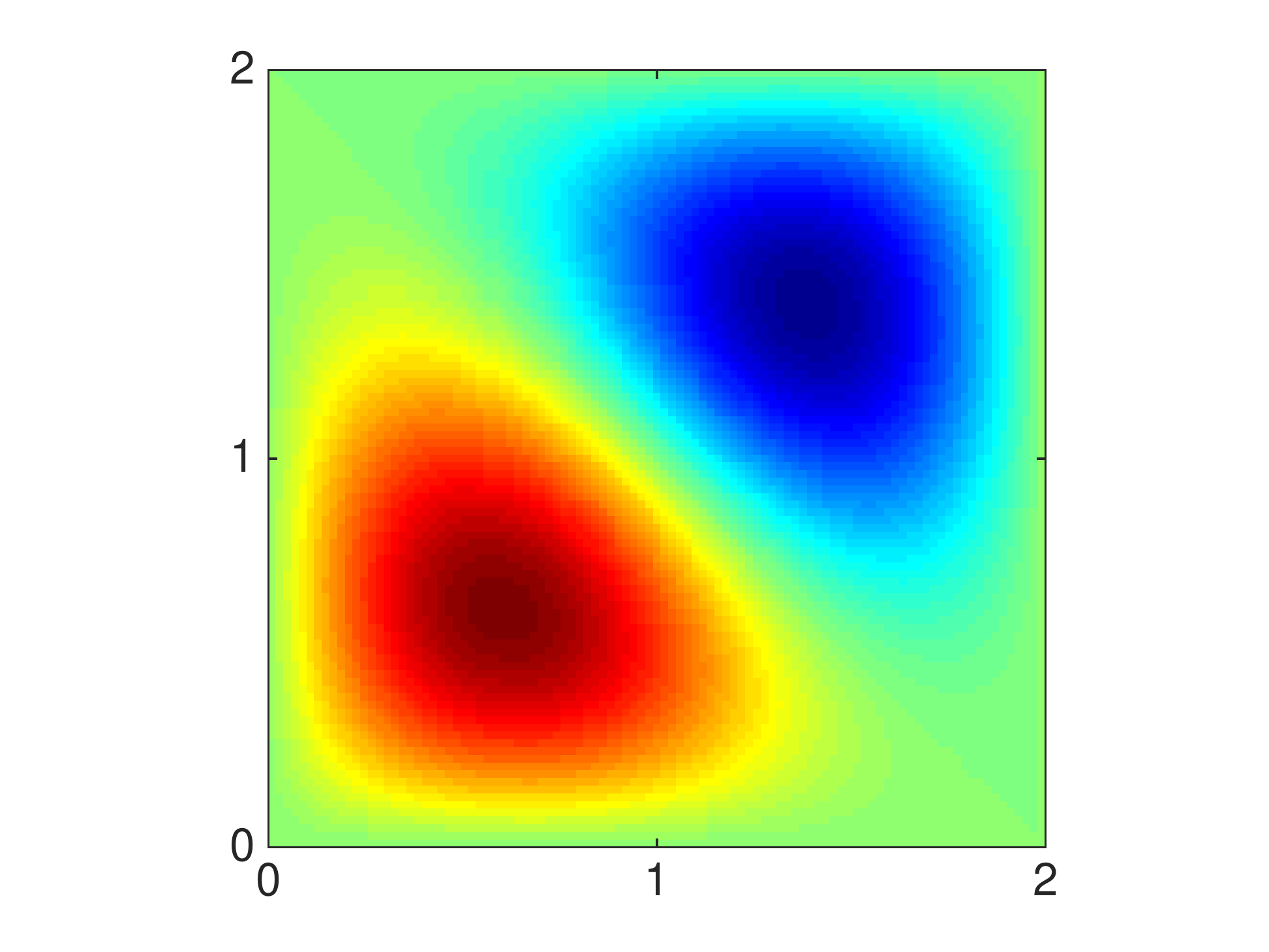,height=3.5cm,width=5cm,angle=0}
\psfig{figure=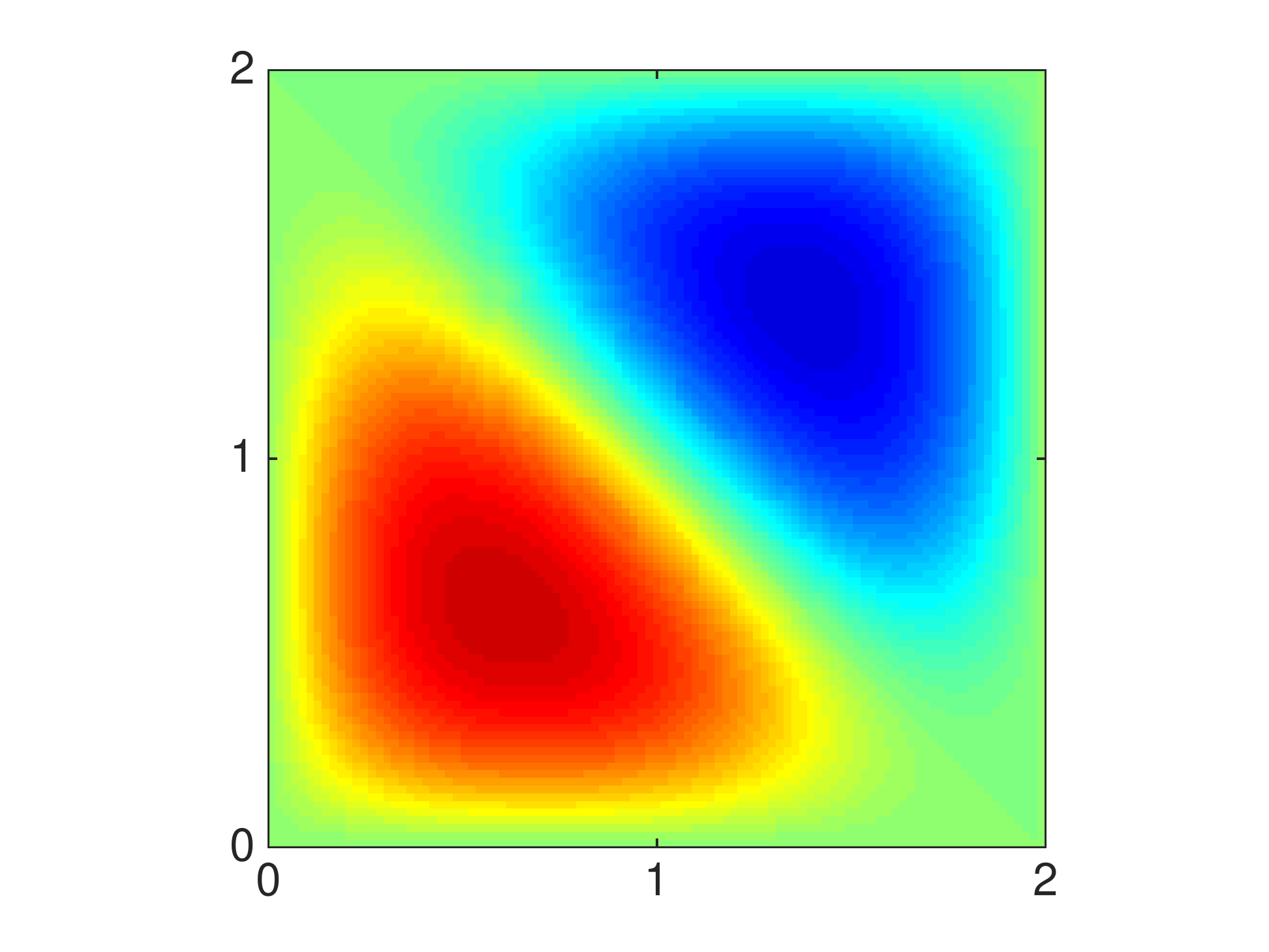,height=3.5cm,width=5cm,angle=0}
\psfig{figure=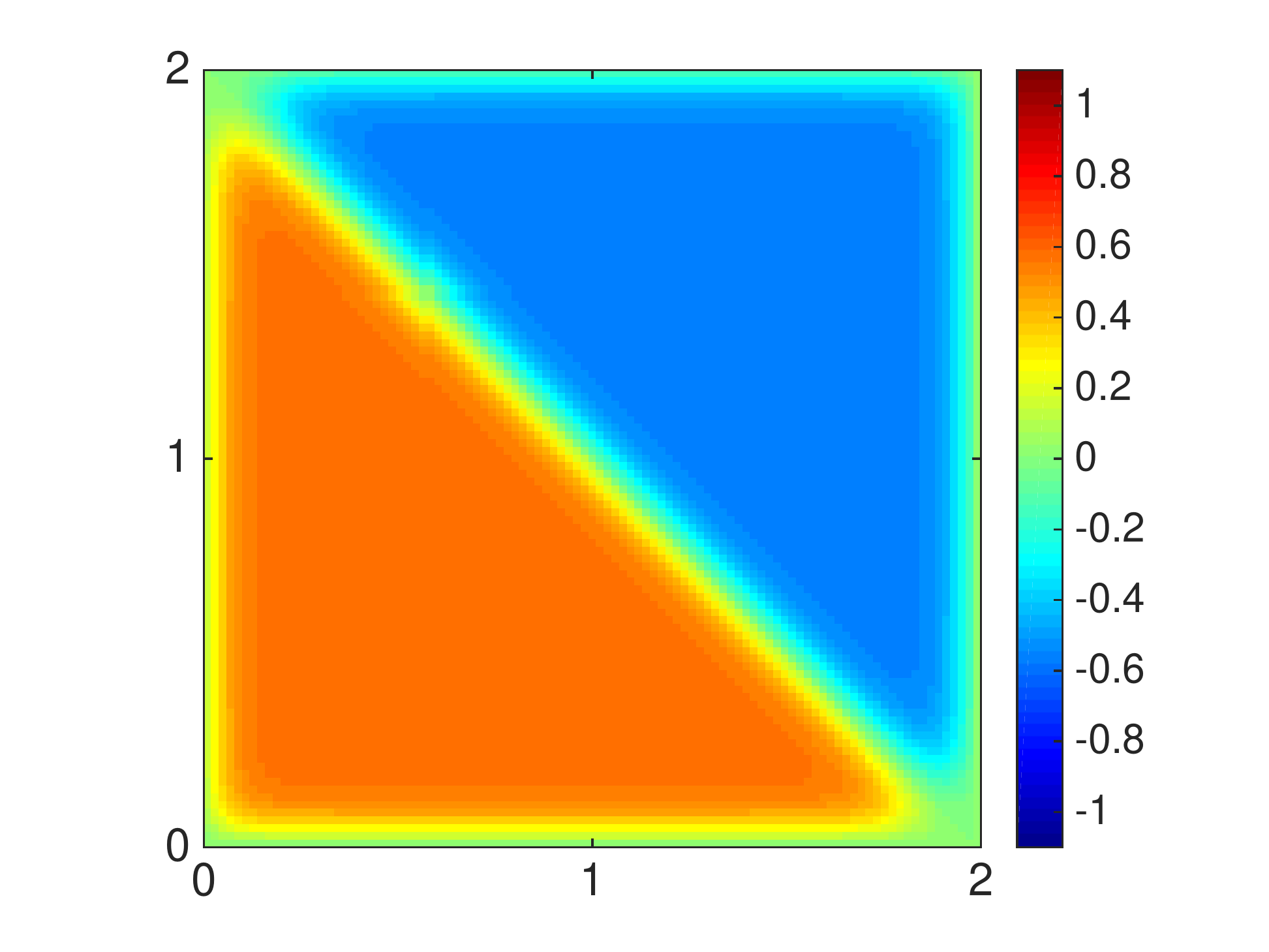,height=3.5cm,width=5cm,angle=0}}
\caption{Ground states $\phi_g^\beta$ (top row), first excited states -- vortex solution $|\phi_1^\beta=\phi_{1,v}^\beta|$ (second row), excited states in the $x_1$-direction $\phi_{1,x}^\beta$(fourth row) and excited states in the diagonal direction $\phi_{1,c}^\beta$ (fifth row) for $\beta=0$ (left column), $\beta=10$ (middle column) and $\beta=500$ (right column).
Here the phase of the vortex solution -- first excited state -- is displayed
in the third row.}
\label{fig:box2d_sol_degenerate}
\end{figure}

\begin{figure}[htb]
\centerline{\psfig{figure=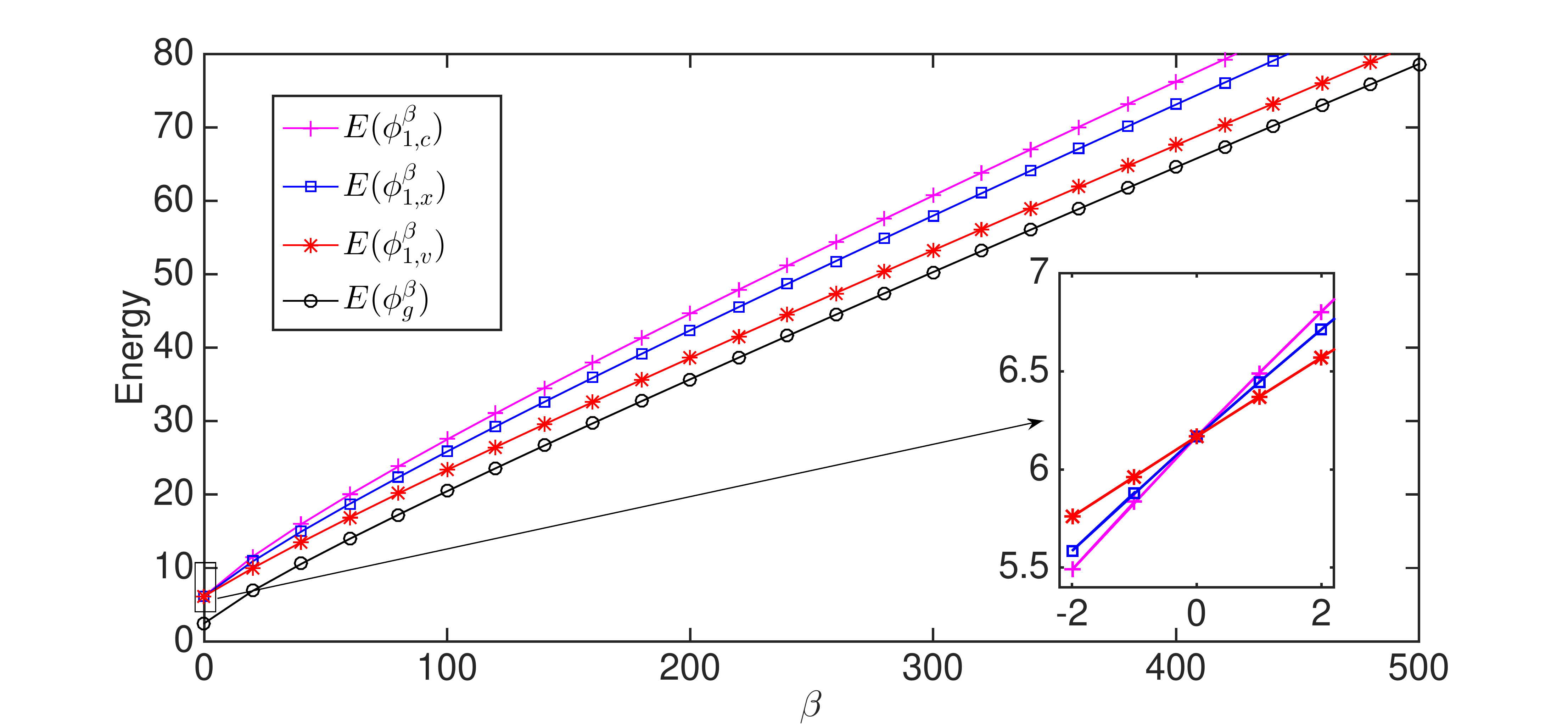,height=4cm,width=13cm,angle=0}}
\centerline{\psfig{figure=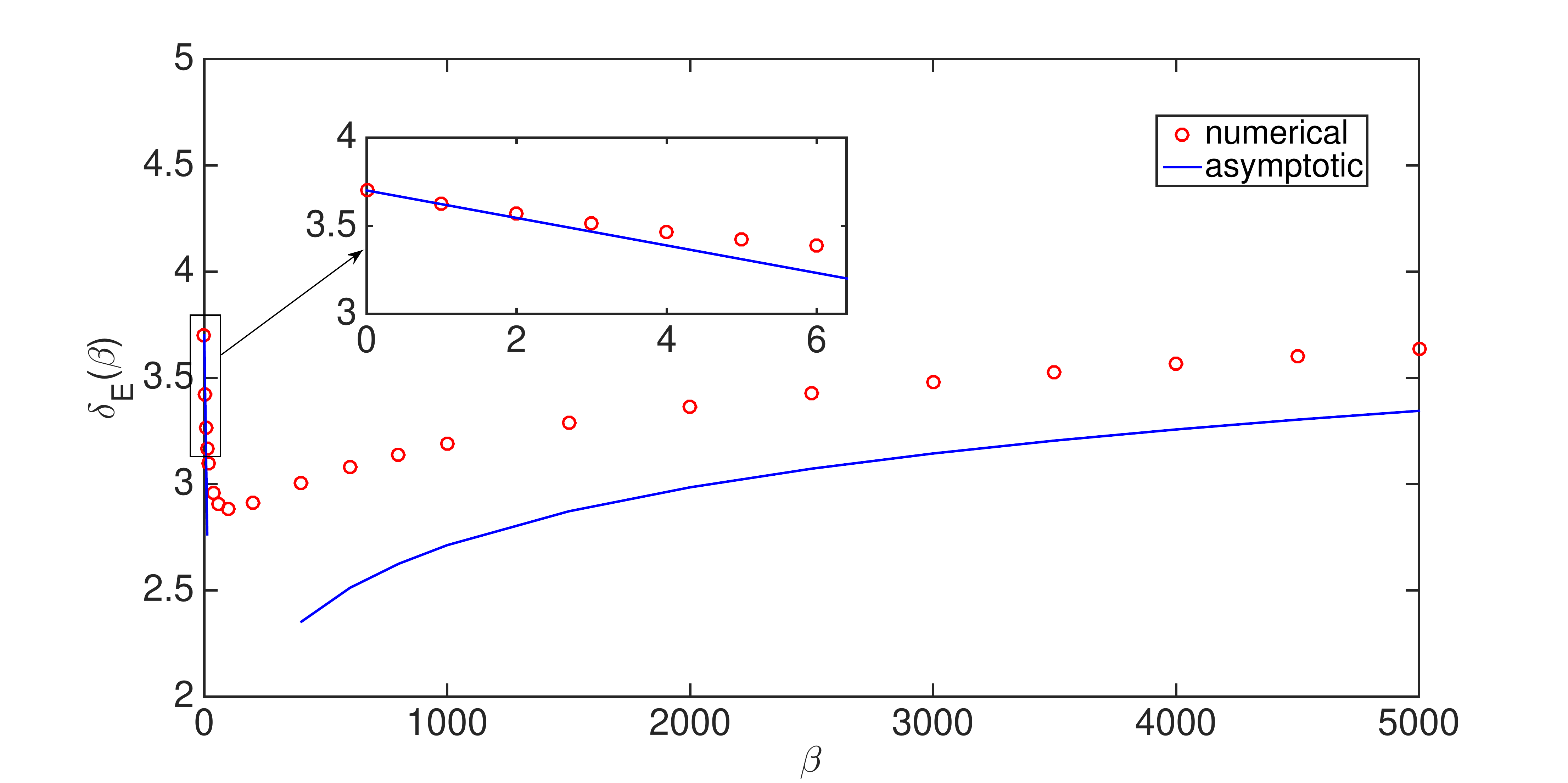,height=4cm,width=13cm,angle=0}}
\caption{Energy $E_g(\beta):=E(\phi_g^{\beta})<E_1(\beta):=E(\phi_1^\beta=\phi_{1,v}^{\beta})
<E_2(\beta):=E(\phi_{1,x_1}^{\beta}) =E(\phi_{1,x_2}^{\beta})<E_3(\beta):=E(\phi_{1,c}^{\beta})$ of
GPE in 2D under a box potential with $\Omega=(0,2)^2$ for different $\beta\ge0$ (top)  and the fundamental gaps in energy  $\delta_E(\beta)$
(bottom). Here a band crossing in energy
happens at $\beta=0$ for the excited states (cf. top). }
\label{fig:box_2d_degenerate}
\end{figure}

\begin{figure}[h!]
\centerline{\psfig{figure=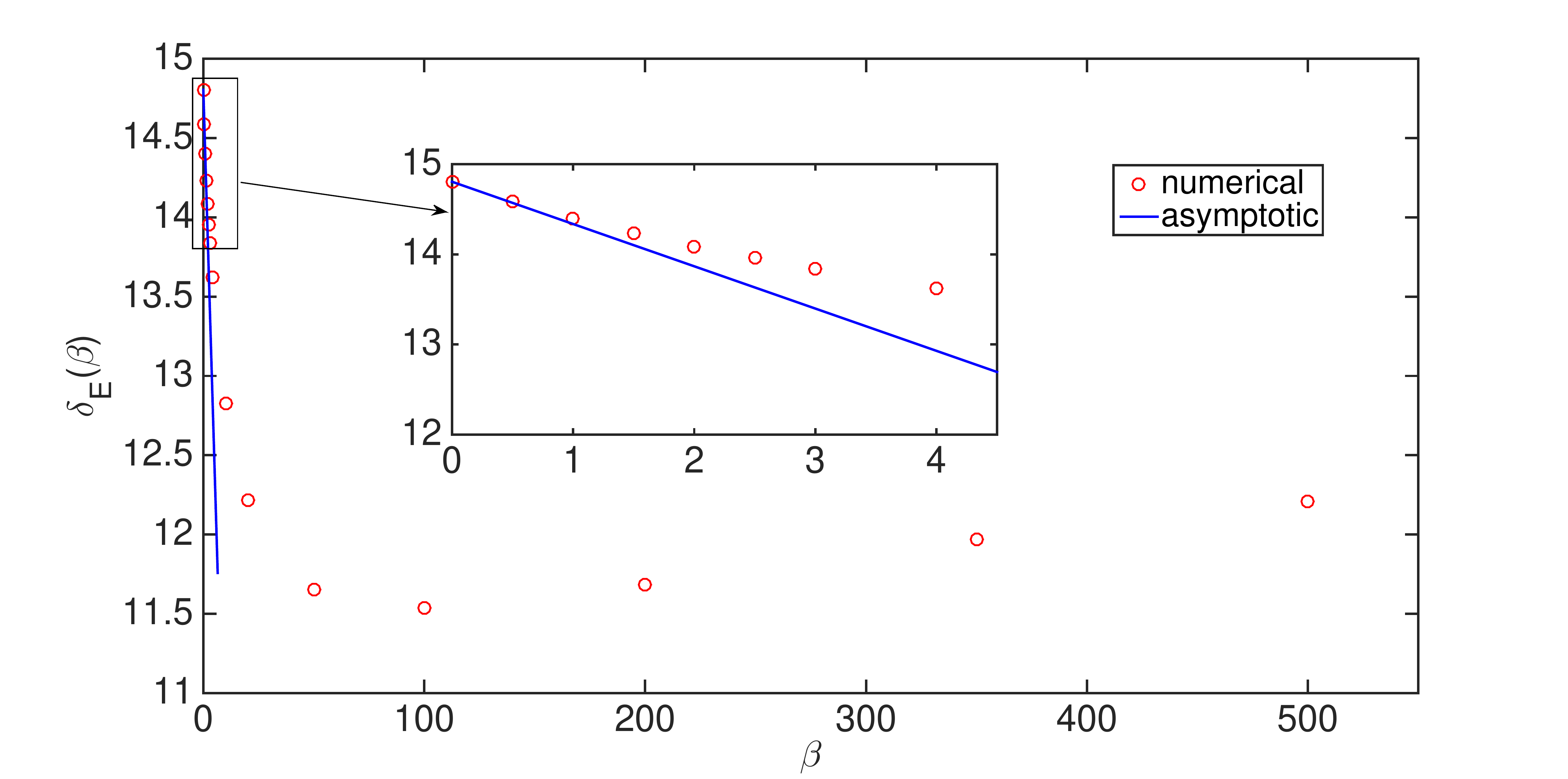,height=4cm,width=13cm,angle=0}}
\caption{The fundamental gaps in energy of GPE in 3D under a box potential with $\Omega=(0,1)^3$.}
\label{fig:box_3d_E_degenerate}
\end{figure}

\begin{figure}[t!]
\centerline{
\psfig{figure=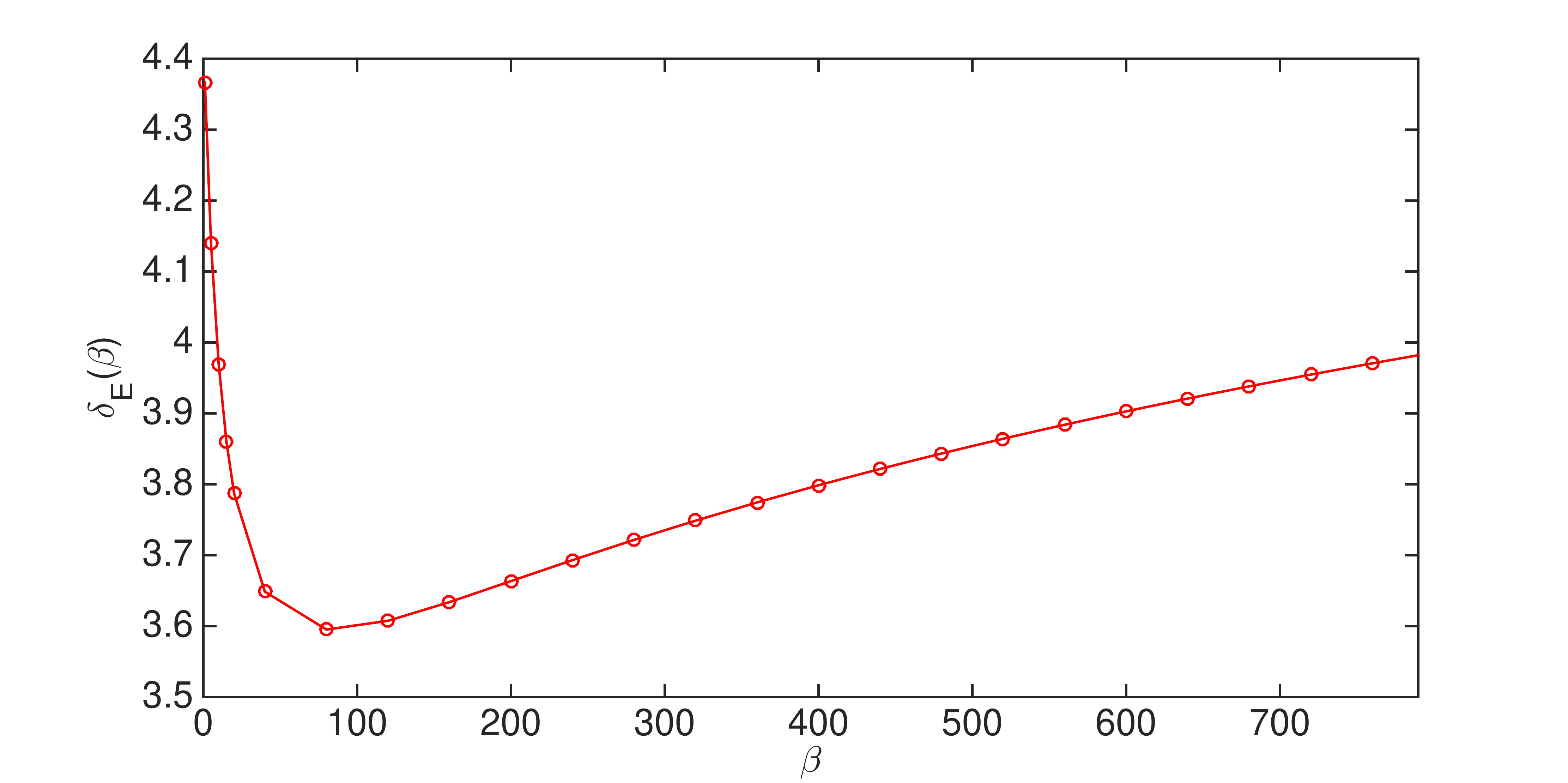,height=3.5cm,width=10cm,angle=0}}
\caption{
The fundamental gaps in energy of GPE under a box potential with $\Omega=B_1({\bf 0})$. It is obviously that the fundamental gap is larger than the lower bound proposed in the gap conjecture  (\ref{bgap976_degenerate}), which is $\frac{\pi^2}{2D^2}\approx1.234$.}
\label{fig:box2d_disk_degenerate}
\end{figure}

Based on the asymptotic results in Proposition \ref{asym:box_degen} and
the above numerical results as well as additional extensive numerical results not shown here for brevity \cite{Ruan},
we speculate the following gap conjecture.

\textbf{Gap conjecture}
(For GPE in 2D on a bounded domain with homogeneous Dirichlet BC in degenerate case)
Suppose $\Omega\subset\mathbb{R}^2$ is a convex bounded domain,
  the external potential $V(\mathbf{x})$ is convex and
   $\dim(W_1)\ge2$, we have
\begin{equation}\label{bgap976_degenerate}
\delta^{\infty}_E:=\inf_{\beta\ge0} \delta_E(\beta)\ge\frac{\pi^2}{2D^2},
\qquad\delta^{\infty}_{\mu}:=\inf_{\beta\ge0} \delta_\mu(\beta)\ge\frac{3\pi^2}{8D^2}.
\end{equation}



\section{Fundamental gaps of GPE in the whole space}\label{har}
\setcounter{equation}{0}
\setcounter{figure}{0}
In this section, we obtain asymptotically the fundamental gaps of the GPE \eqref{eq:eig} in the whole space under a harmonic potential and numerically under general potentials growing at least quadratically in the far field.
Based on the results, we formulate a novel gap conjecture for this case.
Here we take $\Omega={\mathbb R}^d$ and denote $V_h(\mathbf{x})=\frac{1}{2}\sum\limits_{j=1}^d\gm_j^2x_j^2$  satisfying
$0<\gm_1 \le  \gm_2 \le \dots \le \gm_d$.

\subsection{Nondegenerate case, i.e. $\dim(W_1)=1$}\label{sec:har_asym}
We first consider the special case by taking
$V(\mathbf{x})=V_h(\bx)$ satisfying $d=1$ or $\gm_1 <  \gm_2$ when $d\ge2$.
For simplicity, we define
\be\label{harct11}
B_0=\prod_{j=1}^d\sqrt{\frac{\gm_j}{2\pi}},\quad B_1=\frac{1}{2}\sum_{j=1}^d\gamma_j,\quad
B_2=\prod_{j=1}^d\gm_j, \quad C_d=
\begin{cases}
2, &d=1,\\
\pi, &d=2,\\
\frac{4\pi}{3}, &d=3.
\end{cases}
\ee

In this scenario, when $\beta=0$, all eigenfunctions can be
obtained via the Hermite functions \cite{BaoL,LZBao}. Thus the ground state $\phi_g^0(\bx)$
and the first excited state
$\phi_1^0(\bx)$ can be given explicitly as \cite{BaoL,LZBao}
\be\label{har00}
\phi_g^0(\mathbf{x})=\prod_{j=1}^d\left(\frac{\gm_j}{\pi}
\right)^{\frac{1}{4}}e^{-\frac{\gm_jx_j^2}{2}},
\quad  \phi_1^{0}(\mathbf{x})=\sqrt{2\gm_1}x_1\prod_{j=1}^d\left(\frac{\gm_j}{\pi}
\right)^{\frac{1}{4}}e^{-\frac{\gm_jx_j^2}{2}},\quad \bx\in{\mathbb R}^d.
\ee

\begin{lemma}\label{har:ground_weak}
In the weakly repulsive interaction regime, i.e. $0<\beta\ll1$, we have
\begin{align} \label{harsmall}
E_g(\beta)&=B_1+\frac{B_0}{2}\beta+o(\beta),\quad
\mu_g(\beta)=B_1+B_0\beta+o(\beta),\quad 0\le \beta\ll1,\\
\label{harsmall1}
E_1(\beta)&=\gm_1+B_1+\frac{3B_0}{8}\beta+o(\beta), \quad
\mu_1(\beta)=\gm_1+B_1+\frac{3B_0}{4}\beta+o(\beta).
\end{align}
\end{lemma}

\begin{proof}
When $0<\beta\ll1$, we can approximate the ground
state $\phi_g^\beta(\bx)$ and the first excited state $\phi_1^\beta(\bx)$ by $\phi_g^0(\bx)$ and
$\phi_1^0(\bx)$, respectively. Thus we have
\be\label{har01}
\phi_g^\beta(\mathbf{x})\approx \phi_g^0(\mathbf{x}),\qquad
\phi_1^{\beta}(\mathbf{x})\approx \phi_1^{0}(\mathbf{x}),
\qquad \bx\in{\mathbb R}^d.
\ee
Plugging \eqref{har01} into \eqref{def:mu} and \eqref{def:E}, after a detailed computation which is omitted here for brevity \cite{Ruan}, we can obtain \eqref{harsmall} and \eqref{harsmall1}.
\end{proof}

\begin{lemma}\label{har:ground_strong}
In the strongly repulsive interaction regime, i.e. $\beta\gg1$, we have
\begin{align}
\label{hartfm1}
\mu_g(\beta)&\approx \mu_g^{\rm TF}=\frac{1}{2}\left(\frac{(d+2)
 B_2\beta}{C_d}\right)^{\frac{2}{d+2}},\,
\mu_1(\beta)\approx \mu_1^{\rm MA}=\mu_g^{\rm TF}+\frac{\sqrt{2}}{2}\gm_1+o(1),\\
\label{hartfg1}
E_g(\beta)&=\frac{2+d}{4+d}\mu_g^{\rm TF}+o(1),\quad
E_1(\beta)=E_g(\beta)+\frac{\sqrt{2}}{2}\gm_1+o(1), \qquad \beta\gg1.
\end{align}
\end{lemma}

\begin{proof}
When $\beta\gg1$,
the ground and  first excited states can be approximated by
the TF approximations and/or uniformly accurate matched asymptotic
approximations.
For $d=1$ and $V(x)=\frac{\gamma^2 x^2}{2}$, these approximations have
been given explicitly and verified numerically in the literature \cite{Wz1,Wg1,BaoL,LZBao}, and the results can be extended to $d$ dimensions ($d=1,2,3$) as
\be \label{harg31}
\phi_g^{\beta}(\mathbf{x})\approx\phi_g^{\rm TF}(\mathbf{x})=
\sqrt{\frac{(\mu_g^{\rm TF}-V(\bx))_+}{\beta}}, \qquad \bx\in{\mathbb R}^d,
\ee
\be
\label{hare31}
\phi_1^{\beta}(\mathbf{x})\approx\phi_1^{\rm MA}(\mathbf{x})=
\begin{cases}
\sqrt{\frac{g_1(\bx)}{\beta}}+
\sqrt{\frac{g_2(\bx)}{\beta}}\left[\tanh(x_1\sqrt{g_2(\bx)})-1\right], &g_1(\bx)\ge0\&x_1\ge0,\\
-\sqrt{\frac{g_1(\bx)}{\beta}}+
\sqrt{\frac{g_2(\bx)}{\beta}}\left[1+\tanh(x_1\sqrt{g_2(\bx)})\right], &g_1(\bx)\ge0\&x_1<0,\\
0, &{\rm otherwise},
\end{cases}
\ee
where $(f)_+:=\max\{f,0\}$, $g_1(\bx)=\mu_1^{\rm MA}-\frac{1}{2}\sum_{j=1}^d\gm_j^2x_j^2$ and $g_2(\bx)=\mu_1^{\rm MA}-\frac{1}{2}\sum_{j=2}^d\gm_j^2x_j^2$, and
$\mu_g^{\rm TF}$ and $\mu_1^{\rm MA}$ can be obtained
via the normalization condition \eqref{norm}. Inserting \eqref{harg31} and \eqref{hare31} into \eqref{def:mu}, after a detailed computation which is omitted here for brevity \cite{Ruan}, we get \eqref{hartfg1}.
\end{proof}


From Lemmas \ref{har:ground_weak} and \ref{har:ground_strong}, we have asymptotic results for the fundamental gaps.
\begin{proposition}[For GPE under a harmonic potential in nondegenerate case]\label{asym:har}
When  $V(\mathbf{x})=V_h(\bx)$
satisfying $d=1$ or $\gm_1< \gm_2$ when $d\ge2$, i.e. GPE with a harmonic potential,
we have
\be \label{gaphar1}
\delta_E(\beta)=\begin{cases}
\gm_1-\frac{B_0}{8}\beta+o(\beta),\\
\frac{\sqrt{2}}{2}\gm_1+o(1),
\end{cases}
\delta_{\mu}(\beta)=\begin{cases}
\gm_1-\frac{B_0}{4}\beta+o(\beta),&0\le\beta\ll1,\\
\frac{\sqrt{2}}{2}\gm_1+o(1),&\beta\gg1.
\end{cases}
\ee

\end{proposition}

\begin{proof}  When $0\le \beta \ll1$, subtracting (\ref{harsmall})
from  (\ref{harsmall1}),  we obtain (\ref{gaphar1})
in this parameter regime. Similarly, when $\beta\gg1$,
we get the result by recalling \eqref{hartfm1} and \eqref{hartfg1}.
\end{proof}

\begin{remark} Similar to Lemma \ref{asym:boxbt}, when $\beta\gg1$, by performing asymptotic expansion
to the next order, we can obtain
\begin{align}\label{E1:higher_order}
&\delta_E(\beta)=
\frac{\sqrt{2}}{2}\gm_1+\frac{\gm_1^2(d+2)^{\frac{d}{d+2}}}{4}
\left(\frac{C_d}{B_2\beta}\right)^{\frac{2}{d+2}}+o(\beta^{-\frac{2}{d+2}}),\qquad \beta\gg1,\\
\label{mu1:higher_order}
&\delta_{\mu}(\beta)=
\frac{\sqrt{2}}{2}\gm_1+\frac{\gm_1^2d(d+2)^{-\frac{2}{d+2}}}{4}
\left(\frac{C_d}{B_2\beta}\right)^{\frac{2}{d+2}}+o(\beta^{-\frac{2}{d+2}}),\qquad\beta\gg1.
\end{align}
\end{remark}

\begin{figure}[htb]
\centerline{\psfig{figure=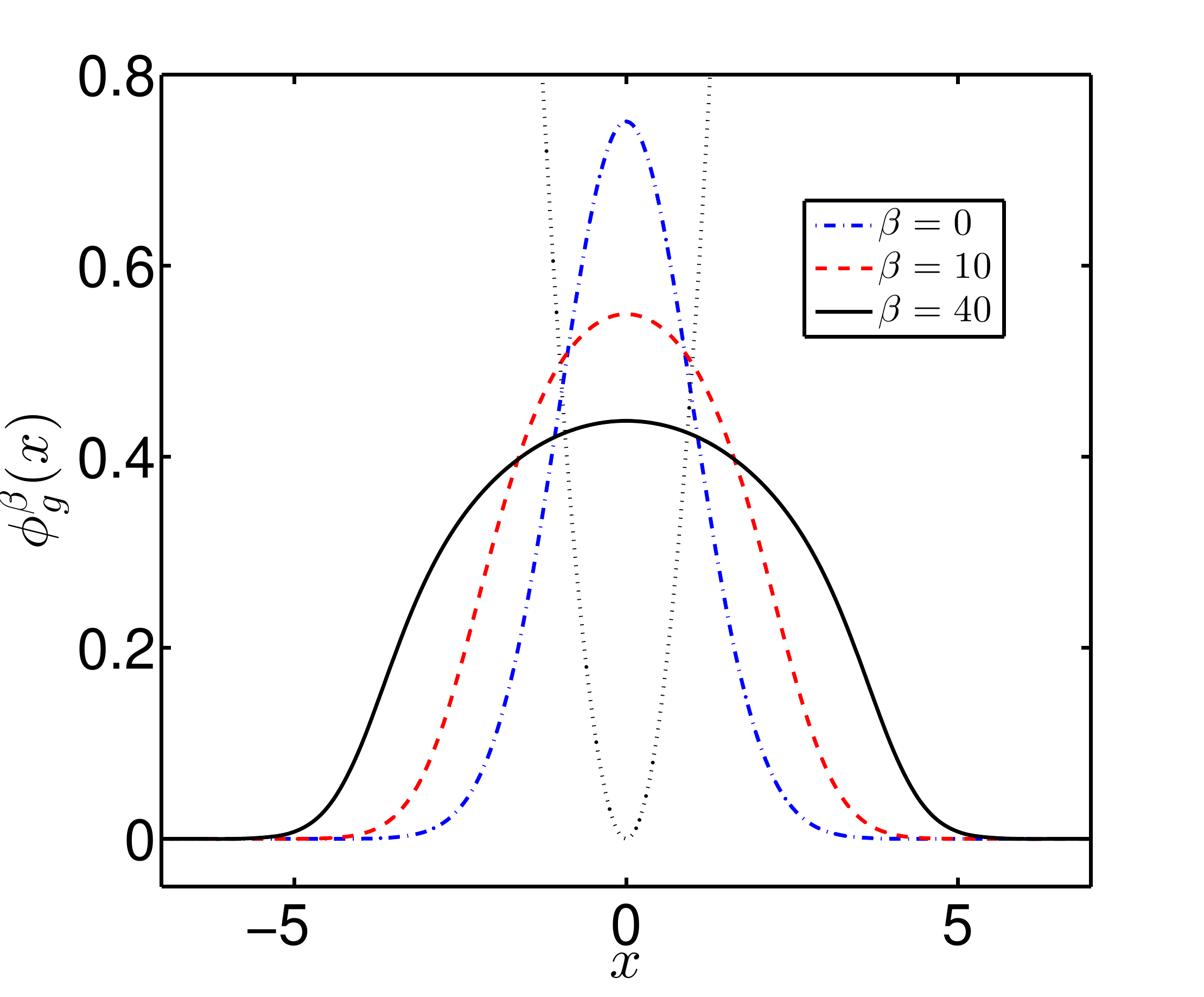,height=4cm,width=6.5cm,angle=0}
\psfig{figure=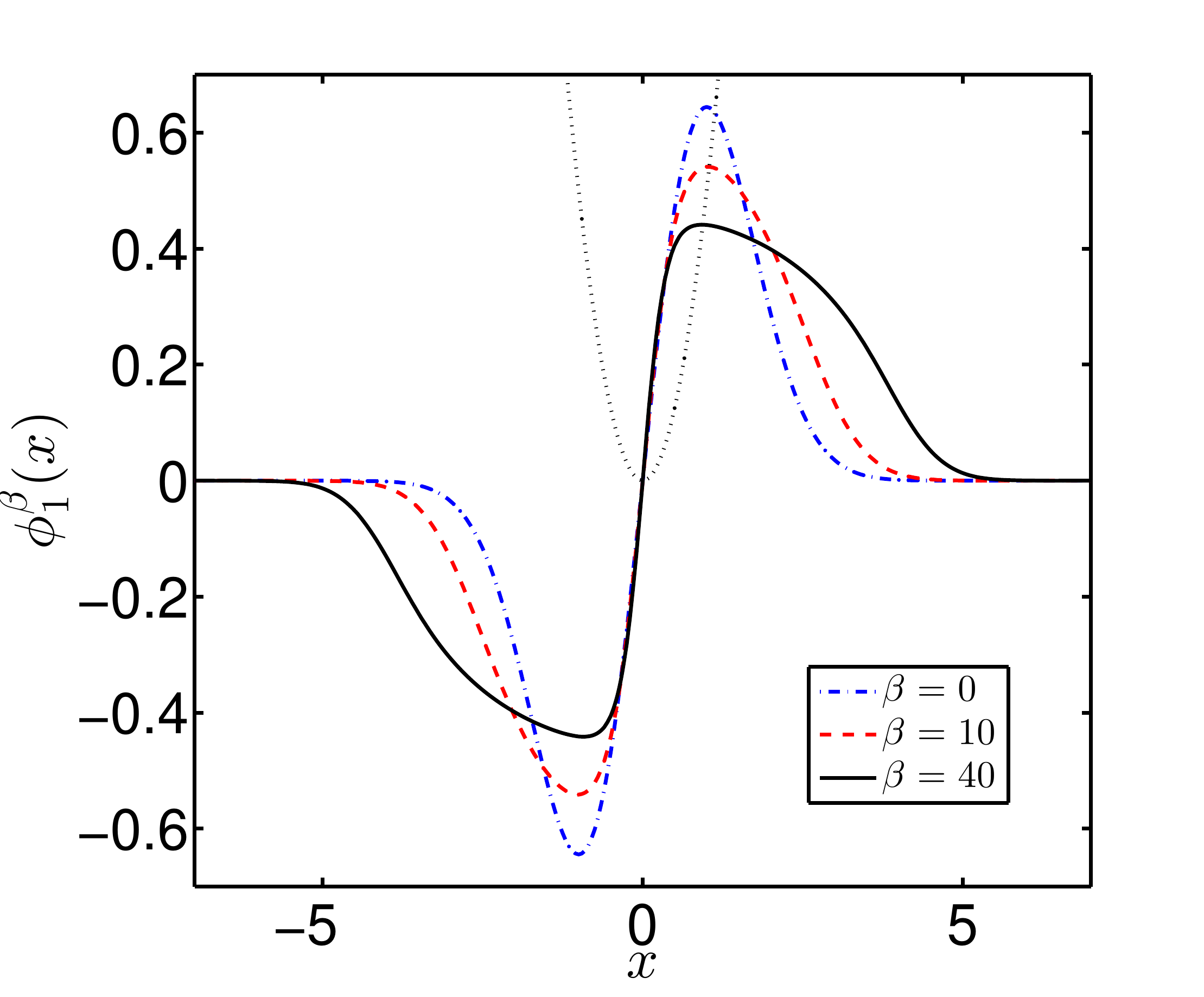,height=4cm,width=6.5cm,angle=0}}
\caption{Ground states (left) and first excited states (right) of GPE in 1D with a
harmonic potential $V(x)=x^2/2$ (dot line) for different $\beta\ge0$.}
\label{fig:har_1d_sol}
\end{figure}

\begin{figure}[htb]
\centerline{\psfig{figure=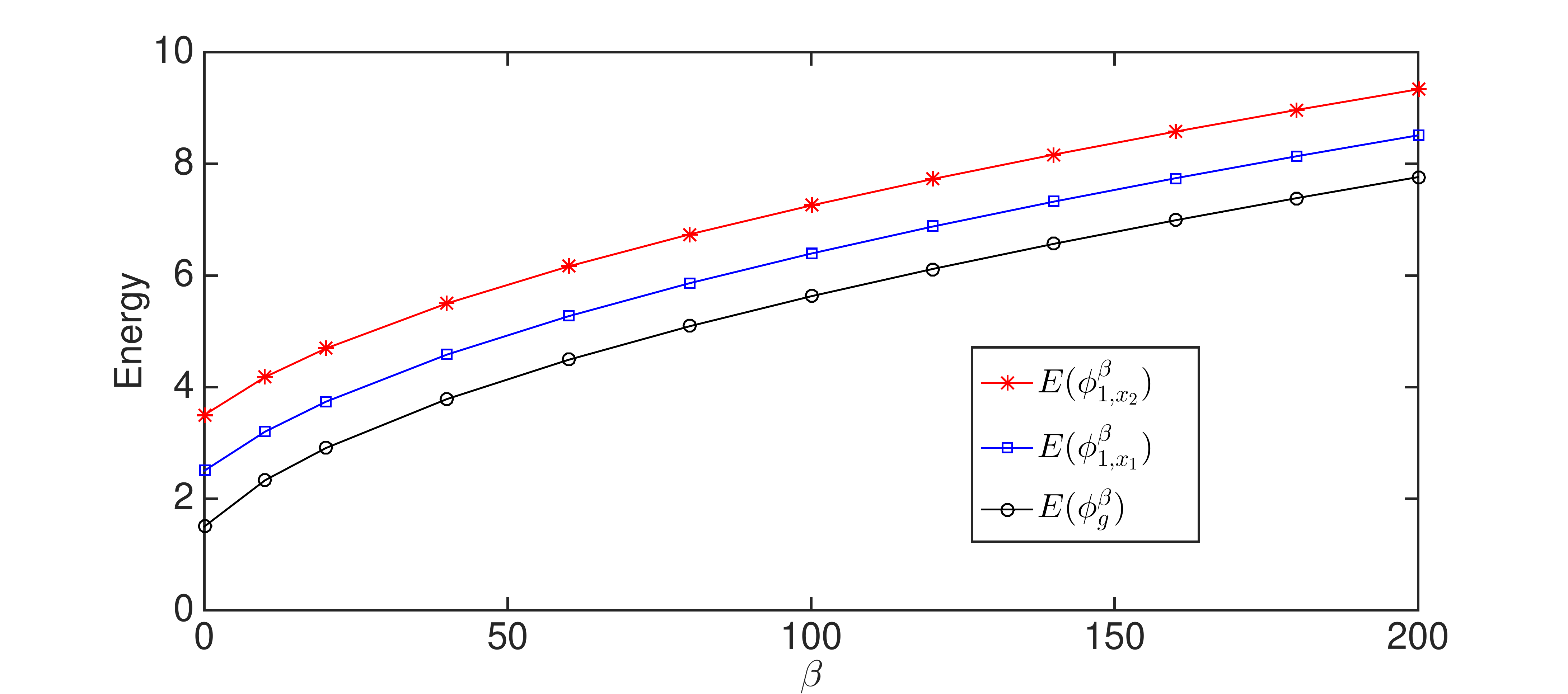,height=4cm,width=13cm,angle=0}}
\caption{Energy $E_g(\beta):=E(\phi_g^{\beta})<E_1(\beta):=E(\phi_1^\beta=\phi_{1,x}^{\beta}) <E_2(\beta):=E(\phi_{1,y}^{\beta})$ of  GPE in 2D under a harmonic potential with $\gm_1=1<\gm_2=2$ for different $\beta\ge0$.
}
\label{fig:har_dE_general_2d}
\end{figure}

\begin{figure}[htb]
\centerline{\psfig{figure=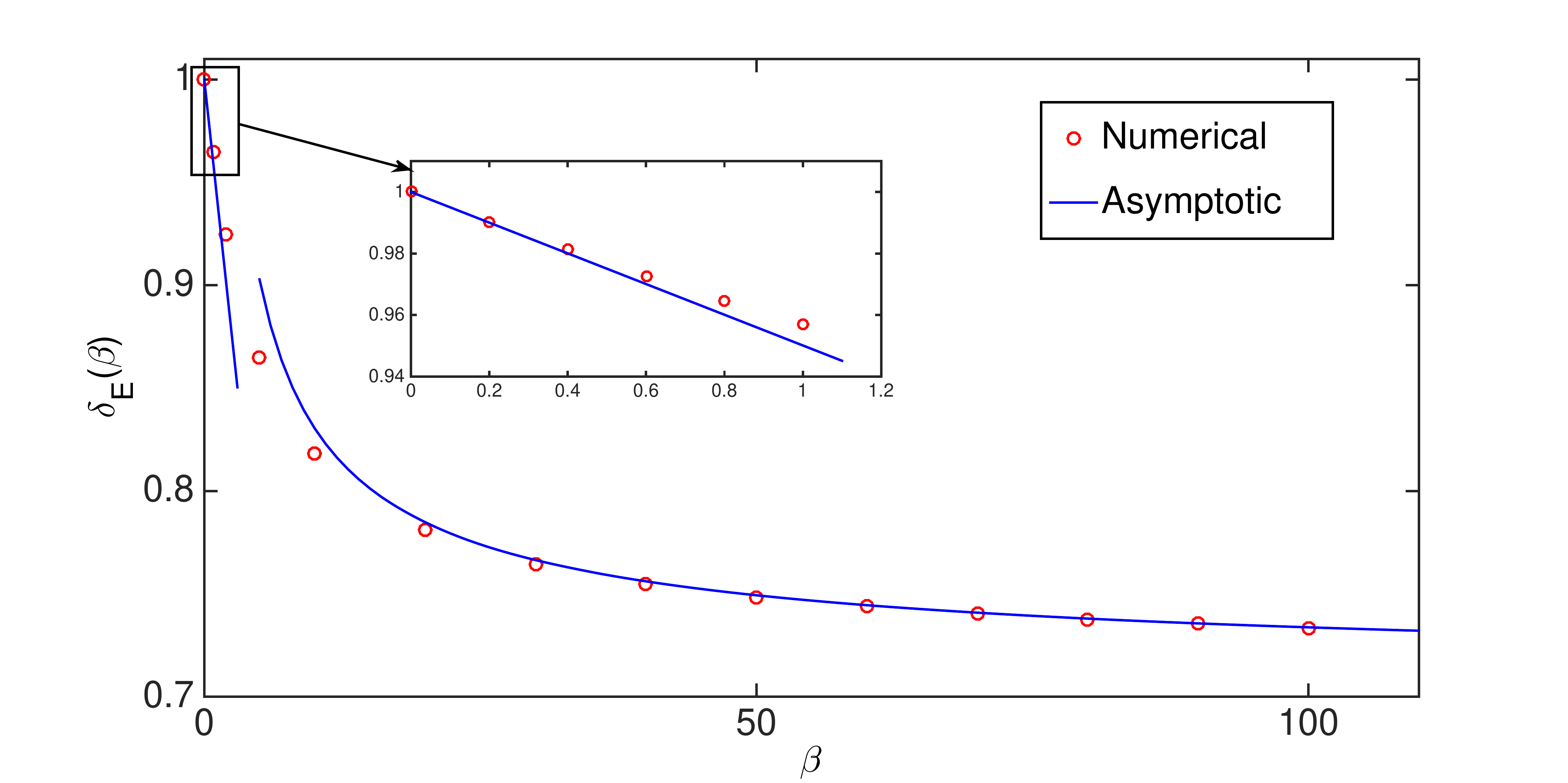,height=4cm,width=13cm,angle=0}}
\centerline{\psfig{figure=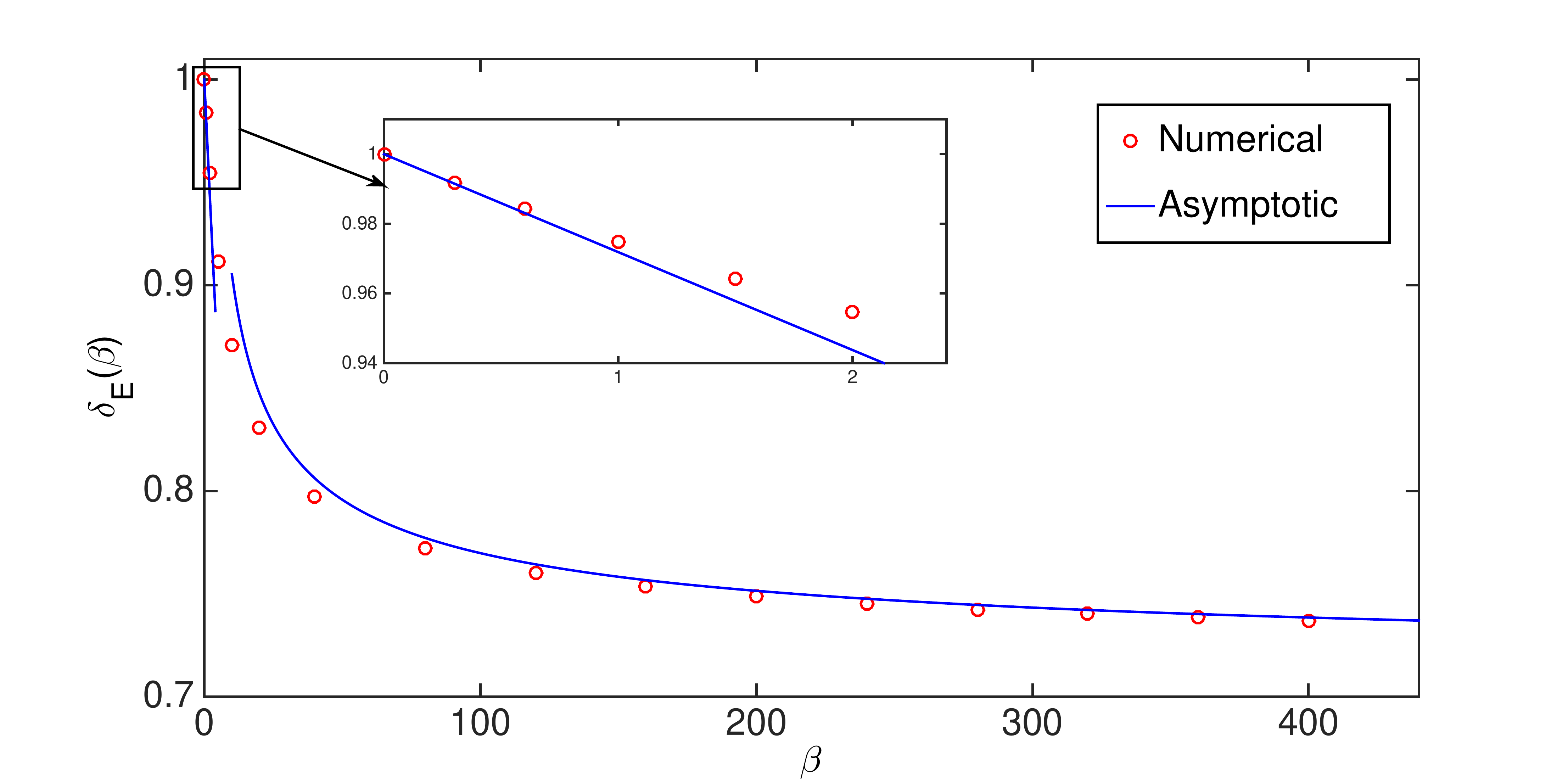,height=4cm,width=13cm,angle=0}}
\centerline{\psfig{figure=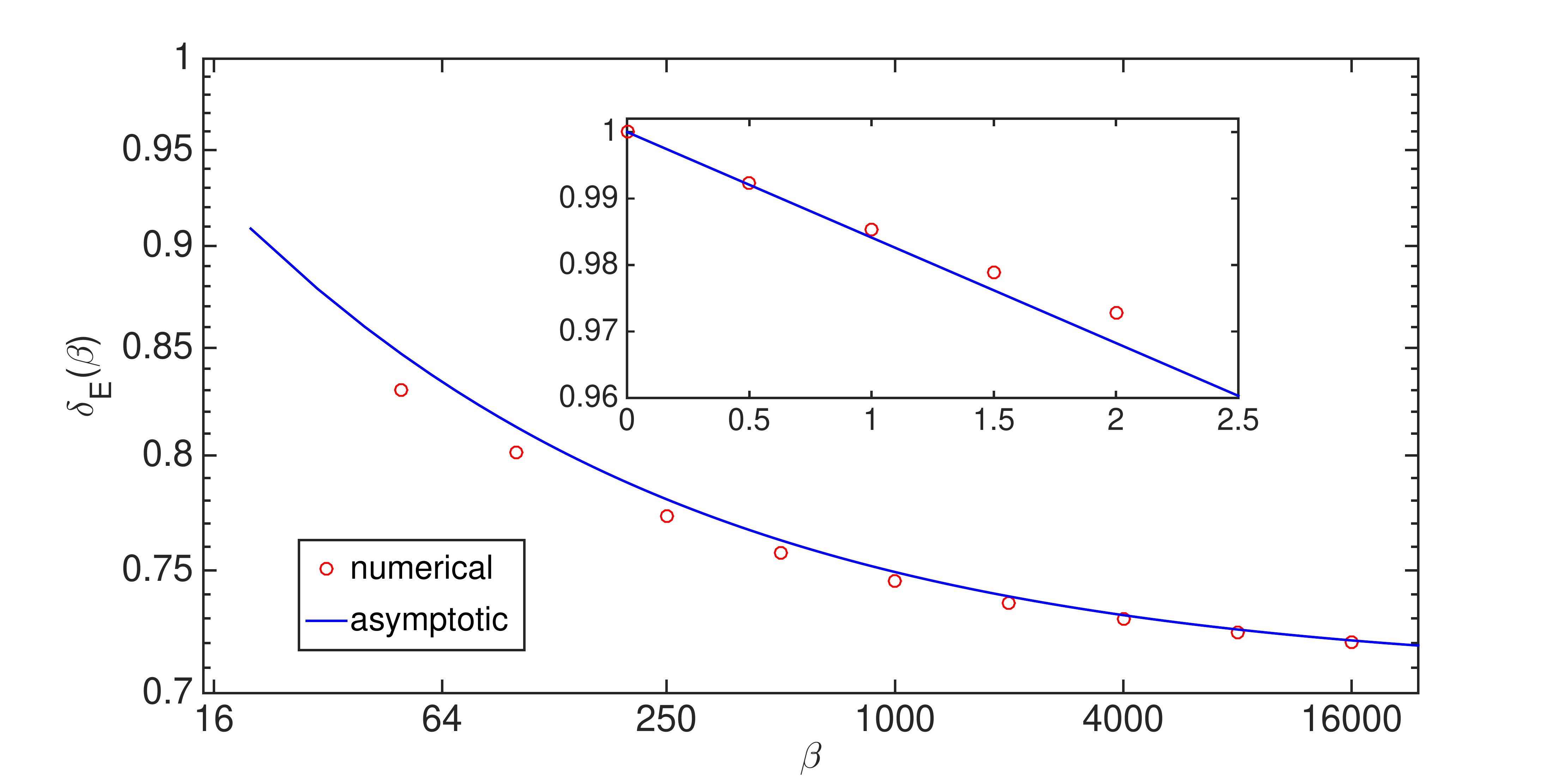,height=4cm,width=13cm,angle=0}}
\caption{Fundamental gaps in energy of GPE with a harmonic potential
in 1D with $\gm_1=1$ (top), in 2D with $\gm_1=1<\gm_2=2$ (middle),
and in 3D with $\gm_1=1<\gm_2=\gm_3=2$ (bottom).}
\label{fig:har_dE_asym}
\end{figure}

Again, to verify numerically our asymptotic results in Proposition \ref{asym:har},
Fig.~\ref{fig:har_1d_sol} shows the ground and first excited states of
GPE in 1D with $\gm_1=1$
for different $\beta\ge0$,  which are obtained numerically \cite{Bao_comp1,comp_gf,Bao2013,Wz1}.
Fig. \ref{fig:har_dE_general_2d} shows energy of the ground state,
first excited state, i.e. excited state in the $x_1$-direction,  and excited states in the $x_2$-direction
and Fig.~\ref{fig:har_dE_asym} depicts fundamental gaps in energy obtained numerically
and asymptotically (cf. Eqs. (\ref{E1:higher_order}), (\ref{mu1:higher_order}) and (\ref{gaphar1}))
in 1D, 2D and 3D.
From Fig. \ref{fig:har_dE_asym}, we can see that the asymptotic results in Proposition \ref{asym:har}
are very accurate in both weakly repulsive interaction regime, i.e. $0\le \beta\ll1$,
and strongly repulsive interaction regime, i.e. $\beta\gg1$. In addition,
our numerical results suggest that both $\delta_E(\beta)$ and $\delta_\mu(\beta)$
are decreasing functions for $\beta\ge0$ (cf. Fig. \ref{fig:har_dE_asym}).


\begin{figure}[htb]
\centerline{\psfig{figure=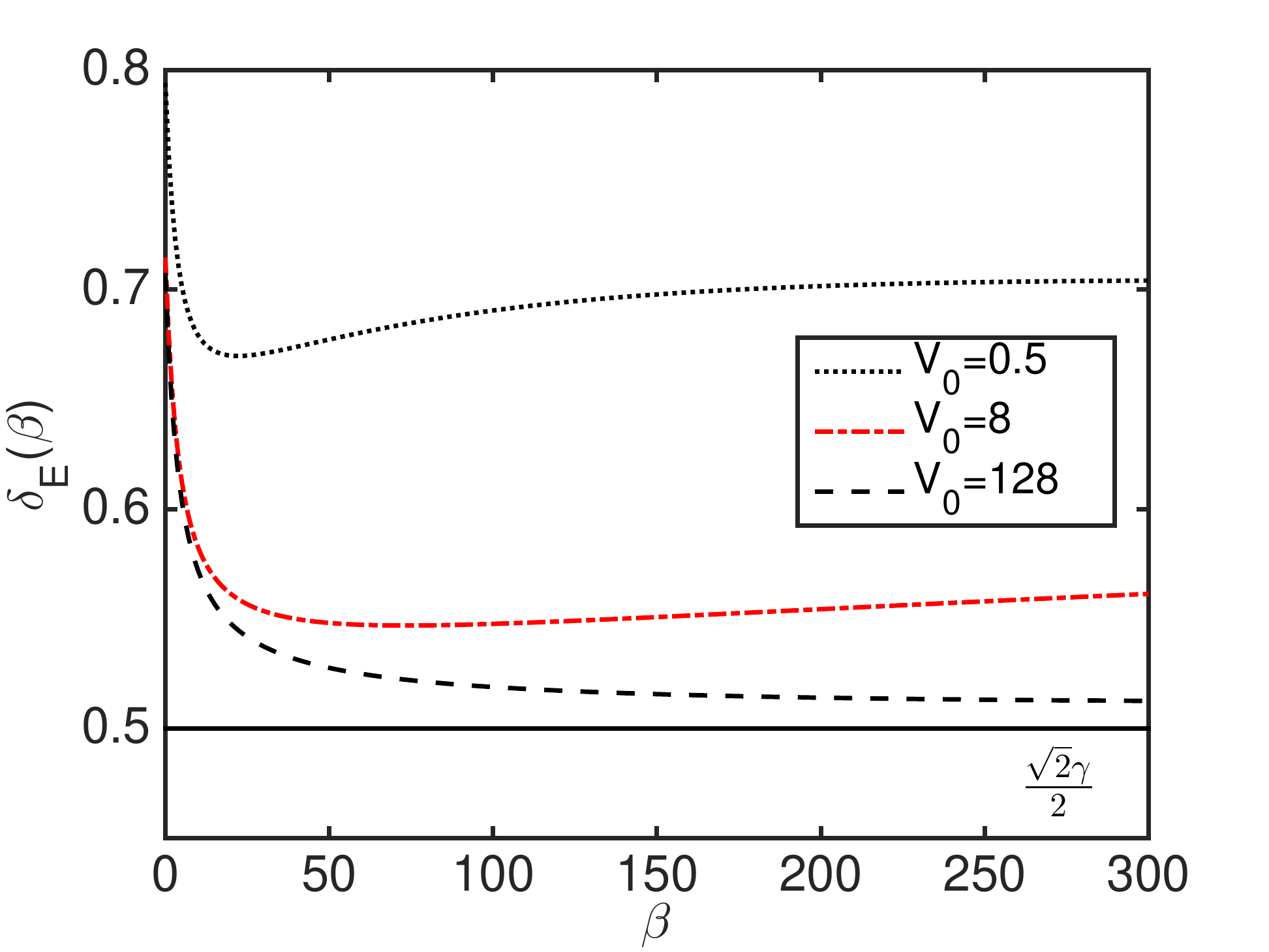,height=4cm,width=6.5cm,angle=0}
\psfig{figure=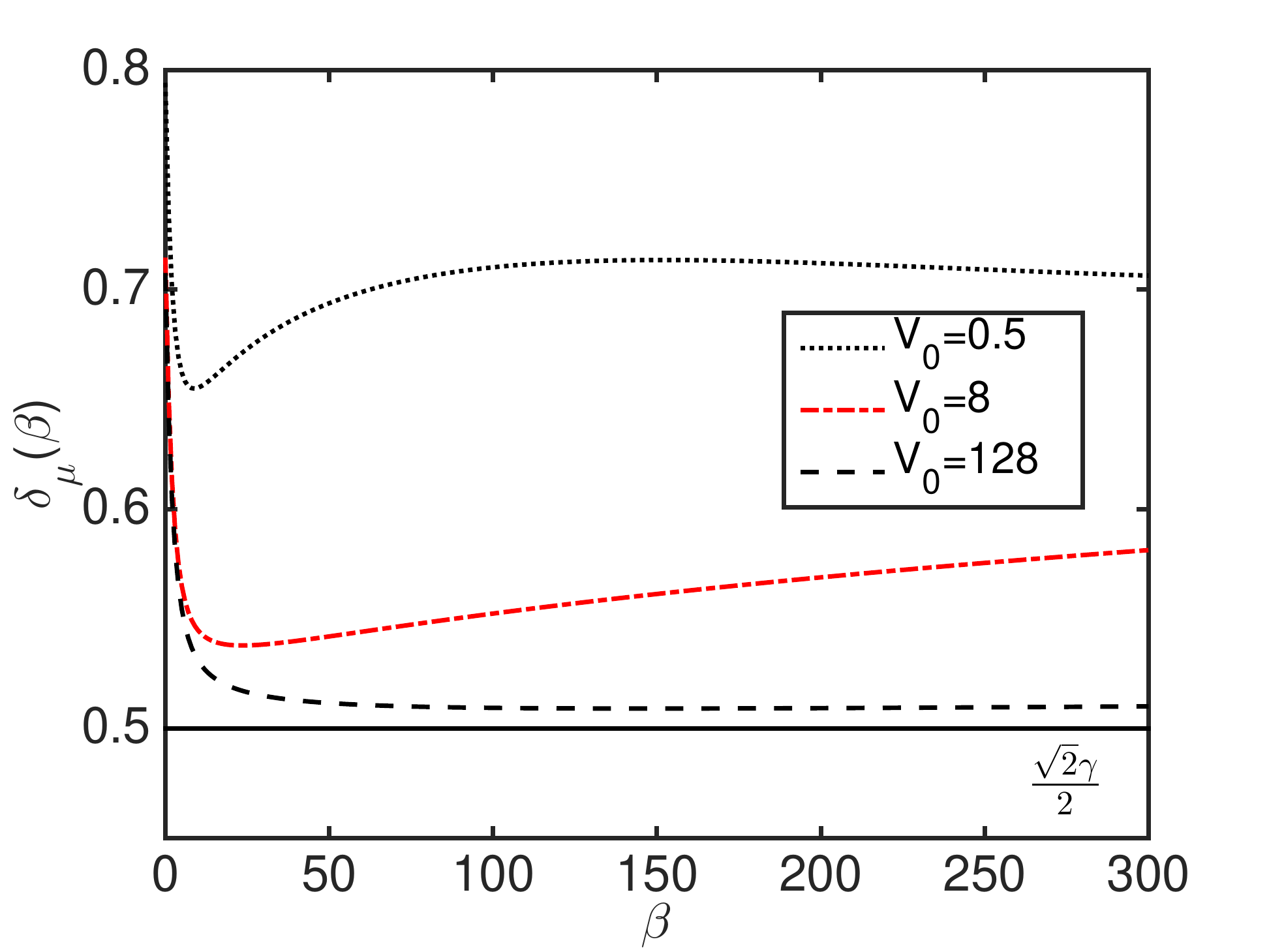,height=4cm,width=6.5cm,angle=0}}
\caption{Fundamental gaps in energy (left) and chemical potential (right)
of GPE in 1D with $V(x)=\frac{x^2}{2}+V_0\cos(kx)$ satisfying $V_0k^2=0.5$
for different $\beta$, $V_0$ and $k$.}
\label{fig:test_har_dE_dmu}
\end{figure}

Again, for general external potentials, the ground and first excited states as well as their energy and chemical potential
can be computed numerically \cite{Bao_comp1,comp_gf,Bao2013,Wz1}.
  Fig. \ref{fig:test_har_dE_dmu}
 depicts fundamental gaps
in energy and chemical potential of GPE in 1D with $V(x)=\frac{x^2}{2}+V_0\cos(kx)$ for different $\beta$, $V_0$ and $k$,
and
Fig. \ref{fig:test_genhar_dE} shows the fundamental gaps of GPE in 1D
with different  convex trapping potentials growing at least quadratically in the far field
 for different $\beta\ge0$.


\begin{figure}[htb]
\centerline{\psfig{figure=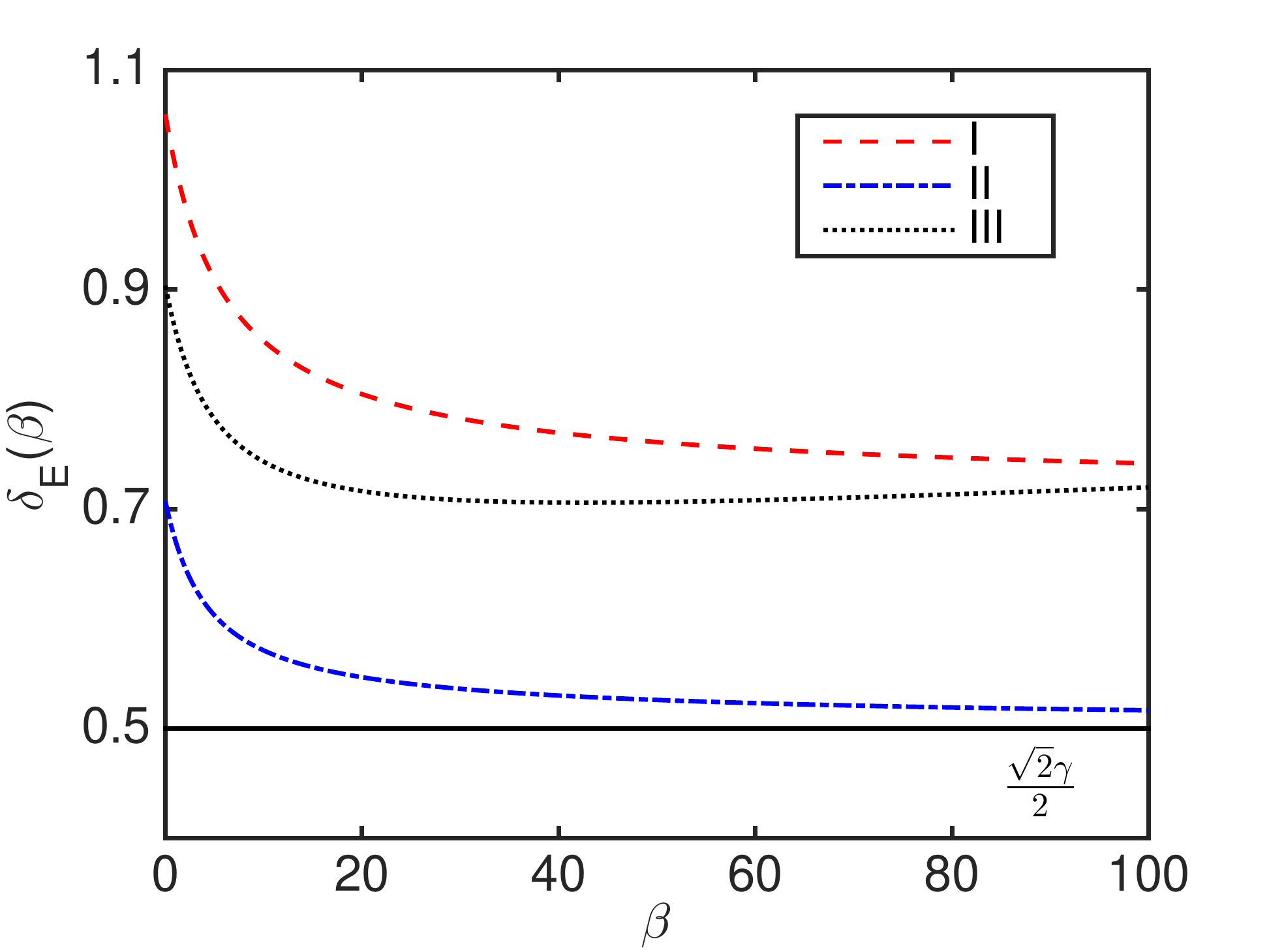,height=4cm,width=6.5cm,angle=0}
\psfig{figure=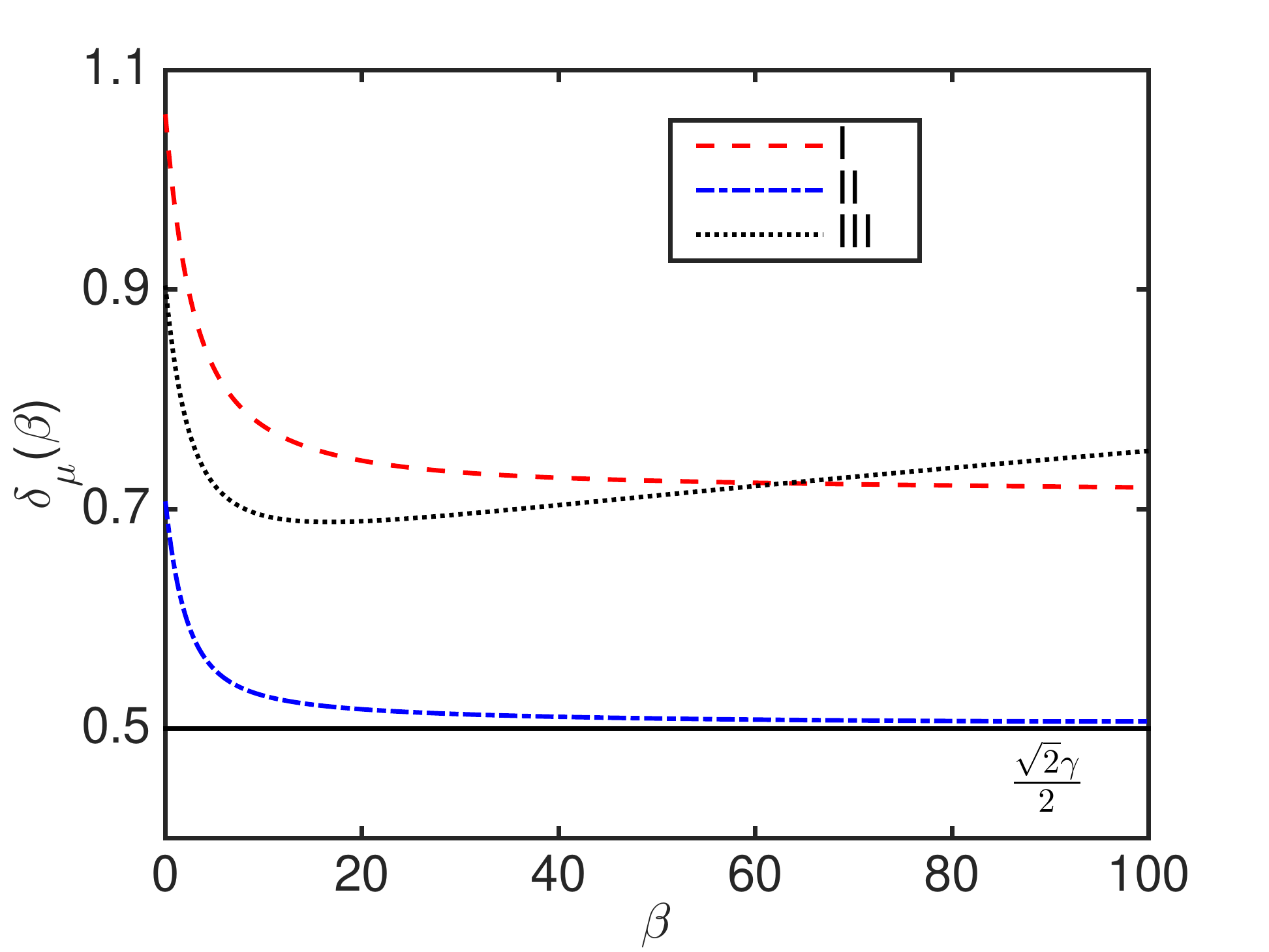,height=4cm,width=6.5cm,angle=0}}
\caption{Fundamental gaps in energy (left) and chemical potential (right)
of GPE in 1D with (I) $V(x)=\frac{x^2}{2}+0.5\sin(x)$, (II) $V(x)=\frac{x^2}{4}-x$,  (III) $V(x)=\frac{x^2}{4}+\frac{x^4}{100}+x$
for different $\beta\ge0$. }
\label{fig:test_genhar_dE}
\end{figure}


Based on the asymptotic results in Proposition \ref{asym:har} and the above numerical results as well as additional extensive numerical results not shown here for brevity \cite{Ruan},
we speculate the following gap conjecture.


\textbf{Gap conjecture}
(For GPE in whole space in nondegenerate case)
Suppose  the external potential $V(\mathbf{x})$ satisfies
$D^2V(\bx)\ge\gm_v^2I_d$ for $\bx\in{\mathbb R}^d$ with $\gm_v>0$  a constant and  $\dim(W_1)=1$, we have
\begin{equation}\label{bgap977}
\delta^{\infty}_E:=\inf_{\beta\ge0} \delta_E(\beta)\ge\frac{\sqrt{2}}{2}\gm_v,
\qquad\delta^{\infty}_{\mu}:=\inf_{\beta\ge0} \delta_\mu(\beta)\ge\frac{\sqrt{2}}{2}\gm_v.
\end{equation}

\subsection{Degenerate case, i.e. $\dim(W_1)\ge2$}\label{sec:har_asym_degenerate}
We first consider a special case by taking $V(\mathbf{x})=V_h(\bx)$ satisfying  $d\ge2$ and $\gm_1=\gm_2:=\gm$. In this case, the approximations to the ground states and their energy and chemical potential are the same as those in the previous subsection by letting $\gm_2 \to\gm_1 = \gm$. Therefore, we only need to focus on the approximations to the first excited states, which are completely different with those in the non-degenerate case.

\begin{lemma}\label{har:ground_weak_degenerate2}
For weakly interaction regime, i.e. $0<\beta\ll1$, we have for $d\ge2$
\begin{align} \label{harsmall_degenerate}
E_1(\beta)&=\frac{3\gm}{2}+\frac{B_0d}{8}\beta+o(\beta), \quad
\mu_1(\beta)=\frac{3\gm}{2}+\frac{B_0d}{4}\beta+o(\beta).
\end{align}
\end{lemma}
\begin{proof}
For simplicity, we only present the 2D case and extension to 3D is straightforward.
Denote
\be
\phi_g^0(x)=\left(\frac{\gm}{\pi}\right)^{\frac{1}{4}}e^{-\frac{\gm x^2}{2}}, \qquad \phi_1^0(x)=\sqrt{2\gm}\left(\frac{\gm}{\pi}
\right)^{\frac{1}{4}}xe^{-\frac{\gm x^2}{2}}.
\ee
When $d=2$ and $\beta=0$, it is easy to see that
$\varphi_1(\bx):=\phi_1^{0}(x_1)\phi_g^{0}(x_2)$
and $\varphi_2(\bx):=\phi_g^{0}(x_1)\phi_1^{0}(x_2)$
are two linearly independent orthonormal first excited states.
In fact, $W_1={\rm span}\{\varphi_1,\varphi_2\}$.
In order to find an appropriate approximation of the first
excited state when $0<\beta \ll1$, we take an ansatz
\be\label{har01_degenerate}
\varphi_{a,b}(\mathbf{x})= a
\varphi_1(\bx)+ b\varphi_2(\bx), \qquad \bx\in{\mathbb R}^2,
\ee
where $a,b\in\mathbb{C}$ satisfying $|a|^2+|b|^2=1$ implies
$\|\varphi_{a,b}\|_2=1$.
Then $a$ and $b$ can be determined by
minimizing $E(\varphi_{a,b})$. Plugging \eqref{har01_degenerate}
into \eqref{def:E}, we have for $\beta\ge0$
\begin{align}
E(\varphi_{a,b})=3\gm+\frac{\gm\beta}{16\pi}\left[|a^2+b^2|^2+2
(|a|^2+|b|^2)^2\right]\ge3\gm+\frac{\gm\beta}{8\pi},
\end{align}
which is minimized when $a^2+b^2=0$, i.e. $a=\pm ib$. By taking $a=1/\sqrt{2}$ and $b=i/\sqrt{2}$, we get an approximation of the first excited state as
\be\label{phi1har11}
\phi_1^\beta(\bx)\approx \phi_{1,v}(\bx)=\frac{\gm}{\sqrt{\pi}}(x_1+ix_2)
e^{-\frac{\gm(x_1^2+x_2^2)}{2}}=\frac{\gm}{\sqrt{\pi}}r
e^{-\frac{\gm r^2}{2}} e^{i\theta},
\ee
where $(r,\theta)$ is the polar coordinate.
Substituting \eqref{phi1har11} into (\ref{def:E}) and (\ref{def:mu}),
we get (\ref{harsmall_degenerate}).
\end{proof}


\begin{lemma}\label{har:ground_strong_degenerate}
For the 2D case with strongly repulsive interaction, i.e. $d = 2$ and $\beta\gg1$, we have
\bea \label{harstrong_degenerate_E}
\qquad \quad E_1(\beta)=E_g^{\rm TF}+\frac{\gm}{2}\sqrt{\frac{\pi}{\beta}}\ln(\beta)+\mathcal{O}
(\beta^{-\frac{1}{2}}),
\mu_1(\beta)=\mu_g^{\rm TF}+\frac{\gm}{4}\sqrt{\frac{\pi}{\beta}}\ln(\beta)+\mathcal{O}
(\beta^{-\frac{1}{2}}),
\eea
where $\mu_g^{\rm TF}$ is given in \eqref{hartfm1} and $E_g^{\rm TF}=\frac{2+d}{4+d}\mu_g^{\rm TF}$.
\end{lemma}

\begin{proof}
From Lemma \ref{har:ground_weak_degenerate2}, when $0<\beta\ll1$,
the first excited state needs to be taken as a vortex-type solution.
By assuming that there is no band crossing when $\beta>0$,
the first excited state can be well approximated by the vortex-type solution when $\beta\gg1$ too. Thus when $\beta\gg1$, we approximate
the first excited state via a matched asymptotic approximation.

(i) In the outer region, i.e. $|\bx|>o(1)$, it is approximated
by the TF approximation as
\be
\phi_1^\beta(\bx)\approx \phi^{\rm{out}}(\mathbf{x})\approx
\sqrt{\frac{(2\mu_1-\gm^2r^2)_+}{2\beta}}, \qquad r>o(1),
\ee
where  $\mu=\mu_1(\beta)$ is  the chemical potential of the first excited state.

(ii) In the inner region near the origin, i.e. $|\bx|\ll1$, it is approximated by a
vortex solution with winding number $m=1$ as
\begin{equation}\label{sol:vortex1}
\phi_1^\beta(\bx)\approx \phi^{\rm{in}}(\mathbf{x})=\sqrt{\frac{\mu_1}
{\beta}}f(r)e^{i\theta}, \qquad |\bx|\ll1,
\end{equation}
Substituting (\ref{sol:vortex}) into (\ref{eq:eig}), we get the equation for $f(r)$
\begin{equation}\label{eq:vortex1}
-\frac{1}{2}f''(r)-\frac{1}{2r}f'(r)+\frac{1}{2r^2}f(r)+\frac{\gm^2 r^2}{2}f(r)+\mu_1 f^3(r)=\mu_1f(r), \qquad r>0,
\end{equation}
with BC $f(0)=0$. When $\beta\gg1$ and $0\le r\ll 1$, by dropping the terms $-\frac{1}{2}f''(r)$ and $\frac{\gm^2 r^2}{2}f(r)$
in (\ref{eq:vortex1}) and then solving it analytically with the far field limit $\lim_{r\to+\infty}f(r)=1$, we get
\eqref{sol:vortex_approx}.
Combining the outer and inner approximations via the matched asymptotic
technique, we obtain an asymptotic approximation of the density of the
first excited state as
\be\label{proof:har_1st_degen}
\rho_1^\beta(\bx):=|\phi_1^{\beta}(\mathbf{x})|^2\approx \frac{2\mu_1 r^2}{1+2\mu_1 r^2}\frac{(2\mu_1-\gm^2r^2)_+}{2\beta}, \qquad r\ge0.
\ee
Substituting \eqref{proof:har_1st_degen} into
the normalization condition $\|\phi_1^\beta\|_2=1$ and \eqref{def:mu}, a detailed computation gives
the approximation of the chemical potential and energy in (\ref{harstrong_degenerate_E}).
The details of the computation are omitted here for brevity \cite{Ruan}.
\end{proof}


\begin{figure}[htb]
\centerline{\psfig{figure=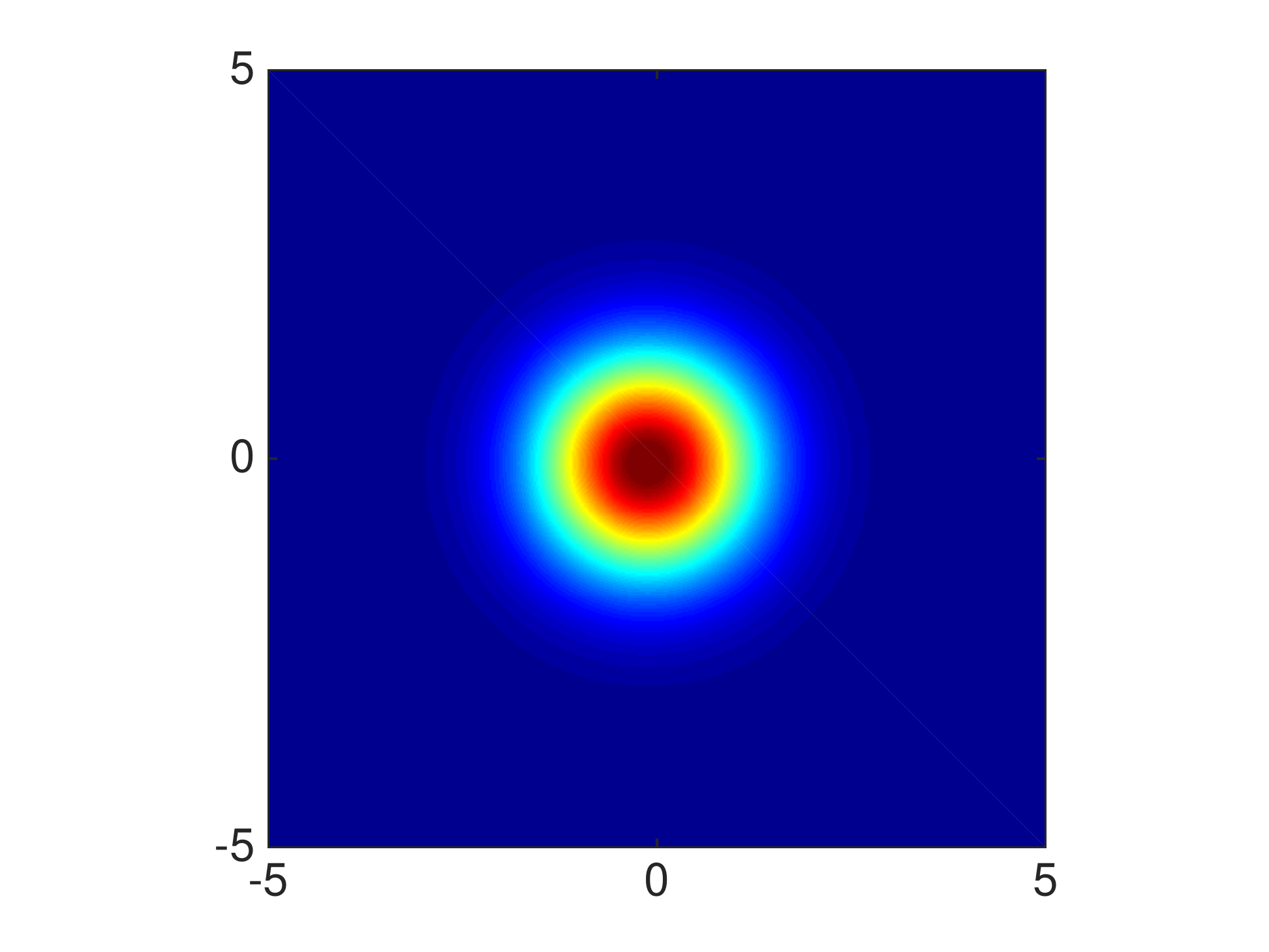,height=3.5cm,width=5cm,angle=0}
\psfig{figure=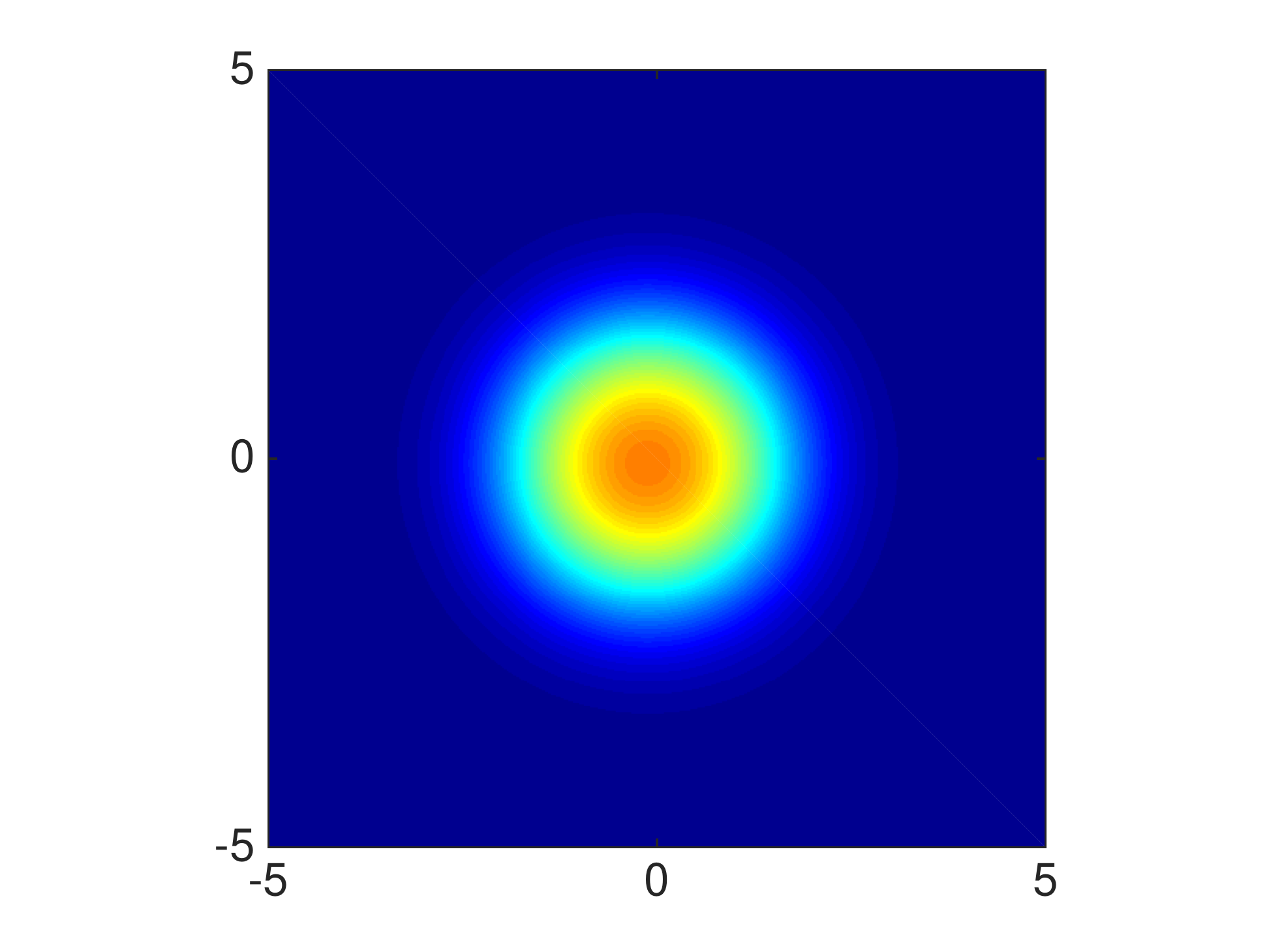,height=3.5cm,width=5cm,angle=0}
\psfig{figure=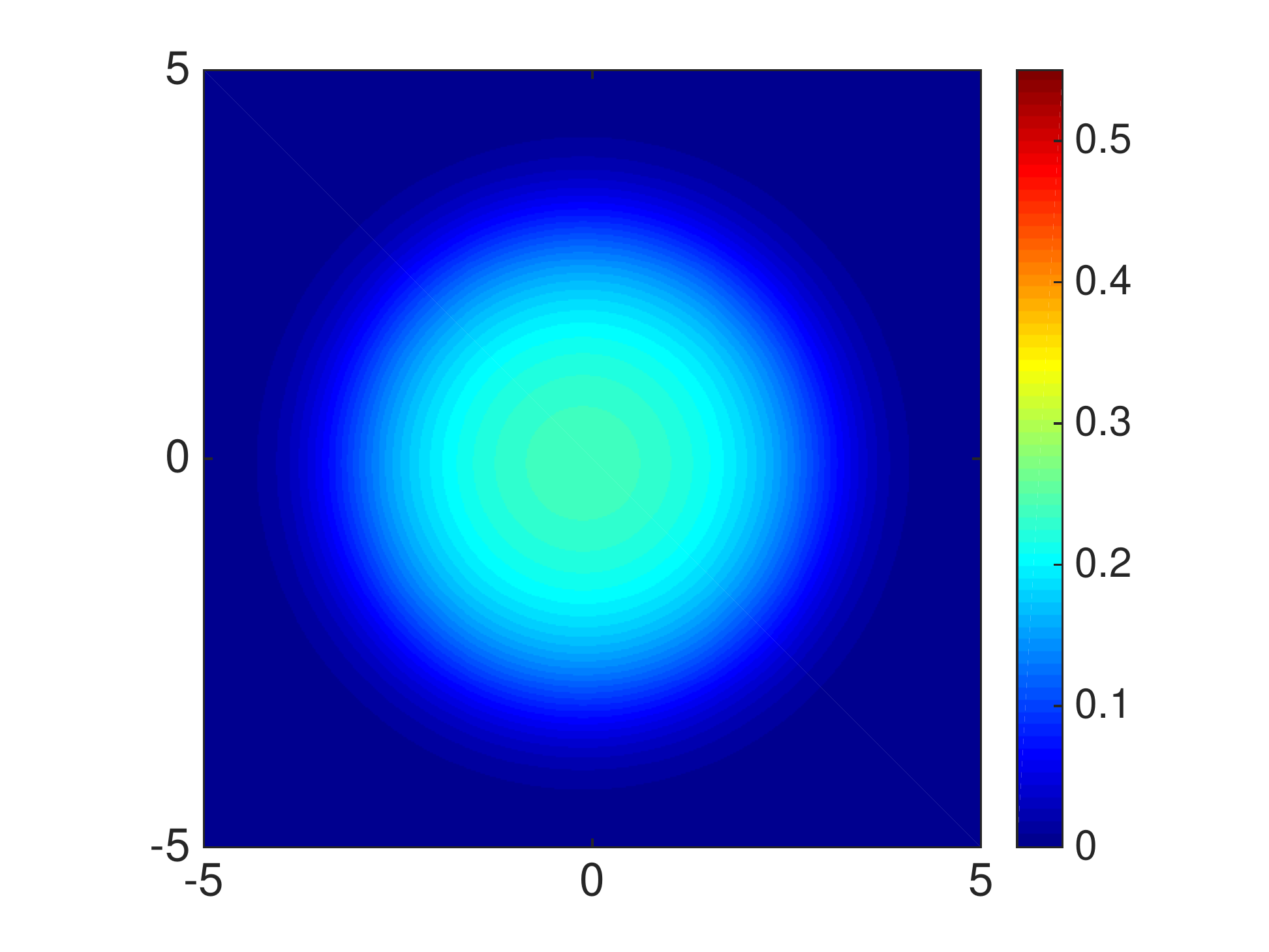,height=3.5cm,width=5cm,angle=0}}
\centerline{\psfig{figure=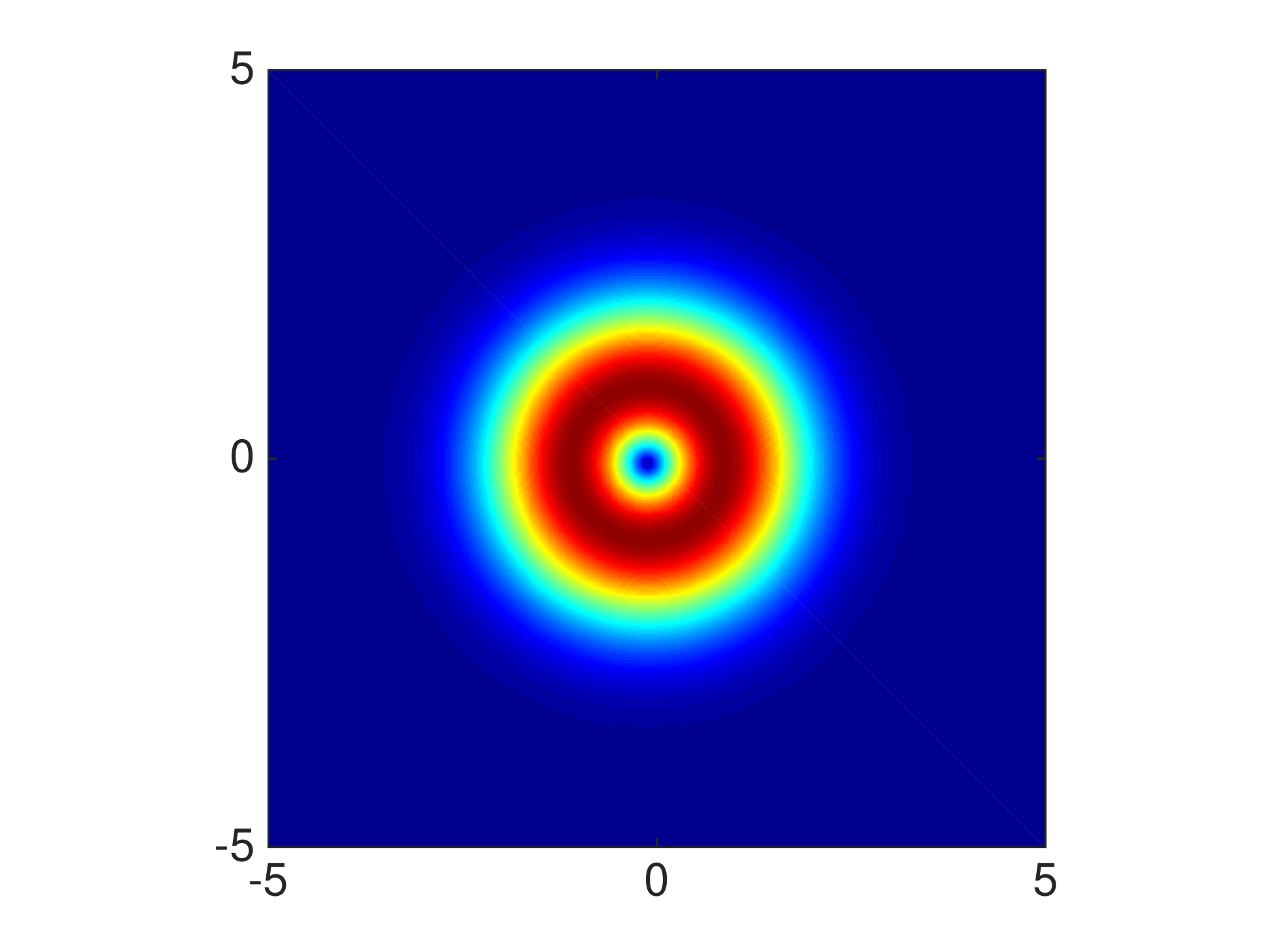,height=3.5cm,width=5cm,angle=0}
\psfig{figure=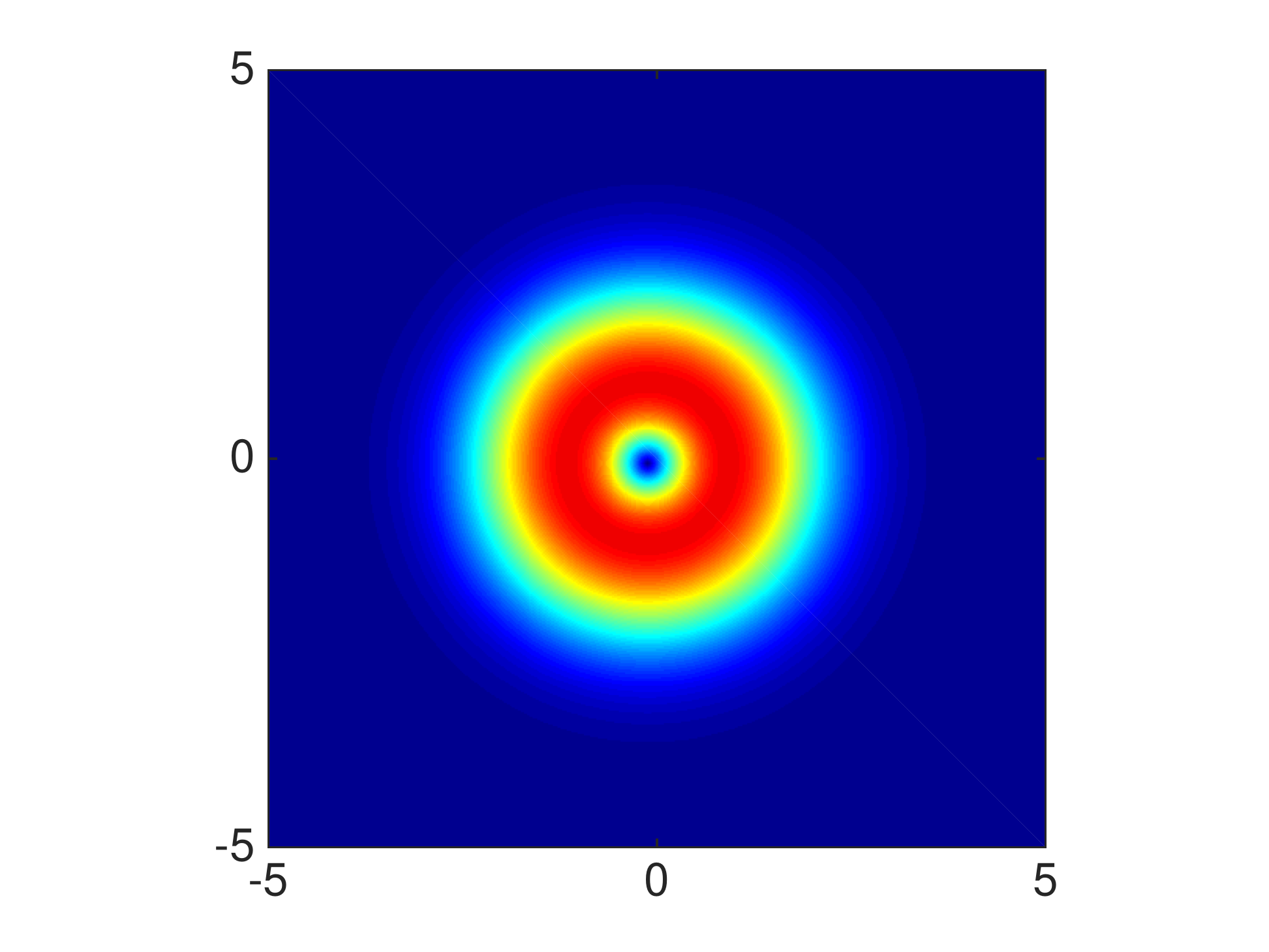,height=3.5cm,width=5cm,angle=0}
\psfig{figure=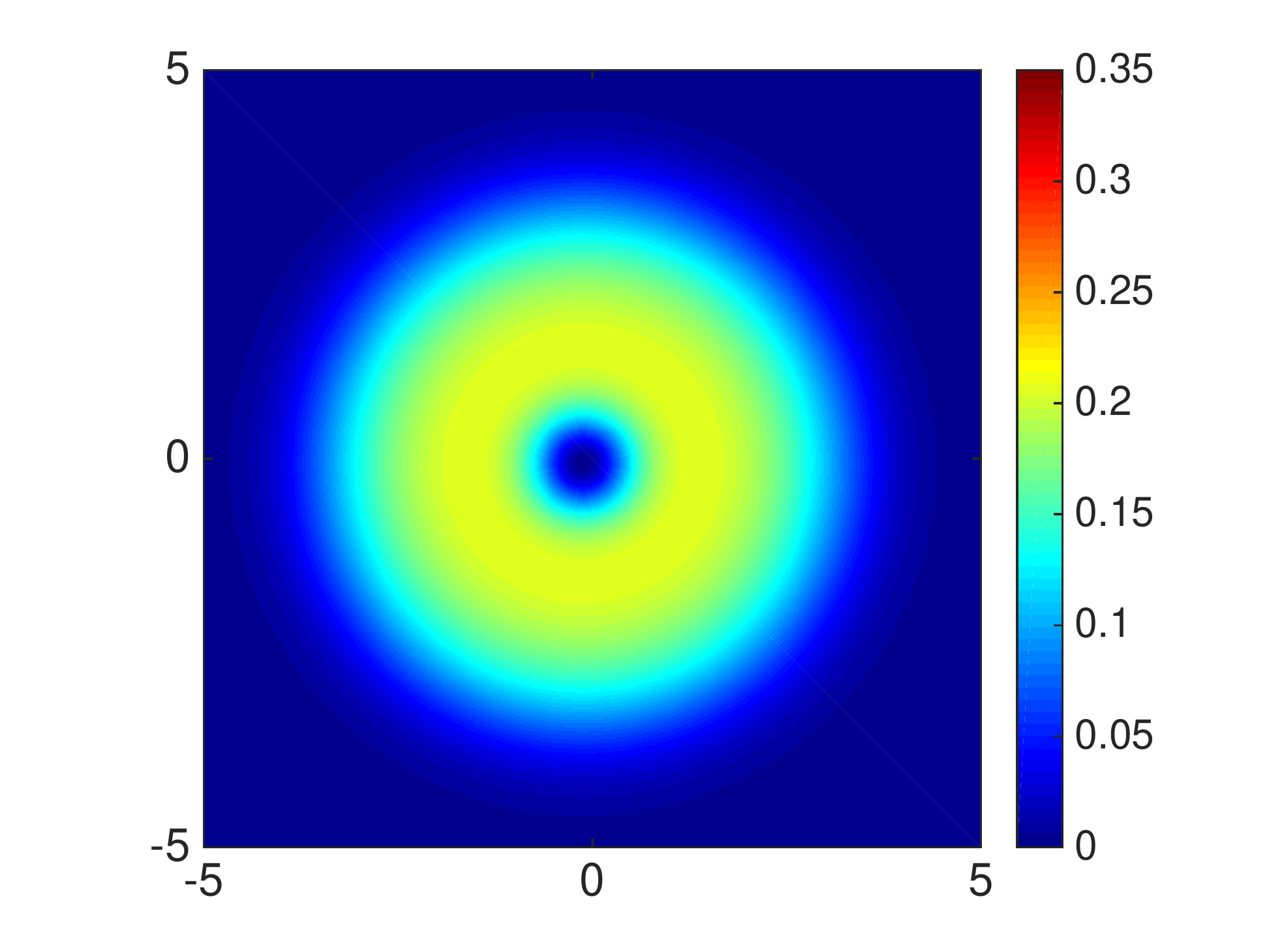,height=3.5cm,width=5cm,angle=0}}
\centerline{\psfig{figure=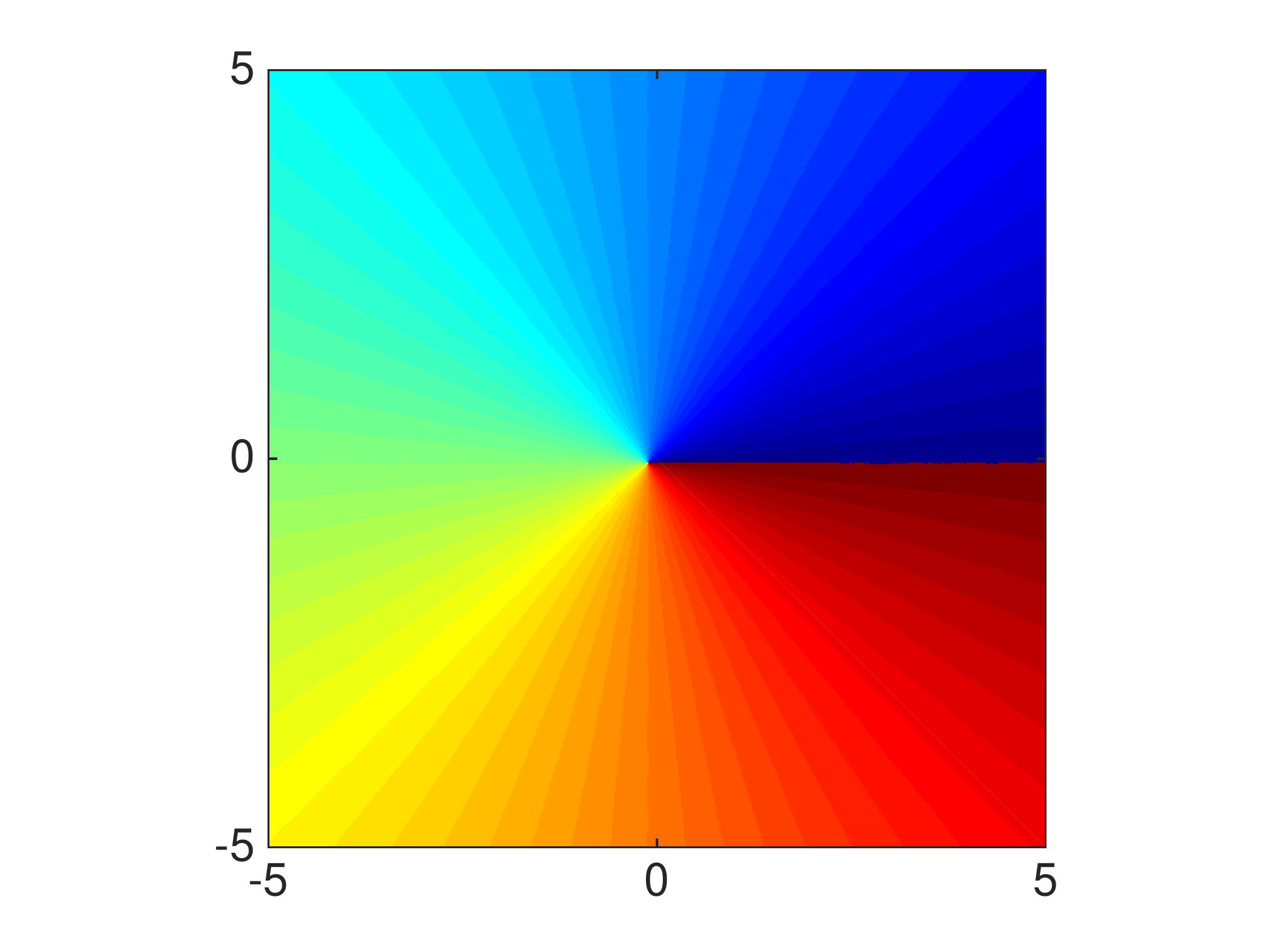,height=3.5cm,width=5cm,angle=0}
\psfig{figure=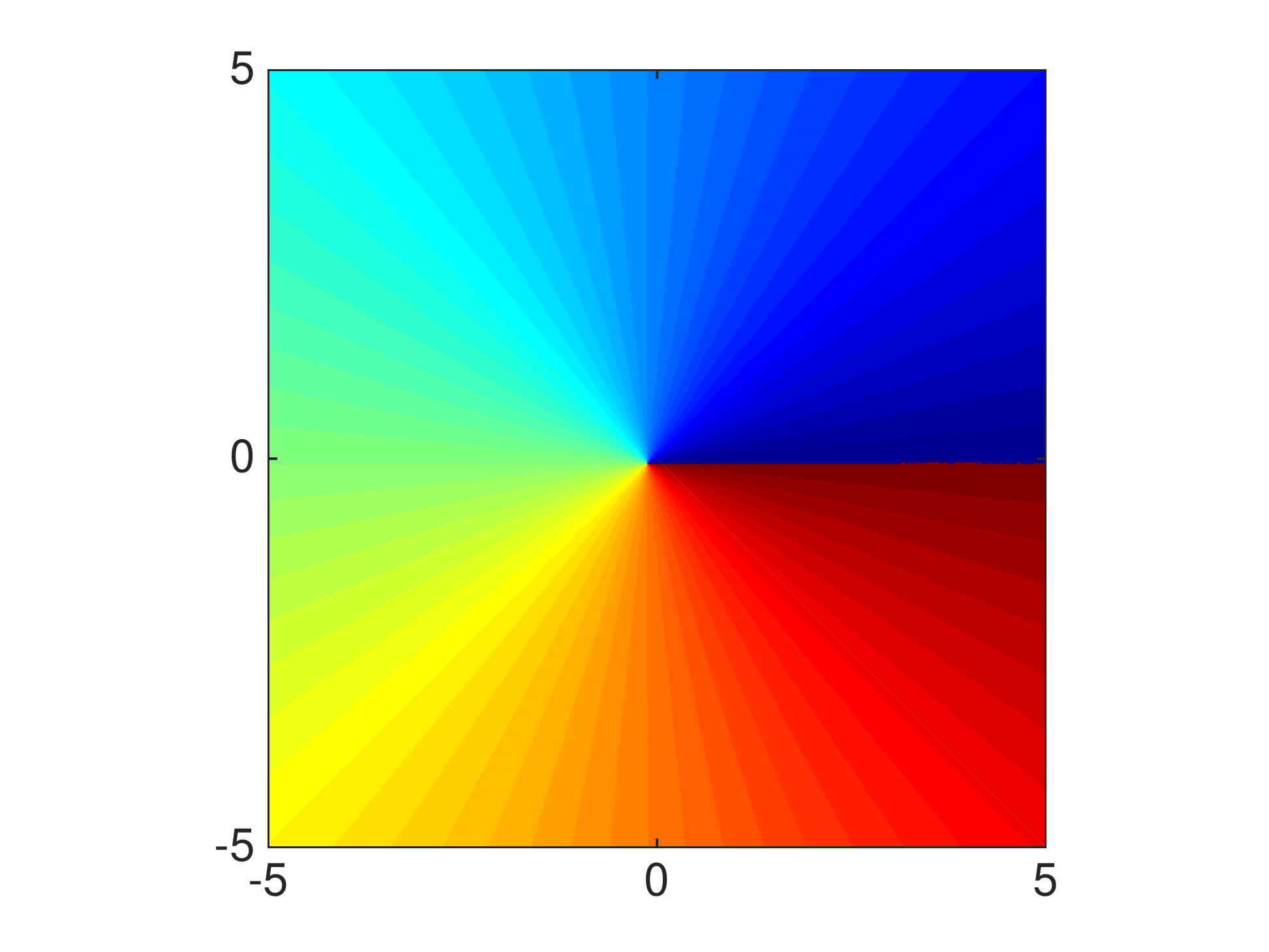,height=3.5cm,width=5cm,angle=0}
\psfig{figure=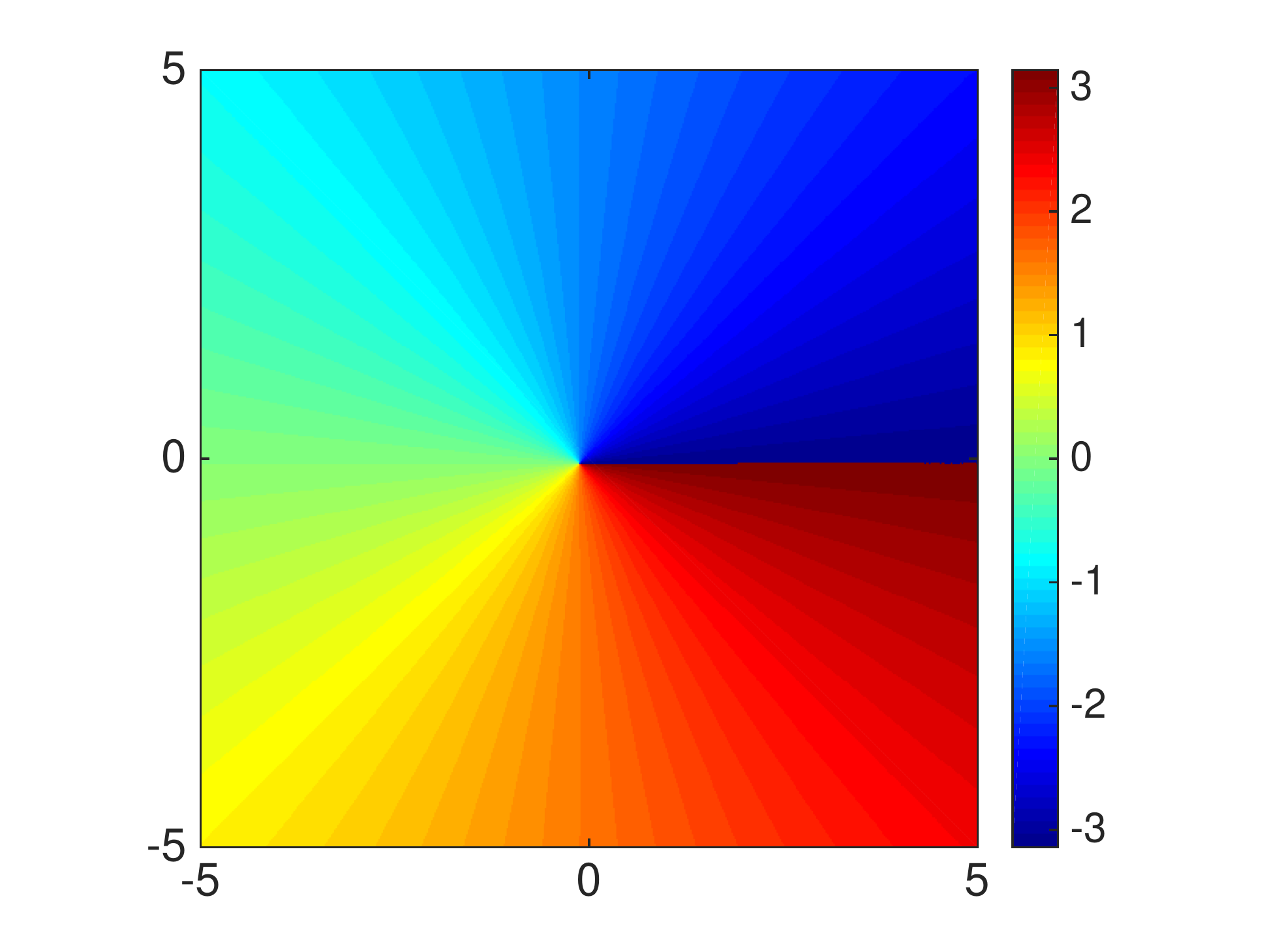,height=3.5cm,width=5cm,angle=0}}
\centerline{\psfig{figure=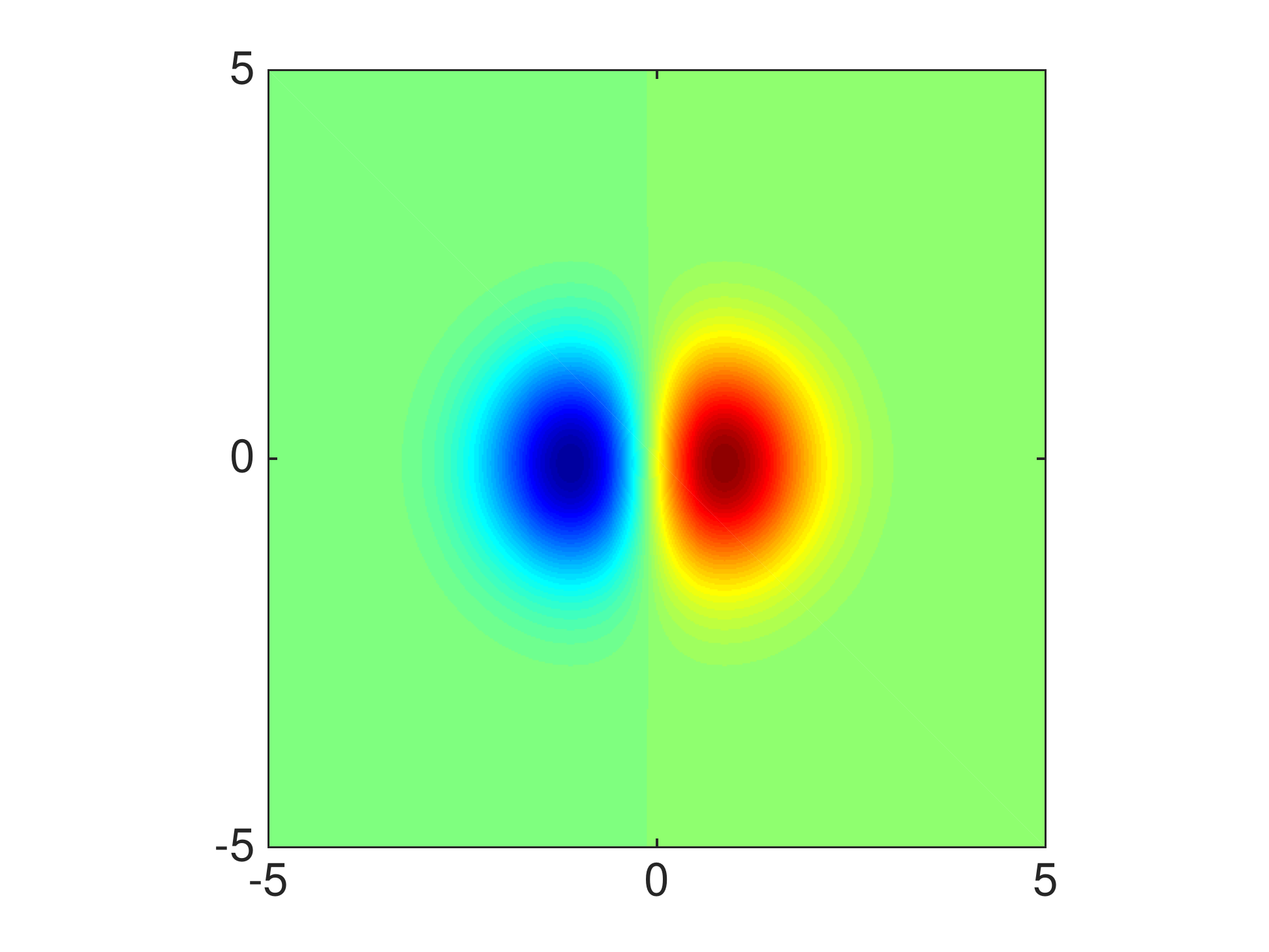,height=3.5cm,width=5cm,angle=0}
\psfig{figure=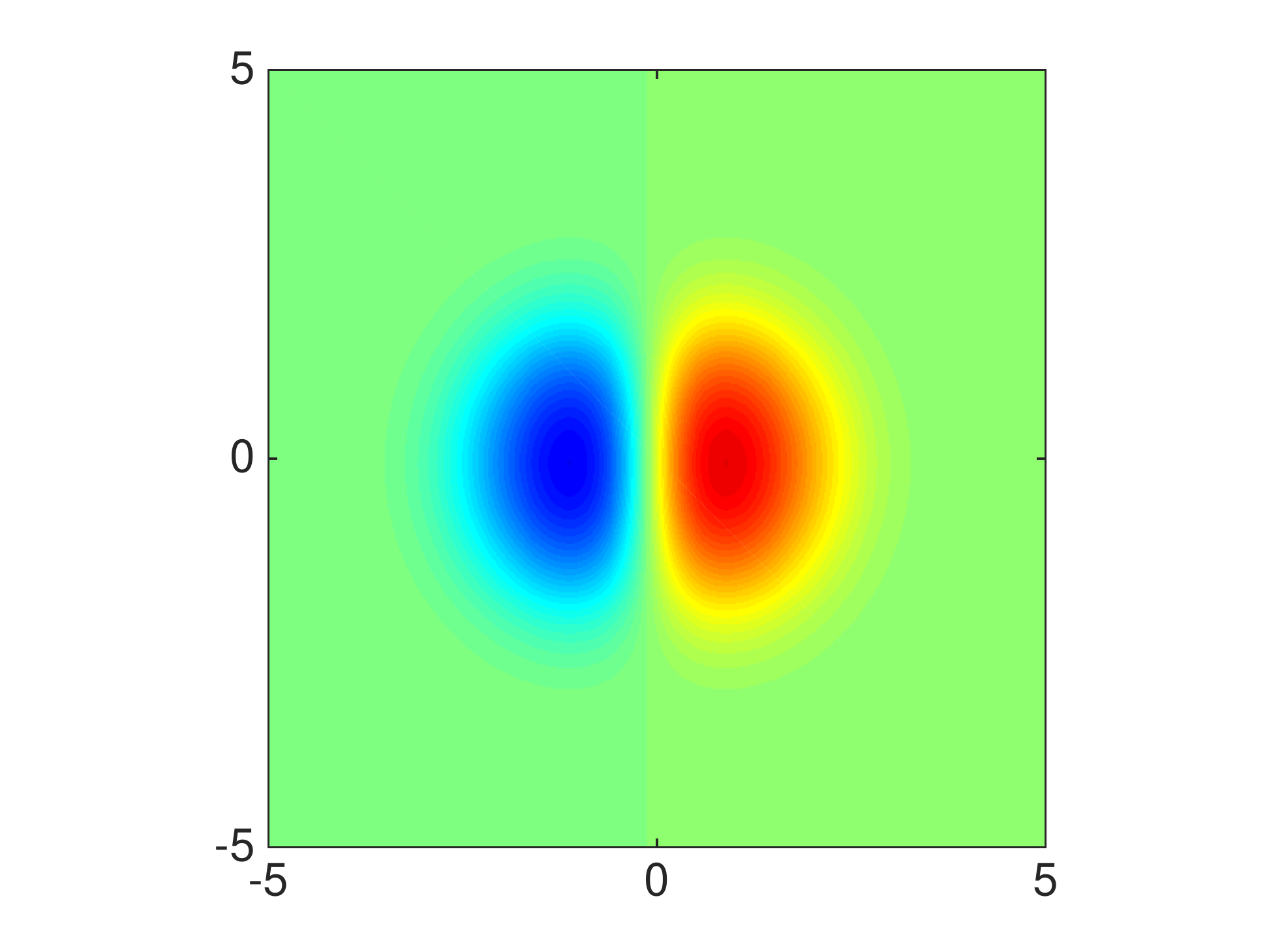,height=3.5cm,width=5cm,angle=0}
\psfig{figure=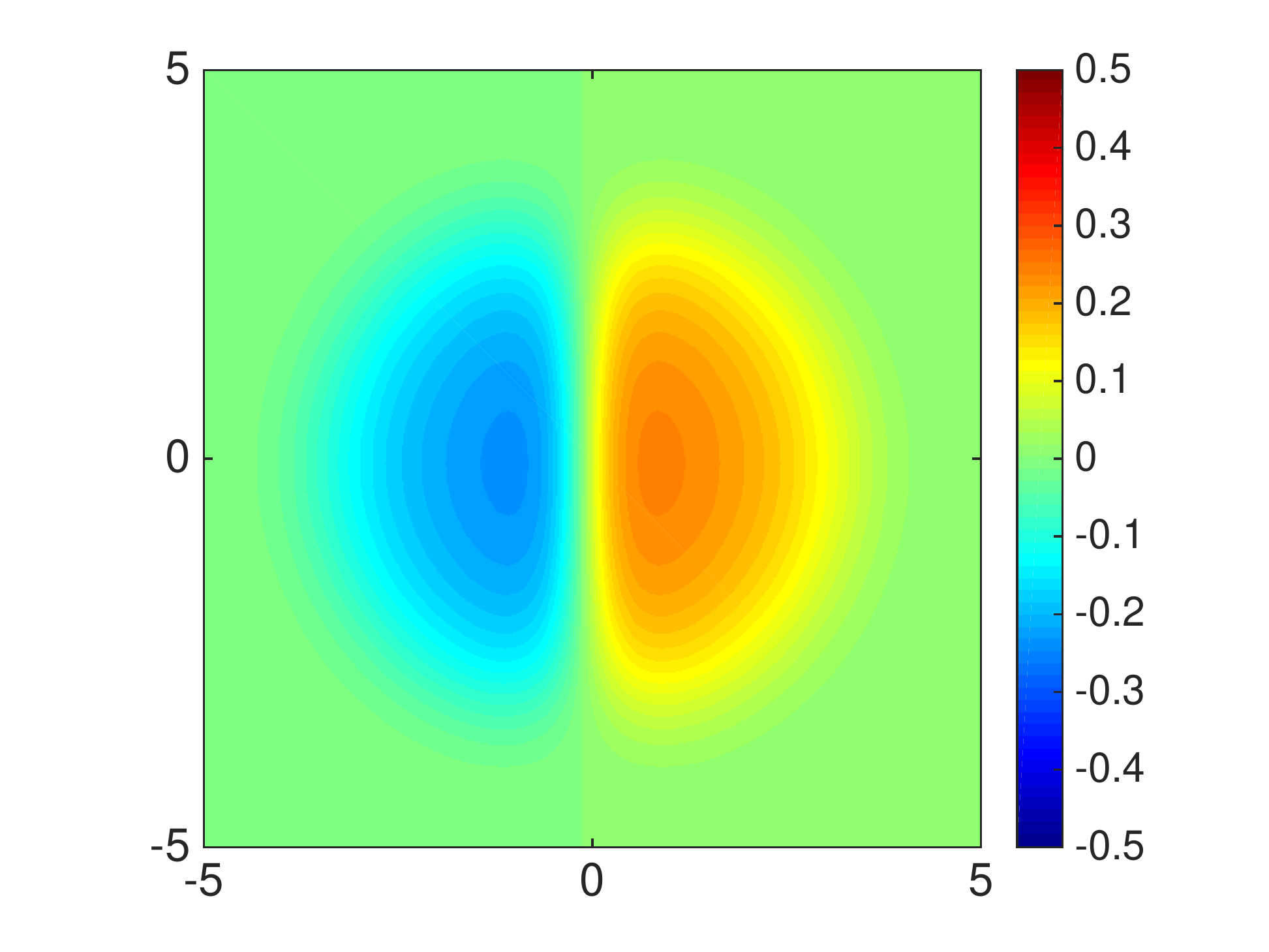,height=3.5cm,width=5cm,angle=0}}
\caption{Ground state $\phi_g^\beta$ (top row), first excited state -- vortex solution $|\phi_1^\beta=\phi_{1,v}^\beta|$(second row) and higher excited state in $x_1$-direction $\phi_{1,x}^\beta$(bottom row) of GPE in 2D with a harmonic potential ($\gm=1$) for
$\beta=0$ (left column), $\beta=10$ (middle column) and
$\beta=100$ (right column). The phase of the first excited state
$\phi_1^\beta=\phi_{1,v}^\beta$ is displayed in the third row.}
\label{fig:har2d_sol_degenerate}
\end{figure}


\begin{figure}[htb]
\centerline{\psfig{figure=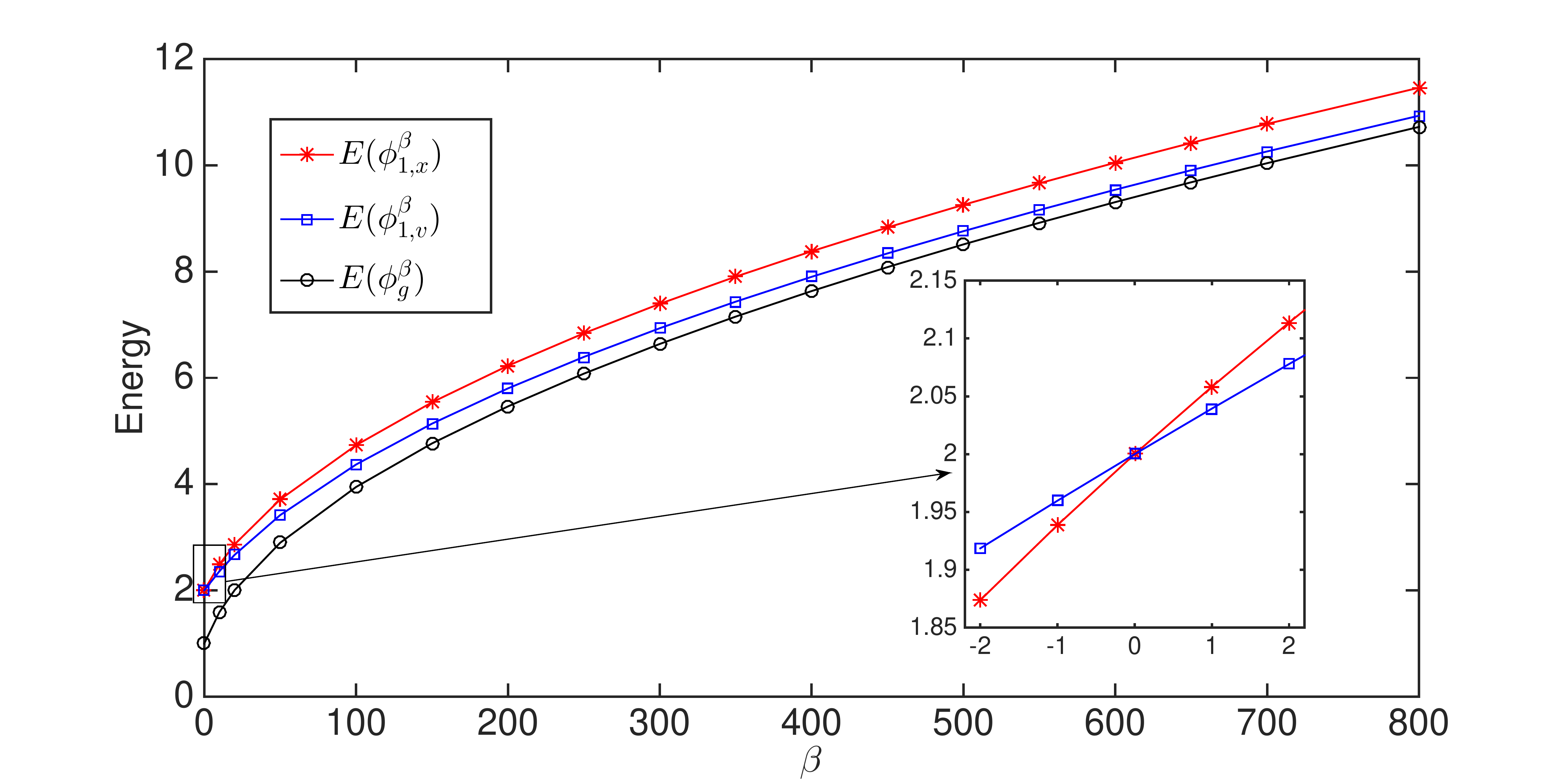,height=4cm,width=13cm,angle=0}}
\centerline{\psfig{figure=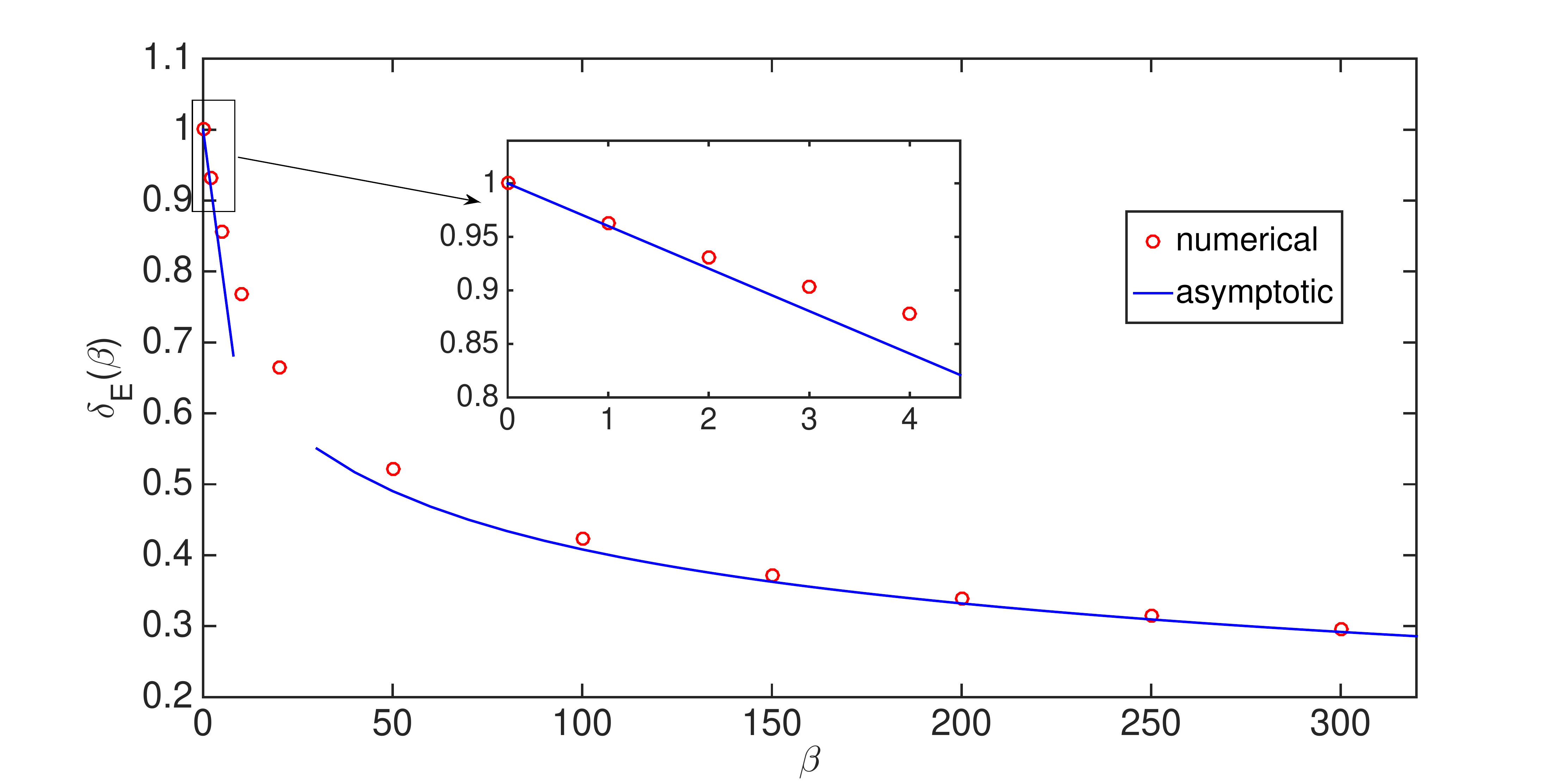,height=4cm,width=13cm,angle=0}}
\caption{Energy $E_g(\beta):=E(\phi_g^{\beta})<E_1(\beta):=E(\phi_1^\beta=\phi_{1,v}^{\beta})
<E_2(\beta):=E(\phi_{1,x_1}^{\beta}) =E(\phi_{1,x_2}^{\beta})$ of GPE in 2D under the harmonic potential $V(\bx)=\frac{x_1^2+x_2^2}{2}$ for different $\beta\ge0$ (top)
and the fundamental gap in energy $\delta_E(\beta)$ for different $\beta\ge0$ (bottom).}
\label{fig:har_2d_degenerate}
\end{figure}

\begin{figure}[h!]
\centerline{\psfig{figure=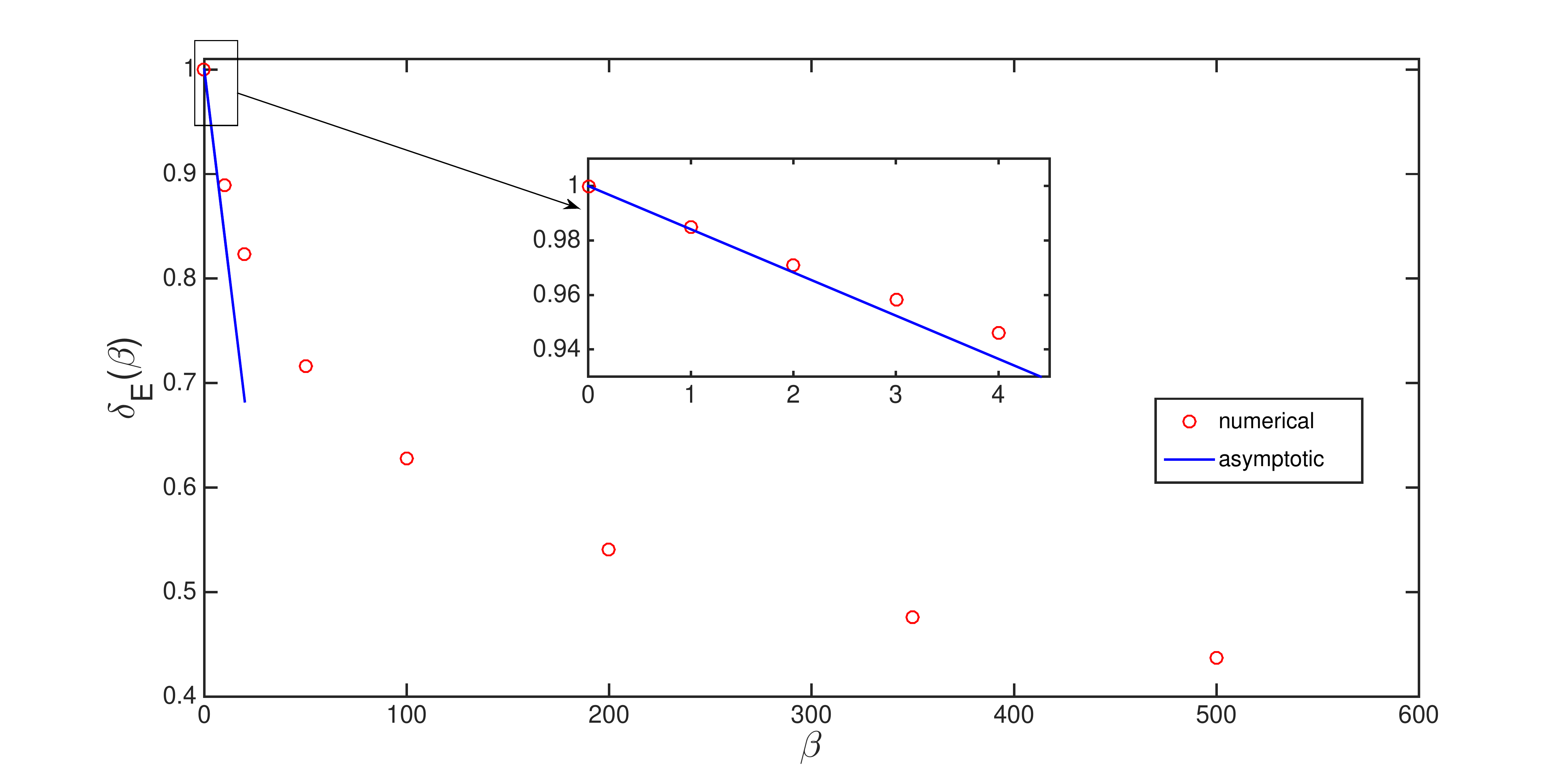,height=4cm,width=13cm,angle=0}}
\caption{The fundamental gaps in energy of GPE in 3D under a harmonic potential $V(\bx)=(x_1^2+x_2^2+x_3^2)/2$ for different $\beta\ge0$.}
\label{fig:har_3d_E_degenerate}
\end{figure}

From Lemmas \ref{har:ground_weak_degenerate2}\&\ref{har:ground_strong_degenerate}, we have asymptotic results for the fundamental gaps.

\begin{proposition}[For GPE under a harmonic potential in degenerate case]\label{asym:har_degen}
When  $V(\mathbf{x})=V_h(\bx)$ with $d\ge2$ and $\gm_1=\gm_2:=\gm$, i.e. GPE with a harmonic potential, we have

(i)  when $0 \le \beta\ll 1$ and $d\ge2$
\be \label{gaphar2}
\delta_E(\beta)=
\gm-\frac{(4-d)B_0}{8}\beta+o(\beta),\quad
\delta_{\mu}(\beta)=
\gm-\frac{(4-d)B_0}{4}\beta+o(\beta);
\ee

(ii) when $\beta\gg1$ and  $d=2$,
\be\label{gaphar3}
\delta_E(\beta)=\frac{\gm_1}{2}\sqrt{\frac{\pi}{\beta}}
\ln(\beta)+o\left(\frac{\ln(\beta)}{\sqrt{\beta}}\right),\quad \delta_{\mu}(\beta)=\frac{\gm_1}{4}\sqrt{\frac{\pi}{\beta}}\ln(\beta)+
o\left(\frac{\ln(\beta)}{\sqrt{\beta}}\right),
\ee
which implies $\delta_E(\beta)\to 0$ and $\delta_{\mu}(\beta)\to0$ as $\beta\to\infty$. 
\end{proposition}

Again, to verify numerically our asymptotic results in Proposition \ref{asym:har_degen}, Fig.~\ref{fig:har2d_sol_degenerate} plots
the ground state $\phi_g^\beta$, the first excited state
$\phi_1^\beta=\phi_{1,v}^\beta$, and the higher excited states
$\phi_{1,x}^\beta$, of GPE in 2D with a harmonic potential ($\gm=1$)
for different $\beta\ge0$,
which were obtained numerically \cite{Bao_comp1,comp_gf,Bao2013,Wz1}.
Fig.~\ref{fig:har_2d_degenerate} depicts the energy $E_g(\beta)=E(\phi_g^\beta)<E_1(\beta)=E(\phi_1^\beta=\phi_{1,v}^\beta)
<E(\phi_{1,x}^\beta)$ for different $\beta\ge0$ and the corresponding
fundamental gaps in energy, and Fig.~\ref{fig:har_3d_E_degenerate} shows
the fundamental gaps in energy of GPE in 3D with a harmonic
potential. In addition,
our numerical results suggest that both $\delta_E(\beta)$ and $\delta_\mu(\beta)$
are decreasing functions for $\beta\ge0$ (cf. Figs. \ref{fig:har_2d_degenerate}\&\ref{fig:har_3d_E_degenerate}).

Based on the asymptotic results in Proposition \ref{asym:har_degen} and
the above numerical results as well as additional extensive numerical results not shown here for brevity \cite{Ruan},
we speculate the following gap conjecture.

\textbf{Gap conjecture}
(For GPE in the whole space in degenerate case)
Suppose the external potential $V(\mathbf{x})$ satisfies
$D^2V(\bx)\ge\gm_v^2I_d$ for $\bx\in{\mathbb R}^d$ with $\gm_v>0$  a constant and  $\dim(W_1)\ge2$. When $0\le \beta\ll1$, we have
\begin{equation}\label{bgap97788}
\delta_E(\beta)\ge\gm_v-C_1\beta,
\qquad\delta_\mu(\beta)\ge\gm_v-C_2\beta,
\end{equation}
where $C_1>0$ and $C_2>0$ are two constants independent of $\beta$.
In addition, we have $\lim_{\beta\to+\infty}\delta_E(\beta)=0$ and
$\lim_{\beta\to+\infty}\delta_\mu(\beta)=0$.

We remark here that the fundamental gap $\delta_E(\beta)$ gives
an upper bound of the critical rotating speed $\Omega_c^\beta$ in rotating BEC for $\beta\ge0$
\cite{Sei,BaoWP,Bao,Bao2013}, which implies that $\lim_{\beta\to+\infty}\Omega_c^\beta=0$.


\section{Extensions to other BCs}\label{asym:neumann_periodic}
\setcounter{equation}{0}
\setcounter{figure}{0}
In this section, we study the fundamental gaps of GPE on bounded domains with either periodic BC or homogeneous Neumann BC.
\subsection{Results for the periodic BC}\label{ring}
Take  $\Omega=\Omega_0$ and assume that $\phi$ satisfies the periodic BC.
When $d=1$, it corresponds to a BEC on a ring \cite{Bao2013}; and when $d=2$, it corresponds
to a BEC on a torus.
In this case, the ground state $\phi_g^\beta$  is defined the same as in (\ref{def:ground}) provided that
the set $S$ is replaced by $S=\{\phi\, |\, \|\phi\|_2^2:=\int_{\Omega} |\phi(\bx)|^2d\bx =1,\ E(\phi)<\infty,
\ \phi \ \hbox{is periodic on }\ \partial \Omega \}$,
and the first excited state $\phi_1^\beta$ and the eigenspace $W_1$ are defined similarly.
We have the following results for the energy and chemical potential of the ground  and first excited states.

\begin{lemma}\label{ring:ground} Assume $V(\bx)\equiv 0$,
for all $\beta\ge0$ and $d=1,2,3$, we have
\begin{align}
 \label{eng:ring_ground}
E_g(\beta)=\frac{A_0^2}{2}\beta, \quad \mu_g(\beta)=A_0^2\beta, \quad
E_1(\beta)=\frac{2\pi^2}{L_1^2}+\frac{A_0^2}{2}\beta,\quad
\mu_1(\beta)=\frac{2\pi^2}{L_1^2}+A_0^2\beta.
\end{align}
\end{lemma}

\begin{proof}
For any $\phi\in S$, the Cauchy-Schwarz inequality implies that
\be\label{proof:const_int}
1=\|\phi\|_2^4=\left(\int_{\Omega}|\phi|^2d\mathbf{x}\right)^2
\le \int_{\Omega}|\phi|^4d\mathbf{x} \, \int_{\Omega}1d\mathbf{x}=
\frac{1}{A_0^2}\int_{\Omega}|\phi|^4d\mathbf{x}.
\ee
Thus, for all $\beta\ge0$ and any $\phi\in S$, we have
\be\label{engper789}
E(\phi)=\int_{\Omega}\left[\frac{1}{2}|\nabla\phi|^2+\frac{\beta}{2}
|\phi|^4\right]d\mathbf{x}\ge\frac{\beta}{2}A_0^2=\frac{\beta}{2}
\int_{\Omega}|A_0|^4d\mathbf{x}=E(\phi\equiv A_0).
\ee
Therefore,  for all $\beta\ge0$, we have
\be\label{sol:ring_ground}
\phi_g^{\beta}(\mathbf{x})\equiv \phi_g^0(\bx):= A_0, \qquad \bx\in\overline{\Omega}.
\ee
 Plugging (\ref{sol:ring_ground})
into (\ref{def:E}) and (\ref{def:mu}) and noticing $V(\bx)\equiv0$, we obtain the first two equalities in (\ref{eng:ring_ground}).

As for the first excited state, for simplicity, we only present 1D case and
extensions to 2D and 3D are straightforward. When $d=1$ and $\beta=0$,
it is easy to see that
$\varphi_1(x):=\sqrt{2}A_0\cos\left(2\pi x/L_1\right)$
and $\varphi_2(x):=\sqrt{2}A_0\sin\left(2\pi x/L_1\right)$
are two linearly independent orthonormal first excited states.
In fact, in this case, $W_1={\rm span}\{\varphi_1,\varphi_2\}$.
In order to find an appropriate approximation of the first
excited state when $0<\beta \ll1$, we take an ansatz
\be\label{per01_degenerate}
\varphi_{a,b}(\mathbf{x})= a
\varphi_1(x)+ b\varphi_2(x), \qquad 0\le x\le L_1,
\ee
where $a,b\in\mathbb{C}$ satisfying $|a|^2+|b|^2=1$ implies
$\|\varphi_{a,b}\|_2=1$.
Then $a$ and $b$ can be determined by
minimizing $E(\varphi_{a,b})$. Plugging \eqref{per01_degenerate}
into \eqref{def:E}, we have for $\beta\ge0$
\begin{align}
E(\varphi_{a,b})=\frac{2\pi^2}{L_1^2}+\frac{\beta}{4L_1}\left[2(|a|^2+|b|^2)^2+|a^2+b^2|^2\right]\ge\frac{2\pi^2}{L_1^2}+\frac{\beta}{2L_1},
\end{align}
which is minimized when $a^2+b^2=0$, i.e. $a=\pm ib$. By taking $a=1/\sqrt{2}$ and $b=i/\sqrt{2}$, we get an approximation of the first excited state as
\be\label{sol:ring_1st}
\phi_1^{\beta}(x)\approx \phi_1^0(x):=A_0e^{i2\pi x/L_1}, \qquad 0\le x\le L_1.
\ee
Similar to \eqref{proof:const_int} and \eqref{engper789}, we can prove
rigorously that for all $\beta\ge0$
\be\label{sol:ring_excit}
\phi_1^{\beta}(x)\equiv \phi_1^0(x)= A_0e^{i2\pi x/L_1}, \qquad 0\le x\le L_1.
\ee
Plugging \eqref{sol:ring_excit} into \eqref{def:E} and \eqref{def:mu}, we obtain the last two equalities in \eqref{eng:ring_ground}.
\end{proof}



From \eqref{eng:ring_ground}, it is straightforward to have (with the proof
omitted here for brevity).

\begin{proposition} [For GPE on a bounded domain with periodic BC] \label{asym:ring} Assume $V(\mathbf{x})\equiv 0$, we have 
\be
\delta_E(\beta)=\delta_{\mu}(\beta)=\frac{2\pi^2}{L_1^2}, \qquad \beta\ge0.
\ee
\end{proposition}

Based on the above analytical results and extensive numerical results not shown here for brevity \cite{Ruan}, we speculate the following gap conjecture.

\textbf{Gap conjecture}
(For GPE on a bounded domain with periodic BC)
Suppose the external potential $V(\mathbf{x})$ is convex,
we speculate the following gap conjecture
\begin{equation}\label{bgap6776}
\delta^{\infty}_E:=\inf_{\beta\ge0} \delta_E(\beta)\ge\frac{2\pi^2}{D^2},
\qquad\delta^{\infty}_{\mu}:=\inf_{\beta\ge0} \delta_\mu(\beta)\ge\frac{2\pi^2}{D^2}.
\end{equation}

\subsection{Results for homogeneous Neumann BC}
Assume that $\Omega\subset{\mathbb R}^d$ is a bounded domain and $\phi$ satisfies the homogeneous Neumann
BC, i.e. $\left.\partial_{\bf n} \phi\right|_{\partial\Omega}=0$
with ${\bf n}$ the unit outward normal vector.
In this case, the ground state $\phi_g^\beta$  is defined the same as in (\ref{def:ground}) provided that
the set $S$ is replaced by $S=\{\phi\, |\, \|\phi\|_2^2:=\int_{\Omega} |\phi(\bx)|^2d\bx =1,\ E(\phi)<\infty,
\ \left.\partial_{\bf n} \phi\right|_{\partial\Omega}=0\}$,
and the first excited state $\phi_1^\beta$ and the eigenspace $W_1$ are defined similarly.


Similar to Lemma \ref{ring:ground} (with the proof omitted here for brevity), we have

\begin{lemma}\label{neumann:ground} For the ground state $\phi_g^\beta$,
 we have for $\beta\ge0$
\begin{align}
\label{eng:neumann_ground}
\phi_g^{\beta}(\mathbf{x})=\phi_g^0(\bx)\equiv\frac{1}{\sqrt{\int_{\Omega} 1\,d\bx}}:= \tilde A_0, \ \bx\in\overline{\Omega}; \quad E_g(\beta)=\frac{\tilde A_0^2}{2}\beta,
\ \mu_g(\beta)=\tilde A_0^2\beta.
\end{align}
\end{lemma}

However, for the first excited state, we first consider a special case
by taking $\Omega=\Omega_0$ and distinguish two different cases:
(i) non-degenerate case $d=1$ or $L_1>L_2$ when $d\ge2$ ($\Leftrightarrow {\rm dim}(W_1)=1$); and (ii) degenerate case $L_1=L_2$ and
$d\ge2$ ($\Leftrightarrow {\rm dim}(W_1)\ge2$).

\begin{lemma}\label{neumann:nondegenerate} 
Assume $\Omega=\Omega_0$ satisfying $d=1$ or $L_1>L_2$ when $d\ge2$, i.e.
non-degenerate case, we have

(i) in the weakly repulsive interaction regime, i.e. $0<\beta\ll1$,
\begin{align}
\label{eng:neumann_1st_weak}
&E_1(\beta)=\frac{\pi^2}{2L_1^2}+\frac{3A_0^2}{4}\beta+o(\beta),
\quad
\mu_1(\beta)=\frac{\pi^2}{2L_1^2}+\frac{3A_0^2}{2}\beta+o(\beta);
\end{align}

(ii) in the strongly repulsive interaction regime, i.e. $\beta\gg1$,
\bea
\label{chm:neumann_1st_strong}
\qquad\ \ E_1(\beta)=\frac{A_0^2}{2}\beta+\frac{4A_0}{3L_1}\beta^{1/2}+\frac{2}{L_1^2}+o(1),  \ \mu_1(\beta)=A_0^2\beta+\frac{2A_0}{L_1}
\beta^{1/2}+\frac{2}{L_1^2}+o(1).
\eea
\end{lemma}

\begin{proof} Here we only present the proof in 1D case and
extension to high dimensions is similar to that in Lemma \ref{box:ground_strong}.
When $d=1$ and $\beta=0$, the first excited state can be taken as
$\phi_1^0(x)=\sqrt{2}A_0\cos(\pi x/L_1)$ for $x\in[0,L_1]$.
When $0<\beta\ll1$, we can approximate $\phi_1^\beta(x)$ by
$\phi_1^0(x)$, i.e.
\be
\label{sol:neumann_1st_weak}
\phi_1^{\beta}(x)\approx\phi_1^{0}(x)=
\sqrt{2}A_0\cos\left(\pi x/L_1\right), \qquad 0\le x\le L_1.
\ee
Plugging (\ref{sol:neumann_1st_weak}) into (\ref{def:E}) and (\ref{def:mu}) with $V(\bx)\equiv0$, we obtain
(\ref{eng:neumann_1st_weak}).
When $\beta\gg1$,  i.e. in strongly repulsive interaction regime,
the first excited state can be approximated via the matched asymptotic method shown in \cite{BaoL,LZBao} as
\be
\label{sol:neumann_1st_strong}
\phi_1^\beta(x)\approx\phi_1^{MA}(x)=
\sqrt{\frac{\mu_1^{MA}}{\beta}}\tanh\left(\sqrt{\mu_1^{MA}}
\left(\frac{L_1}{2}-x\right)\right), \qquad 0\le x\le L_1.
\ee
Substituting (\ref{sol:neumann_1st_strong}) into the normalization condition
\eqref{norm} and (\ref{def:mu}), we obtain
(\ref{chm:neumann_1st_strong}), while the detailed computation is omitted here for brevity \cite{Ruan}.
\end{proof}

\begin{lemma}\label{neumann:1st_degenerate} 
Assume $\Omega=\Omega_0$ satisfying $L_1=L_2:=L$ and $d\ge2$, i.e.
degenerate case,
we have

 (i) in the weakly repulsive interaction regime, i.e. $0\le \beta\ll1$,
\begin{align}
\label{eng:neumann_1st_weak_degenerate}
&E_1(\beta)=\frac{\pi^2}{2L^2}+\frac{5A_0^2}{8}\beta+o(\beta),
\quad
\mu_1(\beta)=\frac{\pi^2}{2L^2}+\frac{5A_0^2}{4}\beta+o(\beta);
\end{align}

(ii) in the strongly repulsive interaction regime, i.e. $\beta\gg1$, and $d=2$,
\begin{align}
\label{chm:neumann_1st_strong_degenerate}
E_1(\beta)=\frac{\beta}{2L^2}+\frac{\pi}{2L^2}\ln(\beta)+o(\ln(\beta)),
\
\mu_1(\beta)=\frac{\beta}{L^2}+\frac{\pi}{2L^2}\ln(\beta)+o(\ln(\beta)).
\end{align}
\end{lemma}

\begin{proof}
The proof is similar to that for Lemmas \ref{box:ground_weak_degenerate2}\&\ref{box:ground_strong_degenerate} in the box potential case and thus it is omitted here for brevity \cite{Ruan}.
\end{proof}

Lemmas \ref{neumann:nondegenerate}\&\ref{neumann:1st_degenerate} implies the following proposition about the fundamental gaps.

\begin{proposition}
[For GPE on a bounded domain with homogeneous Neumann BC]\label{asym:neumann}
Assume $\Omega=\Omega_0$ and $V(\mathbf{x})\equiv 0$, we have

(i) if $d=1$ or $L_1>L_2$ when $d\ge2$, i.e. non-degenerate case,
\begin{align}
\delta_E(\beta)=
\begin{cases}
\frac{\pi^2}{2L_1^2}+\frac{A_0^2}{4}\beta+o(\beta),\\
\frac{4A_0}{3L_1}\beta^{1/2}+\frac{2}{L_1^2}+o(1),\\
\end{cases}
\delta_{\mu}(\beta)=
\begin{cases}
\frac{\pi^2}{2L_1^2}+\frac{A_0^2}{2}\beta+o(\beta),&0\le\beta\ll1,\\
\frac{2A_0}{L_1}\beta^{1/2}+\frac{2}{L_1^2}+o(1),&\beta\gg1;\\
\end{cases}
\end{align}

(ii) if $L_1=L_2:=L$, i.e. degenerate case, with $0\le \beta\ll1$ and $d\ge2$,
\be
\delta_E(\beta)=\frac{\pi^2}{2L^2}+\frac{\beta}{8L^2}+o(\beta),\quad
\delta_{\mu}(\beta)=\frac{\pi^2}{2L^2}+\frac{\beta}{4L^2}+o(\beta).
\ee

For the degenerate case with $\beta\gg1$ and $d=2$,
\be
\delta_E(\beta)=\frac{\pi}{2L^2}\ln(\beta)+o(\ln(\beta)),\
\delta_{\mu}(\beta)=\frac{\pi}{2L^2}\ln(\beta)+o(\ln(\beta)).
\ee
\end{proposition}

The above asymptotic results have been verified numerically \cite{Ruan},
which are omitted here to avoid this paper to be too long.
In addition, our numerical results suggest that both $\delta_E(\beta)$ and $\delta_\mu(\beta)$
are increasing functions for $\beta\ge0$ \cite{Ruan}.

Based on the above asymptotic results and numerical results not shown here for brevity \cite{Ruan}, we speculate the following gap conjecture.

\textbf{Gap conjecture}
(For GPE on a bounded domain with homogeneous Neumann BC)
Suppose $\Omega$ is a convex bounded domain
  and the external potential $V(\mathbf{x})$ is convex,
we speculate a gap conjecture for the fundamental gaps as
\begin{equation}\label{bgap4576}
\delta^{\infty}_E:=\inf_{\beta\ge0} \delta_E(\beta)\ge\frac{\pi^2}{2D^2},
\qquad\delta^{\infty}_{\mu}:=\inf_{\beta\ge0} \delta_\mu(\beta)\ge\frac{\pi^2}{2D^2}.
\end{equation}

\section{Conclusions}\label{conclusion}
Fundamental gaps in energy and chemical potential of the Gross-Pitaevskii equation (GPE) with repulsive interaction were obtained asymptotically and computed numerically
for different trapping potentials and a gap conjecture on fundamental gaps
was formulated. In obtaining the approximation of the first excited state of GPE and the fundamental gaps, two different cases were identified in high dimensions ($d\ge2$), i.e. non-degenerate and degenerate cases which correspond to the dimensions
${\rm dim}(W_1)=1$ and ${\rm dim}(W_1)\ge2$, respectively,
with $W_1$ the eigenspace associated to the second
smallest eigenvalue of the corresponding Schr\"{o}dinger
operator $H:=-\frac{1}{2}\Delta +V(\bx)$.
Our asymptotic results were confirmed by numerical results.
Rigorous mathematical justification for the fundamental gaps
obtained asymptotically and numerically for the GPE in this paper is
on-going. Finally, we remark here that the fundamental gaps in the degenerate
case are the same as those in the nondegenerate case when one requires the solution $\phi$ of \eqref{eq:eig} to be real-valued function instead of complex-valued function.




\begin{thebibliography}{}




\bibitem{Anderson}
M.~H. Anderson, J.~R. Ensher, M.~R. Matthews,
C.~E. Wieman and E.~A. Cornell,
\emph{Observation of Bose-Einstein condensation in a dilute atomic
  vapor},
 Science, 269 (1995), 198--201.

\bibitem{gap}
B.~Andrews and J.~Clutterbuck,
\emph{Proof of the fundamental gap conjecture},
J. Amer. Math. Soc., 24 (2011), 899--899.

\bibitem{intro_gap}
 M. Ashbaugh,
 \emph{The fundamental gap},
 Workshop on Low Eigenvalues of
Laplace and Schr\"{o}dinger Operators, American Institute of Mathematics,
Palo Alto, California (2006).

\bibitem{gap2}
 M.~S. Ashbaugh and R.~Benguria,
 \emph{Optimal lower bound for the gap between the first two
  eigenvalues of one-dimensional Schr\"{o}dinger operators with symmetric
  single-well potentials},
 Proc. Amer. Math. Soc., 105
  (1989), 419--424.

\bibitem{Atk}
 P.~W. Atkins,
 \emph{Physical Chemistry},
 Oxford University Press, 1978.

\bibitem{Bao}
 W. Bao,
 \emph{Mathematical models and numerical methods for Bose-Einstein condensation},
 Proceedings of the International Congress of Mathematicians (Seoul 20140),
 IV (2014), 971--996.


\bibitem{Bao2013}
 W.~Bao and Y.~Cai,
 \emph{Mathematical theory and numerical methods for Bose-Einstein
  condensation},
 Kinet. Relat. Models, 6 (2012), 1--135.

\bibitem{Bao_comp1}
 W.~Bao and M.-H. Chai,
 \emph{A uniformly convergent numerical method for singularly
  perturbed nonlinear eigenvalue problems},
   Commun. Comput. Phys., 4 (2008), 135-160.


\bibitem{comp_gf}
 W.~Bao, I.-L. Chern and F.~Y. Lim,
 \emph{Efficient and spectrally accurate numerical methods for
  computing ground and first excited states in Bose-Einstein condensates},
 J. Comput. Phys., 219 (2006), 836--854.

\bibitem{Wz1}
 W.~Bao and Q.~Du,
 \emph{Computing the ground state solution of Bose-Einstein
  condensates by a normalized gradient flow},
 SIAM J. Sci. Comput., 25 (2004), 1674--1697.

\bibitem{Wg1}
 W. Bao, Y. Ge, D. Jaksch, P. A. Markowich and R. M. Weishaeupl,
 \emph{Convergence rate of dimension reduction in Bose-Einstein condensates},
 Comput. Phys. Commun., 177 (2007), 832--850.


\bibitem{BJP}
 W.~Bao, D. Jaksch and P.A. Markowich,
 \emph{Numerical solution of the Gross-Pitaevskii equation for Bose-Einstein condensation},
 J. Comput. Phys., 187 (2003), 318--342.

\bibitem{BaoL}
 W.~Bao and F.~Y. Lim,
 \emph{Analysis and computation for the semiclassical limits of the ground and excited states of the Gross-Pitaevskii equation},
 Proc. Sympos. Appl. Math., Amer. Math. Soc., 67 (2009), 195--215.


\bibitem{LZBao}
 W.~Bao, F.~Y. Lim and Y.~Zhang,
 \emph{Energy and chemical potential asymptotics for the ground state
  of Bose-Einstein condensates in the semiclassical regime},
 Bull. Inst. Math. Acad. Sin. (N.S.), 2
  (2007), 495--532.

\bibitem{BaoT}
 W.~Bao, W. Tang,
 \emph{Ground state solution of Bose-Einstein condensate by directly minimizing the energy functional},
 J. Comput. Phys.,  187 (2003), 230--254.


\bibitem{BaoWP}
W. Bao, H. Wang and P. A. Markowich,
\emph{Ground, symmetric and central vortex states in rotating Bose-Einstein condensates}, Commun. Math. Sci., 3 (2005), 57--88.


\bibitem{gap1}
 M.~Berg,
 \emph{On condensation in the free-boson gas and the spectrum of the
  Laplacian},
 J. Stat. Phys., 31 (1983), 623--637.

\bibitem{Can}
 E.~Canc\`{e}s,
 \emph{Mathematical models and numerical methods for
electronic structure calculation},
 Proceedings of the International Congress of Mathematicians (Seoul 20140),
 IV (2014), 1017--1042.

\bibitem{Can1}
 E.~Canc\`{e}s, R. Chakir and Y. Maday,
 \emph{Numerical analysis of nonlinear eigenvalue problems},
 J. Sci. Comput., 45 (2010), 90--117.

\bibitem{MFT}
 P.~M. Chaikin and T.~C. Lubensky,
 \emph{Principles of Condensed Matter Physics},
 Cambridge University Press, Cambridge, 1995.


\bibitem{Dalfovo1}
 F.~Dalfovo, S.~Giorgini, L.~P. Pitaevskii and S.~Stringari,
 \emph{Theory of Bose-Einstein condensation in trapped gases},
 Rev. Modern Phys., 71 (1999), 463--512.

\bibitem{Dalfovo2}
 F.~Dalfovo and S.~Stringari,
 \emph{Bosons in anisotropic traps: ground state and vortices},
 Phys. Rev. A, 53 (1996), 2477--2485.




\bibitem{Hoh}
 P. Hohenberg and W. Kohn,
 \emph{Inhomogeneous electron gas},
 Phys. Rev., 136 (1964), B864.

\bibitem{Hook}
 J.~R. Hook and H.E. Hall,
 \emph{Solid State Physics},
 John Wiley \& Sons, 2010.



\bibitem{Kohn}
 W. Kohn and L.~J. Sham,
 \emph{Self-consistent equations including exchange and correlation effects},
 Phys. Rev., 140 (1965), A1133.

\bibitem{GPE_BEC1}
 A.~J. Leggett,
 \emph{Bose-Einstein condensation in the alkali gases: Some
  fundamental concepts},
 Rev. Modern Phys., 73 (2001), 307--356.

\bibitem{Quan_chem}
 I.~N. Levine,
 \emph{Quantum Chemistry} (seventh edition),
 Pearson, Boston, 2014.


\bibitem{TGmath}
 E.~H. Lieb, R.~Seiringer, J.~P. Solovej and J.~Yngvason,
 \emph{The Mathematics of the Bose Gas and its Condensation}, Oberwolfach seminars 34,
 Birkh\"{a}user, Basel, 2005.

\bibitem{Lieb}
 E.~H. Lieb, R.~Seiringer and J.~Yngvason,
 \emph{Bosons in a trap: A rigorous derivation of the Gross-Pitaevskii
  energy functional},
 Phys. Rev. A, 61 (2000), 759--771.



\bibitem{intro_DFT}
R.~G. Parr and W.~Yang,
 {\em Density-Functional Theory of Atoms and Molecules},
 Oxford University Press, 1989.


\bibitem{intro}
A.~C. Phillips,
 {\em Introduction to Quantum Mechanics},
 Wiley, Chichester; New York; 2003.

\bibitem{Pitaevskii}
 L.~P. Pitaevskii and S.~Stringari,
 \emph{Bose-Einstein Condensation},
 Clarendon Press, Oxford, 2003.

\bibitem{Ruan}
X. Ruan, \emph{Mathematical Theory and Numerical Methods for Bose-Einstein Condensation with Higher Order Interactions}, PhD Thesis, National University of Singapore, 2017.


\bibitem{Sch}
 E.~{Schr\"{o}dinger},
 \emph{An undulatory theory of the mechanics of atoms and molecules},
 {Phys. Rev.}, 28 (1926),
  1049--1070.


\bibitem{Sei}
R. Seiringer,
\emph{Gross-Pitaevskii theory of the rotating Bose gas},
Comm. Math. Phys., 229 (2002), 491–509.


\bibitem{gap0}
 I. M.~{Singer}, B.~{Wong}, S.-T. {Yau} and S.~S.-T. {Yau},
 \emph{{An estimate of the gap of the first two eigenvalues in the
  Schr\"odinger operator}},
 {Ann. Sc. Norm. Super. Pisa Cl. Sci. (4)}, 12 (1985),
  319--333.





%
\end{thebibliography}
\end{document}